%% file: class_orth_sep_scc_Main.tex
\documentclass[a4paper,11pt,subeqn]{article} 
\usepackage{fancyhdr}
\usepackage{graphicx}
\usepackage{microtype} 
\usepackage{mathdots} 
\usepackage[shortlabels]{enumitem}

\usepackage[backend=bibtex,citestyle=alphabetic,bibstyle=alphabetic,backref=true,url=false]{biblatex}
\usepackage{amsfonts,amssymb,amstext,amsmath} 
\usepackage{cool} 
\usepackage[amsmath,thmmarks,hyperref]{ntheorem}
\usepackage{tensor}
\usepackage{tensind} 
\tensordelimiter{?}

\usepackage[bookmarksnumbered,hypertexnames=false]{hyperref} 
\hypersetup{
    plainpages=false,       
    unicode=false,          
    pdftoolbar=true,        
    pdfmenubar=true,        
    pdffitwindow=false,     
    pdfstartview={FitH},    
    pdftitle={Concircular tensors in Spaces of Constant Curvature: With Applications to Orthogonal Separation of The Hamilton-Jacobi Equation},    
    pdfauthor={Krishan Rajaratnam},     
    pdfnewwindow=true,      
    linktoc=page,           
    colorlinks=true,        
    linkcolor=blue,         
    citecolor=black,        
    filecolor=magenta,      
    urlcolor=cyan           
}

\numberwithin{equation}{section} 
\usepackage[capitalize]{cleveref} 

\usepackage{amsmath,amsfonts,amssymb}

\newcommand{\diag}{\operatorname{diag}}
\newcommand{\spa}[1]{\operatorname{span} \{#1\}}
\newcommand{\End}{\operatorname{End}}
\newcommand{\Ima}{\operatorname{Im}}

\newcommand{\scalprod}[2]{\left\langle #1,#2 \right\rangle}
\newcommand{\tr}[1]{\operatorname{tr}(#1)}

\def\sp(#1,#2){\left\langle #1,#2 \right\rangle}
\def\bp#1{\sp(#1)}

\DeclareFontFamily{U}{matha}{\hyphenchar\font45}
\DeclareFontShape{U}{matha}{m}{n}{
      <5> <6> <7> <8> <9> <10> gen * matha
      <10.95> matha10 <12> <14.4> <17.28> <20.74> <24.88> matha12
      }{}
\DeclareSymbolFont{matha}{U}{matha}{m}{n}
\DeclareFontFamily{U}{mathx}{\hyphenchar\font45}
\DeclareFontShape{U}{mathx}{m}{n}{
      <5> <6> <7> <8> <9> <10>
      <10.95> <12> <14.4> <17.28> <20.74> <24.88>
      mathx10
      }{}
\DeclareSymbolFont{mathx}{U}{mathx}{m}{n}

\DeclareMathSymbol{\obot}         {2}{matha}{"6B}
\DeclareMathSymbol{\bigobot}       {1}{mathx}{"CB}

\newcommand{\R}{\mathbb{R}}
\newcommand{\N}{\mathbb{N}}
\newcommand{\Z}{\mathbb{Z}}
\newcommand{\C}{\mathbb{C}}

\newcommand{\Fi}{\mathbb{F}}
\renewcommand{\E}{\mathbb{E}} 
\newcommand{\Si}{\mathbb{S}}

\newcommand{\eunn}{\mathbb{E}^{n}_{\nu}}

\newcommand{\punn}{\mathbb{P}^{n}_{\nu}}
\newcommand{\sgn}{\operatorname{sgn}}

\newcommand{\conj}[1]{\overline{#1}}

\def\deriv#1#2{\dfrac{d#1}{d#2}}
\def\pderiv#1#2{\dfrac{\partial #1}{\partial #2}}

\def\spderiv#1#2#3{\dfrac{\partial^2 #1}{\partial #2\partial #3}}

\newcommand{\lied}[2]{\mathcal{L}_{#2} #1} 
\def\ve{\mathfrak{X}} 

\newcommand{\Ei}{{\mathcal{E}}}
\newcommand{\F}{{\mathcal{F}}}

\usepackage{amsfonts,amssymb,amsmath}

\theoremnumbering{arabic}
\theoremstyle{break} 
\usepackage{latexsym}
\theoremsymbol{\ensuremath{_\Box}}
\theorembodyfont{\itshape}
\theoremheaderfont{\normalfont\bfseries}
\theoremseparator{}
\newtheorem{theorem}{Theorem}[section]
\newtheorem{proposition}[theorem]{Proposition}
\newtheorem{corollary}[theorem]{Corollary}
\newtheorem{lemma}[theorem]{Lemma}
\theoremsymbol{\ensuremath{\diamondsuit}} 
\newtheorem{thmMy}[theorem]{Theorem}
\newtheorem{propMy}[theorem]{Proposition}

\theoremsymbol{\ensuremath{_\Box}}

\theorembodyfont{\upshape}
\newtheorem{example}[theorem]{Example}

\newtheorem{remark}[theorem]{Remark}
\newtheorem{definition}[theorem]{Definition}

\theoremstyle{nonumberplain}
\theoremheaderfont{\scshape}
\theorembodyfont{\normalfont}
\theoremsymbol{\ensuremath{_\blacksquare}}
\usepackage{amssymb}
\newtheorem{proof}{Proof}

\qedsymbol{\ensuremath{_\blacksquare}}

\numberwithin{equation}{section} 

\newcounter{partCounter}
\newenvironment{parts}{
	\begin{list}{\bfseries{}Case \arabic{partCounter}~}{\usecounter{partCounter}}
}{
	\end{list}
}

\renewcommand{\d}{{\textrm d}}

\usepackage{autonum} 

\usepackage[sanitize={symbol=false},style=index,sort=def,acronym,toc,nomain]{glossaries} 
\newglossary[nlg]{notation}{not}{ntn}{List of Notations} 

\newglossaryentry{tenRCurv}{name=Riemann curvature tensor, description={
$	R(X,Y)Z = \nabla_{X}\nabla_{Y} - \nabla_{Y}\nabla_{X} - \nabla_{[X,Y]}$}}

\newglossaryentry{tenHess}{name=Hessian of a function f, description={
\begin{equation*}
	H^{f}(X,Y) = XY f - (\nabla_{X} Y) f
\end{equation*}
}}

\newglossaryentry{meanCVf}{name=Mean curvature vector field, description={
\begin{equation*}
	H = \frac{1}{n} \sum\limits_{i=1}^{n} \epsilon_{i} \mathbf{II}
\end{equation*}
}}

\newglossaryentry{ot}{name={orthogonal tensor},
description={A valence 2 symmetric contravariant tensor whose uniquely determined endomorphism is point-wise diagonalizable.}}

\newglossaryentry{nTorLess}{name={torsionless},
description={A $\binom{1}{1}$-tensor is called torsionless if it's Nijenhuis torsion vanishes.}}

\newglossaryentry{eunn}{type=notation,name={\ensuremath{\eunn}}, description={pseudo-Euclidean space, an $n$-dimensional vector space equipped with a metric with signature $\nu$}} 

\newglossaryentry{eunnk}{type=notation,name={\ensuremath{\eunn(\kappa)}}, description={ A hyperquadric of pseudo-Euclidean space. More precisely the central hyperquadric of \gls{eunn} with curvature $\kappa$.}}

\newglossaryentry{obot}{type=notation,
name={\ensuremath{\obot}},
description={The orthogonal direct sum.
}}

\newglossaryentry{sgn}{type=notation,
name={\ensuremath{\sgn}},
description={Given a real number $a$, $\sgn a$ is the sign of $a$ if $a \neq 0$ and $0$ if $a = 0$.
}}

\newglossaryentry{fm}{type=notation,
name={\ensuremath{\F(M)}},
description={The set of functions defined on the manifold $M$.
}}

\newglossaryentry{vem}{type=notation,
name={\ensuremath{\ve(M)}},
description={The set of vector fields defined on the manifold $M$.
}}

\newglossaryentry{secd}{type=notation,
name={\ensuremath{\Gamma(E)}},
description={The set of vector fields tangent to the distribution $E$.
}}

\newglossaryentry{spm}{type=notation,
name={\ensuremath{S^{p}(M)}},
description={The set of symmetric contravariant tensors of valence $p$ defined on the manifold $M$.
}}

\newglossaryentry{con}{type=notation,
name={\ensuremath{\operatorname{C}^{p}(M)}},
description={The vector space of concircular contravariant tensors of valence $p$ defined on the manifold $M$.
}}

\newglossaryentry{cono}{type=notation,
name={\ensuremath{\operatorname{C}^p_0(M)}},
description={The vector space of covariantly constant contravariant tensors of valence $p$ defined on the manifold $M$.
}}

\newglossaryentry{odot}{type=notation,
name={\ensuremath{\odot}},
description={The symmetric product of two tensors, i.e. if $u,v$ are tensors then $u \odot v$ is the symmetrization of $u \otimes v$.},
sort=od}

\newglossaryentry{end}{type=notation, name={\ensuremath{\End}},
description={We refer to a linear map from a vector space (vector bundle) into itself as an \emph{endomorphism}.}}

\newacronym{scc}{SCC}{space of constant curvature}

\newacronym{kv}{KV}{Killing vector}
\newacronym{ckv}{CKV}{conformal Killing vector}
\newacronym{kt}{KT}{Killing tensor}
\newacronym{chkt}{ChKT}{characteristic Killing tensor}
\newacronym{ckt}{CKT}{conformal Killing tensor}
\newacronym{cv}{CV}{concircular vector}
\newacronym{ct}{CT}{concircular tensor also called a C-tensor}
\newacronym{oct}{OCT}{orthogonal concircular tensor}
\newacronym{ict}{ICT}{irreducible concircular tensor}

\newacronym{kss}{KS-space}{Killing-Stackel space}
\newacronym{kem}{KEM}{Kalnins-Eisenhart-Miller}
\newacronym{bekm}{BEKM}{Benenti-Eisenhart-Kalnins-Miller}
\newacronym{kbd}{KBD}{Killing Bertrand-Darboux}
\newacronym{kbdt}{KBDT}{Killing Bertrand-Darboux tensor}

\newacronym{wpnet}{WP-net}{warped product net}
\newacronym{tpnet}{TP-net}{twisted product net}

\newacronym{hj}{HJ}{Hamilton-Jacobi}
\makeglossary
\newcommand{\defn}[1]{\emph{\gls*{#1}}} 

\begin{document}
\pagenumbering{roman}
\title{Concircular tensors in Spaces of Constant Curvature: With Applications to Orthogonal Separation of The Hamilton-Jacobi Equation}
\author{Krishan Rajaratnam\footnote{e-mail: k2rajara@uwaterloo.ca}, Raymond G. McLenaghan\footnote{e-mail: rgmclenaghan@uwaterloo.ca}\\Department of Applied Mathematics, University of Waterloo, Canada}
\date{\today} 
\maketitle
\begin{abstract}\centering
	We study concircular tensors in spaces of constant curvature and then apply the results obtained to the problem of the orthogonal separation of the Hamilton-Jacobi equation on these spaces. Any coordinates which separate the geodesic Hamilton-Jacobi equation are called separable. Specifically for spaces of constant curvature, we obtain canonical forms of concircular tensors modulo the action of the isometry group, we obtain the separable coordinates induced by irreducible concircular tensors, and we obtain warped products adapted to reducible concircular tensors. Using these results, we show how to enumerate the isometrically inequivalent orthogonal separable coordinates, construct the transformation from separable to Cartesian coordinates, and execute the Benenti-Eisenhart-Kalnins-Miller (BEKM) separation algorithm for separating natural Hamilton-Jacobi equations.
\end{abstract}

%
%
%

\pagenumbering{arabic}
\newpage
\input{class_orth_sep_scc_Body}
\appendix
\include{app1}
\phantomsection
\addcontentsline{toc}{section}{References}
\printbibliography
\end{document}

%% file: class_orth_sep_scc_Body.tex
\section{Introduction}

It is shown in \cite{Rajaratnam2014a} that any point-wise diagonalizable concircular tensor, hereafter called a \defn{oct}, can be used to recursively construct separable coordinates for the (geodesic) Hamilton-Jacobi equation. Such coordinates are called Kalnins-Eisenhart-Miller (KEM) coordinates. In \cite{Rajaratnam2014d} it is shown that all orthogonal separable coordinates for the Hamilton-Jacobi equation in spaces of constant curvature are KEM coordinates. The work done in \cite{Rajaratnam2014d} serves as an independent verification of the Kalnins-Miller classification of separable coordinates for Riemannian spaces of constant curvature \cite{Kalnins1986}. Hence the classification of OCTs in spaces of constant curvature is crucial for classifying orthogonal separable coordinates in these spaces.

Specifically, OCTs have the following uses:

\begin{enumerate}
	\item An algebraic classification of these tensors modulo the action of the isometry group can be used to obtain a notion of inequivalence for KEM coordinate systems.
	\item Crampin \cite{Crampin2003} shows that one can obtain transformations to separable coordinates for OCTs with functionally independent eigenfunctions. It's evident from the results in \cite{Rajaratnam2014a,Rajaratnam2014d} that a knowledge of the warped product decompositions of the space is sufficient to construct transformations to separable coordinates for any KEM coordinate system. We will expand on this idea later.
	\item When concircular tensors have simple eigenfunctions, it is shown in \cite{Benenti2005a} (see also \cite{Benenti1992c,Benenti1993,Benenti2004}) that a basis for the Killing-Stackel space can be obtained. Using the theory presented in \cite{Rajaratnam2014a} one can generalize this result to arbitrary KEM coordinate systems.
	\item With a classification of concircular tensors, the BEKM separation algorithm (presented in \cite{Rajaratnam2014a}), can be executed to solve the separation of variables problem for natural Hamiltonians.
\end{enumerate}

Thus an unsolved problem is to obtain a complete classification of these tensors in spaces of constant curvature. A partial classification of these tensors in Euclidean space can be found in \cite{Lundmark2003} (cf. \cite{Benenti2005a}). A complete classification of these tensors for Euclidean space and the Euclidean sphere is implicit in \cite{Waksjo2003}.

Building on existing knowledge in \cite{Lundmark2003,Crampin2003} together with new insights \cite{Rajaratnam2014a}, in this article we obtain a complete (local) classification of orthogonal concircular tensors in all spaces of constant curvature with Euclidean and Lorentzian signature\footnote{The classification for other signatures can be obtained fairly easily if one wishes.}. More details on our classification and the way in which it is done is given in \cref{sec:sumRes}, after we have introduced some preliminaries. Some of our results are also summarized in \cref{sec:sumRes}.

Different parts of this problem have been solved for special cases by different researchers over the past few decades. A classification of separable coordinate systems in Riemannian spaces of constant curvature was originally done by Kalnins and Miller in \cite{Kalnins1986a,Kalnins1982}, see also \cite{Kalnins1986} which is a book containing their results. The insight provided by their classification was crucial for the development of the theory which we present here. They have extended this work to spaces of constant curvature with arbitrary signature in \cite{Kalnins1984} to obtain a partial classification. In \cite{Kalnins1975} orthogonal separable coordinates in two dimensional Minkowski space are determined and partial results in three dimensional Minkowski space are given. A more detailed classification in two dimensions is given in \cite{Mclenaghan2002}, and in three dimensions in \cite{Kalnins1976}. This classification in three dimensions is further refined in \cite{Hinterleitner1998} and \cite{Horwood2008a}. A classification of orthogonal separable coordinates for four dimensional Minkowski space has been given in \cite{Kalnins1978} and references therein. Classifications of isometrically inequivalent Killing tensors in two dimensional flat spaces are given in \cite{Mclenaghan2002a}, \cite{McLenaghan2004b} and \cite{Chanu2006}, that in three dimensional Minkowski space in \cite{Horwood2009}, and that on the Euclidean three sphere in \cite{Cochran2011}. Finally, building on results in \cite{Kalnins1986}, a version of the BEKM separation algorithm is given in \cite{Waksjo2003} for Euclidean space and the Euclidean sphere.

Our approach to this problem has several advantages over previous approaches. First we are able to give a unified theory applicable to spaces of constant curvature with both Euclidean and Lorentzian signatures. This approach allows one to solve the different but related problems listed above. We are able to give a precise notion of inequivalence for orthogonal separable coordinate systems in Minkowski space and thereby give a clear, rigorous and complete classification in this space.


\section{Preliminaries and Summary}
\subsection{Notations and Conventions}

All differentiable structures are assumed to be smooth (class $C^{\infty}$). Let $M$ be a pseudo-Riemannian manifold of dimension $n$ equipped with covariant metric $g$. Unless specified otherwise, it is assumed that $n \geq 2$. The contravariant metric is usually denoted by $G$ and $\bp{\cdot, \cdot}$ plays the role of the covariant and contravariant metric depending on the arguments. We denote $\gls*{spm}$ as the set of symmetric contravariant tensor fields of valence p on M. Furthermore $\F(M) = S^{0}(M)$ is the set of functions from M to $\R$ and $\ve(M) = S^1(M)$ denotes the set of vector fields over $M$. If $f \in \F(M)$ then $\nabla f \in \ve(M)$ denotes the gradient of $f$, i.e. the vector field metrically equivalent to $\d f$. Also if $x \in \ve(M)$ then we denote $x^2 := \bp{x,x}$.

Throughout this article we will be working in pseudo-Euclidean space, which is defined as follows. An $n$-dimensional vector space $V$ equipped with metric $g$ of signature\footnote{The signature is equal to the number of negative diagonal entries in a basis which diagonalizes $g$.} $\nu$ is denoted by $\eunn$ and called \emph{pseudo-Euclidean space}. We obtain Euclidean space $\E^n$ in the special case where $\nu = 0$. Also Minkowski space $M^n$ is obtained by taking $\nu = 1$. Also note that since $\eunn$ is a vector space, for any $p \in \eunn$ we identify vectors in $T_p \eunn$ with points in $\eunn$.

Given an open subset $U \subseteq \eunn$ and $\kappa \in \R \setminus \{0\}$, we denote by $U(\kappa)$ the \emph{central hyperquadric} of $\eunn$ contained in $U$, which is defined by:

\begin{equation}
U(\kappa) = \{p \in U \; | \; \bp{p,p} = \kappa^{-1} \}
\end{equation}

Usually $U = \eunn$ and this is denoted $\eunn(\kappa)$. It is well known that $\eunn(\kappa)$ is a pseudo-Riemannian manifold of dimension $n-1$ with signature $\nu + \frac{(\operatorname{sgn} \kappa -1)}{2}$ and constant curvature $\kappa$ (see \cite{barrett1983semi} or \cite[Appendix~D]{Rajaratnam2014}). We often refer to these manifolds as the spherical submanifolds of pseudo-Euclidean space (see \cite[Appendix~D]{Rajaratnam2014} for the definition of a spherical submanifold). We use the connected components of these manifolds as the standard models of the corresponding space of constant curvature. Since $\eunn(\kappa) \subset \eunn$, for any $p \in \eunn(\kappa)$ we identify vectors in $T_p \eunn(\kappa)$ with points in $\eunn$.

For the following discussion, suppose $V$ is a pseudo-Euclidean vector space. Without further specification, \emph{tensor} is short for a valence 2-tensor and the type depends on the context. Let $T$ be an endomorphism of $V$. A subspace $D$ is called \emph{$T$-invariant} if $T D \subseteq D$. $T$ is said to have a \emph{simple eigenvalue} $\lambda$, if $\lambda$ is real and has algebraic multiplicity equal to 1. $T$ is said to have \emph{simple eigenvalues} if all its eigenvalues are simple. $T$ is called self-adjoint if

\begin{equation}
\bp{T x , y} = \bp{x , T y} \quad \text{ for all } x,y \in V
\end{equation}

The above condition is equivalent to requiring $T$ to be metrically equivalent to a symmetric contravariant tensor. By an \emph{orthogonal tensor}, we mean a symmetric contravariant tensor whose uniquely determined endomorphism is diagonalizable with real eigenvalues. One can check that the eigenspaces of such an endomorphism are necessarily pair-wise orthogonal non-degenerate subspaces. Finally given a subspace $W \leq V$, the restriction of $T$ to $W$ is denoted $T|_W$.

All the above notions generalize point-wise to a pseudo-Riemannian manifold. Although only locally. For example given a self-adjoint $\binom{1}{1}$-tensor $T$ on $M$, we say it is an \emph{orthogonal tensor} if it is point-wise diagonalizable on some (non-empty) open subset of $M$ and we tacitly work on this subset. Similarly we say $T$ is not an orthogonal tensor on $M$ if $T$ is not point-wise diagonalizable on a open dense subset of $M$. Similar definitions apply to other notions such as constancy of functions on $M$.

\subsection{Self-adjoint operators in pseudo-Euclidean space}

In this section we review the metric-Jordan canonical form of a self-adjoint operator on a pseudo-Euclidean space. The details of the theory behind this canonical form is given in \cite[Appendix~C]{Rajaratnam2014}; these are solutions to exercises 18-19 in \cite[P.~260-261]{barrett1983semi}.

A \emph{Jordan block} of dimension $k$ with eigenvalue $\lambda \in \C$ is a $k \times k$ matrix denoted by $J_k(\lambda)$, and defined as:

\begin{equation}
J_k(\lambda) :=
\begin{pmatrix}
\lambda & 1 &  &  &  \\ 
& \lambda & \ddots &  & 0 \\ 
&  & \ddots & 1 &  \\ 
&  &  & \lambda & 1 \\ 
& 0 &  &  & \lambda
\end{pmatrix} 
\end{equation}

The \emph{skew-diagonal matrix} of dimension $k$ is denoted by $S_k$, and defined as:

\begin{equation}
S_k := \begin{pmatrix}
0 &  & 1 \\ 
& \iddots &  \\ 
1 &  & 0
\end{pmatrix}
\end{equation}

An ordered sequence of vectors $\beta = \{v_1,\dotsc,v_k\}$ where the matrix representation of $g$ with respect to (w.r.t) $\beta$ has the form $g|_\beta = \varepsilon S_k$, is called a \emph{skew-normal sequence} of (length k) and (sign $\varepsilon = \pm 1$). The subspace spanned by a skew-normal sequence is necessarily non-degenerate and of dimension $k$ (see \cite[lemma~8.1.1]{Rajaratnam2014}).

In order to express the metric-Jordan canonical form of a self-adjoint operator on a pseudo-Euclidean space \cite[Appendix~C]{Rajaratnam2014}, we use the signed integer $\varepsilon k \in \Z$ where $k \in \N$ and $\varepsilon = \pm 1$. Then the notation $J_{\varepsilon k}(\lambda)$ is short hand for the pair:

\begin{align}
A & = J_k(\lambda) & g = & \varepsilon S_k 
\end{align}

Furthermore, given matrices $A_1$ and $A_2$, we denote the following block diagonal matrix by $A_1 \oplus A_2$

\begin{equation}
A_1 \oplus A_2 := \begin{pmatrix}
A_1 &  0 \\ 
0 & A_2
\end{pmatrix}
\end{equation}

The (real) metric-Jordan canonical form of a self-adjoint operator is discussed in detail in \cite[Appendix~C]{Rajaratnam2014}. In this article (for convenience) we will be working with the complex version (it can be deduced from \cite[theorem~C.3.7]{Rajaratnam2014}), which is given as follows:

\begin{theorem}[Complex metric-Jordan canonical form \cite{barrett1983semi}] \label{thm:comMetJFor}
	A real operator $T$ on a pseudo-Euclidean space $\eunn$ is self-adjoint iff there exists a (possibly complex) basis $\beta$ such that
	
	\begin{equation}
	T|_{\beta} = J_{\varepsilon_1 k_1}(\lambda_1) \oplus \cdots \oplus J_{\varepsilon_l k_l}(\lambda_l)
	\end{equation}
	
	Furthermore there exists a canonical basis such that the unordered list \\ $\{J_{\varepsilon_1 k_1}(\lambda_1), \dotsc, J_{\varepsilon_l k_l}(\lambda_l) \}$ is uniquely determined by $T$ and an invariant of $T$ under the action of the orthogonal group $O(\eunn)$.
\end{theorem}
\begin{remark}
	Since $T$ is real, each Jordan block $J_{\varepsilon k}(\lambda)$ with $\lambda \in \C \setminus \R$ comes with a complex conjugate pair $J_{\varepsilon k}(\conj{\lambda})$. For complex eigenvalues, we can additionally assume that $\varepsilon = 1$.
\end{remark}

A key fact used to derive the above canonical form and one to keep in mind is that for any self-adjoint operator $T$, any non-degenerate $T$-invariant subspace has a $T$-invariant orthogonal complement.

\subsection{Concircular tensors}


$L \in S^{p}(M)$ is called a \defn{ct} of valence $p$ if there exists $C \in S^{p-1}(M)$ (called the \emph{conformal factor}) such that

\begin{equation} \label{eq:defnCTp}
\nabla_{x}L = C \odot x
\end{equation}

\noindent for all $x \in \ve(M)$.  Concircular tensors of arbitrary valence were originally defined in \cite{Crampin2008}, where they were called special conformal Killing tensors. This is because concircular tensors are conformal Killing tensors \cite{Crampin2008}. When $p = 1$, $L$ is called a \defn{cv}. When $p = 2$, we will simply call $L$ a concircular tensor since we will mainly be working with these objects. Furthermore one should note that the CTs form a real vector space and the symmetric product of CTs is again a CT. Sometimes we denote the space of concircular tensors of valence $p$ by $\operatorname{C}^{p}(M)$ and the subspace of covariantly constant tensors by $\operatorname{C}^p_0(M)$.

An \defn{oct} (also called an OC-tensor) is a concircular tensor which is also an orthogonal tensor. OC-tensors with simple eigenfunctions were studied extensively by Benenti, see \cite{Benenti1992c,Benenti2004,Benenti2005a}; thus in recognition of his contributions we refer to this special class of OC-tensors as \emph{Benenti tensors} (also called L-tensors by Benenti).

OC-tensors have some useful properties. First, given a tensor $L$, let $N_{L}$ be the Nijenhuis tensor (torsion) of $L$ \cite{Gerdjikov2008a}. We say that $L$ is \emph{torsionless} if its Nijenhuis tensor vanishes. Then if $L$ is a concircular tensor, the following equations hold \cite[Lemma 3.1]{Benenti2005a} (cf. \cite{Crampin2003})

\begin{align}
[L, G] & = - 2 \nabla \tr{L} \odot G \quad   ([L, G]_{a b c} = -2\nabla_{(a} L_{b c)}) \\
N_{L} & = 0 
\end{align}

Conversely, by Theorem 19.3 in \cite{Benenti2005a}, an orthogonal tensor satisfying the above equations is a C-tensor. The first of the above equations tells us that a C-tensor is a conformal Killing tensor of trace-type. The second equation can be interpreted if we assume $L$ is an OC-tensor.

Suppose now that $L$ is an OC-tensor with eigenspaces $(E_{i})_{i=1}^{k}$ and corresponding eigenfunctions $\lambda^{1},...,\lambda^{k}$. Since an OC-tensor has Nijenhuis torsion zero, by Theorem~13.29 (Haantjes theorem) in \cite{Gerdjikov2008a}, the eigenspaces $(E_{i})_{i=1}^{k}$  are orthogonally integrable and each eigenfunction $\lambda^{i}$ depends only on $E_{i}$. Furthermore the trace-type condition implies that the eigenfunction corresponding to a multidimensional eigenspace of $L$ is a constant \cite{Rajaratnam2014a}.

Suppose $D$ is a multidimensional eigenspace of a non-trivial\footnote{By a non-trivial concircular tensor, we mean one which is not a multiple of the metric when $n > 1$.} OCT $L$. Denote by $D^{\perp}$ the distribution orthogonal to $D$. Then one can show that (see \cite[Theorem~6.1]{Rajaratnam2014a} for example):

\begin{itemize}
	\item There is a local product manifold $B \times F$ of Riemannian manifolds $(B, g_{B})$ and $(F, g_{F})$ such that: \\
	$\{p\} \times F$ is an integral manifold of $D$ for any $p \in B$ and \\
	$B \times \{q\}$ is an integral manifold of $D^{\perp}$ for any $q \in F$. 
	\item $B \times F$ equipped with the metric $\pi_{B}^{*} g_{B} + \rho^{2} \pi_{F}^{*} g_{F}$ for a specific function $\rho : B \rightarrow \R^{+}$ is locally isometric to $(M,g)$; where $\pi_{B}$ (resp. $\pi_{F}$) is the canonical projection onto $B$ (resp. $F$).
\end{itemize}

Such a product manifold is called a \emph{warped product} and is denoted $B \times_{\rho} F$. We also say in this case that the warped product $B \times_{\rho} F$ is \emph{adapted} to the splitting $(D^{\perp},D)$. The manifold $F$ is a \emph{spherical submanifold} and $B$ is \emph{geodesic submanifold} of $M$ (see \cite[Appendix~D]{Rajaratnam2014} and references therein). An important observation is that $L$ restricted to $B$ is an OCT; we will use this later to construct OCTs from Benenti tensors.

In general if $L$ has multiple multidimensional eigenspaces, we will have to consider more general warped products. So suppose $M = \prod_{i=0}^{k} M_{i}$ is a product manifold of pseudo-Riemannian manifolds $(M_{i}, g_{i})$ where $\dim M_{i} > 0$ for $i > 0$. Equip $M$ with the metric $g = \sum_{i=0}^{k} \rho_{i}^{2} \pi_{i}^{*} g_{i}$ where $\rho_{i} : M_0 \rightarrow \R^{+}$ are functions with $\rho_0 \equiv 1$ and $\pi_{i} : M \rightarrow M_{i}$ are the canonical projection maps. Additionally we assume either $\dim M_{0} > 0$ or $k > 1$. Then $(M,g)$ is called a \emph{warped product} and the metric $g$ is called a \emph{warped product metric}. If $\dim M_{0} = 0$ then $(M,g)$ is called a \emph{pseudo-Riemannian product}. The warped product is denoted by $M_{0} \times_{\rho_{1}} M_{1} \times \cdots \times_{\rho_{k}} M_{k} $. $M_{0}$ is called the geodesic factor of the warped product and the $M_{i}$ for $i > 0$ are called spherical factors. See \cite{Rajaratnam2014} and references therein for more on warped products.

The following class of OCTs are fundamental to the classification:

\begin{definition}[Irreducible concircular tensors]
	An OC-tensor with functionally independent eigenfunctions is referred to as an \defn{ict} or more succinctly an \emph{IC-tensor}. To be precise, an IC-tensor has real eigenfunctions $u^{1},...,u^{k}$ (counted without multiplicity) satisfying:
	\begin{equation}
	\d u^{1} \wedge \cdots \wedge \d u^{k} \neq 0
	\end{equation}
	
	Furthermore an OC-tensor which is not irreducible is called \emph{reducible}.
\end{definition}

\begin{remark}
	IC-tensors were the class of C-tensors mainly studied in \cite{Crampin2003}.
\end{remark}

Since we observed earlier that the eigenfunction associated with a multidimensional eigenspace of an OCT is constant, it follows that an ICT must have simple eigenfunctions, hence ICTs are Benenti tensors. The special property that ICTs have is that their eigenfunctions can be used as (local) coordinates for the separable web they induce \cite{Crampin2003}. We will refer to these coordinates as the \emph{canonical coordinates} induced by these tensors.

Away from singular points, locally, we can assume a reducible OC-tensor has eigenfunctions $u^{1},...,u^{k}$ which are functionally independent and the rest of which are constants. Indeed, for the remainder of this article, this is what we will mean by a reducible OC-tensor. More generally we say a CT is \emph{reducible} if it admits a non-degenerate eigenspace with constant eigenfunction. We will outline in \cref{sec:sumRes} how we will break down the classification in terms of irreducible and reducible OCTs.

\subsubsection{Properties of OCTs}

We will now list some properties of OCTs that will be used later. The following proposition gives a necessary and sufficient (n.s.s) condition to determine when two OCTs (one of which is not covariantly constant) share the same eigenspaces.

\begin{propMy} \label{prop:CtEquiv}
	Suppose $M$ is a connected manifold and $L$ is an OCT on $M$ which is not covariantly constant (around any neighborhood). Then $\tilde{L}$ is a CT sharing the same eigenspaces as $L$ iff there exists $a \in \R \setminus \{0\}$ and $b \in \R$ such that
	
	\begin{equation}
	\tilde{L}= a L + b G
	\end{equation}
\end{propMy}
\begin{proof}
	The proof of this, which is a straightforward calculation, will appear else where.
\end{proof}

The above proposition no longer holds if we relax the assumption that $L$ is not covariantly constant. One can easily see why by considering any non-trivial covariantly constant symmetric tensor in Euclidean space. We now define an important notion for classifying KEM webs.

\begin{definition}[Geometric Equivalence of CTs]
	We say two CTs $L$ and $\tilde{L}$ are \emph{geometrically equivalent} if there exists $a \in \R \setminus \{0\}$, $b \in \R$ and $T \in I(M)$ such that
	\begin{equation}
	\tilde{L}= a T_*L + b G
	\end{equation}
\end{definition}

An immediate corollary of the above proposition is the following:

\begin{corollary}[Geometric Equivalence of OCTs]
	Suppose $M$ is a connected manifold. Suppose $L$ and $\tilde{L}$ are OCTs with respective eigenspaces $\Ei = (E_1,\dotsc,E_k)$ and $\tilde{\Ei} = (\tilde{E}_1,\dotsc,\tilde{E}_k)$. Suppose further that $\Ei$ is not a Riemannian product net \cite{Rajaratnam2014a}, equivalently one of the CTs is not covariantly constant. Then $\Ei$ and $\tilde{\Ei}$ are related by $T \in I(M)$, i.e. $\tilde{E}_i = T_* E_{\sigma(i)}$ for each $i$ (where $\sigma$ is a permutation of $\{ 1,\dotsc,k \}$) iff $L$ and $\tilde{L}$ are geometrically equivalent.
\end{corollary}

The above corollary implies that the classification of isometrically inequivalent KEM webs can be reduced to the classification of geometrically inequivalent OCTs. For the proof of the following theorem, see \cite{Thompson2005,Crampin2007}.

\begin{theorem}[The Vector Space of Concircular tensors  \cite{Thompson2005}] \label{thm:CTsDim}
	If $n > 1$, then the C-tensors of valence $r \leq 2$ form a finite dimensional real vector space with maximal dimension equal to the dimension of the space of constant symmetric $r$-tensors in $\R^{n+1}$. Furthermore the maximal dimension is achieved if and only if the space has constant curvature.
\end{theorem}

The above theorem implies the following:
\begin{corollary}[Concircular tensors in spaces of constant curvature] \label{cor:CTpSCC}
	Suppose $M^n$ is a space of constant curvature with $n > 1$ and let $r \leq 2$. Let $\beta = \{v_1,\dotsc,v_{n+1}\}$ be a basis for the space of concircular vectors, then a given C-tensor of valence $r$ can be written uniquely as a linear combination of $r$-fold symmetric products of the vectors in $\beta$.
\end{corollary}

\subsection{Summary of Results} \label{sec:sumRes}

We first give an overview of the classification. The classification breaks down into three parts: obtaining canonical forms for C-tensors modulo the action of the isometry group (\cref{sec:CtEunn,sec:CtEunnKap}), classifying the webs described by IC-tensors (\cref{sec:cTIrred}) and obtaining warped product decompositions adapted to reducible OCTs (\cref{sec:CtClassRed}).

The webs formed by IC-tensors are the basic building blocks of all separable webs. \Cref{sec:cTIrred} is devoted to obtaining information about these webs from the corresponding IC-tensors. In that section we obtain the transformation from the canonical coordinates $(u^i)$ induced by these tensors to Cartesian coordinates $(x^i)$ and we obtain the metric in canonical coordinates. This is done by first calculating the characteristic polynomial of all CTs in spaces of constant curvature in a Cartesian coordinate system. In examples, we will also show how to obtain the coordinate domains for coordinate systems induced by IC-tensors.

To obtain all orthogonal separable coordinates in spaces of constant curvature, we also have to consider reducible OCTs.  Let $L$ be a non-trivial reducible OCT and suppose $\psi : N_{0} \times_{\rho_{1}} N_{1} \times \cdots \times_{\rho_{k}} N_{k} \rightarrow M$ is a local warped product decomposition of $M$ adapted to the eigenspaces of $L$ such that $L_0 := L|_{N_0}$ is an ICT\footnote{If $L$ has only constant eigenfunctions, we can choose $N_0$ to be a point.}. Let $(x_0) = (u^1,\dotsc,u^{n_0})$ be the canonical coordinates induced by $L_0$ on some open subset of $N_0$. For $i > 0$ suppose $(x_i) = (x_i^1,\dotsc,x_i^{n_i})$ are separable coordinates for $N_i$, then it was shown in \cite[proposition~6.8]{Rajaratnam2014a} that the coordinates $\psi(x_{0},x_{1},\dotsc,x_{k})$ are separable coordinates for $M$. To construct the separable coordinates $(x_i)$ on $N_i$ where $i > 0$, one would apply this procedure again on $N_i$ equipped with the induced metric, which is again a space of constant curvature \cite[lemma~6.10]{Rajaratnam2014a}. It was shown in \cite[theorem~1.3]{Rajaratnam2014d} that all orthogonal separable coordinates for spaces of constant curvature arise this way. Hence a remaining problem is to develop a method to construct warped product decompositions which decompose a given reducible OCT as above; this is done in \cref{sec:CtClassRed}. Together with the results of \cref{sec:cTIrred}, this gives a recursive procedure to construct the orthogonal separable coordinates of these spaces.

In \cref{sec:appNEx} we will show how to apply the theory developed in this article to solve motivating problems. First, in \cref{sec:enumIneqCoord} we will show how to enumerate the isometrically inequivalent separable coordinates in a given space of constant curvature. Then in \cref{sec:constSepCoord} we will show how to construct separable coordinate systems by way of examples. Finally, in \cref{sec:BEKMsep} we will show how to explicitly execute the BEKM separation algorithm in general. We also give the details of executing the BEKM separation algorithm for the Calogero-Moser system.

The classification generally breaks down into one for pseudo-Euclidean space $\eunn$ then one for its spherical submanifolds $\eunn(\kappa)$ (which usually reduces to a similar problem in $\eunn$). We give more details in the following subsections.

\subsubsection{pseudo-Euclidean space}

First we define the \emph{dilatational vector field}, $r$, to be the vector field given in Cartesian coordinates $(x^i)$ by $r = \sum\limits_{i} x^i \partial_i$. The general concircular contravariant tensor in $\eunn$ is given as follows (see \cref{prop:CtFormEunn}):

\begin{equation} \label{eq:CTGenEunn}
L = A + 2 w \odot r + m r \odot r
\end{equation}

\noindent where $A \in C^2_0(\eunn)$, $w \in C^1_0(\eunn)$ and $m \in C^0_0(\eunn)$. For $k \geq 0$, define constants $\omega_k$ as follows:
\begin{equation} \label{eq:omegI}
\omega_k = 
\begin{cases}
m & \text{ if } k = 0 \\
\bp{w,A^{k-1} w} & \text{ else }
\end{cases}
\end{equation}

The above constants aren't necessarily invariant under isometries. But invariants can be defined from them.

\begin{definition} \label{def:CtEunnInd}
	Suppose $L$ is a CT in $\eunn$ as defined above. Then we define the \emph{index} of $L$ to be the first integer $k \geq 0$ for which $\omega_k \neq 0$; $L$ is said to be \emph{non-degenerate} if such an integer exists. Furthermore if $L$ is non-degenerate, it has an associated sign (characteristic):
	
	\begin{equation}
	\varepsilon = 
	\begin{cases}
	1  & 1 \text{ if $k$ is even} \\
	\sgn \omega_k & \text{ if $k$ is odd}
	\end{cases}
	\end{equation}
\end{definition}

The following theorem which is proven in \cref{sec:CtEunn} summarizes our results on the canonical forms of concircular tensors; it classifies C-tensors into five disjoint classes.
\begin{theorem}[Canonical forms for CTs in $\eunn$] \label{thm:conTenCanForm}
	Let $\tilde{L} = \tilde{A} + m r \otimes r^{\flat} + w \otimes r^{\flat} + r \otimes w^{\flat}$ be a CT in $\eunn$. Let $k$ be the index and $\varepsilon$ be the sign of $\tilde{L}$ if $\tilde{L}$ is non-degenerate. These quantities are geometric invariants of $\tilde{L}$. Furthermore, after a possible change of origin and after changing to a geometrically equivalent CT, $L = a \tilde{L}$ for some $a \in \R \setminus \{0\}$, $\tilde{L}$ admits precisely one of the following canonical forms. 
	
	\begin{description}
		\item[Central:] If $k=0$
		\begin{equation}
		L = A + r \otimes r^{\flat}
		\end{equation}
		\item[non-null Axial:] If $k=1$, i.e. $m = 0$, and $\bp{w,w} \neq 0$:
		
		There exists a vector $e_1 \in \spa{w}$ such that $L$ has the following form:
		\begin{align}
		L & = A + e_{1} \otimes r^{\flat}  + r \otimes e_{1}^{\flat} & A e_{1} = 0, \quad \bp{e_{1},e_{1}} = \varepsilon
		\end{align}
		\item[null Axial:] If $k \geq 2$, hence $m = 0$ and $\bp{w,w} = 0$:
		
		There exists a skew-normal sequence $\beta = \{e_{1},...,e_{k}\}$ with $\bp{e_{1},e_{k}} = \varepsilon$ where $e_1 \in \spa{w}$ which is $A$-invariant such that $L$ has the following form:
		
		\begin{align}
		L & = A + e_1 \otimes r^{\flat}  + r \otimes e_1^{\flat} \\
		A|_{\beta} & = J_k(0)^T =
		\begin{pmatrix}
		0 &  &  &  &  \\ 
		1 & 0 &  &  &  \\ 
		& 1 & \ddots &  &  \\ 
		&  & \ddots & 0 &  \\ 
		&  &  & 1 & 0
		\end{pmatrix}
		\end{align}
		\item[Cartesian:] If $k$ doesn't exist, $m = 0$ and $w = 0$
		\begin{equation}
		L = \tilde{A}
		\end{equation}
		\item[degenerate null Axial:] If $k$ doesn't exist and $w \neq 0$
	\end{description}
\end{theorem}
\begin{remark}
	The degenerate null axial concircular tensors will be of no concern to us. In Euclidean space they don't occur and it will be proven later (see \cref{sec:degenCase}) that in Minkowski space that they are never orthogonal concircular tensors.
\end{remark}
\begin{remark}
	The precise classification for Euclidean and Minkowski space can be directly inferred from the above theorem by imposing the signature of the metric. The classification for Euclidean space is clear. In Minkowski space, $k \leq 3$ and when $k = 3$ the sign of the axial CT must be positive (see \cite[lemma~8.1.1]{Rajaratnam2014}).
\end{remark}
\begin{remark}
	When $k = 0$ and $1$ respectively, the translation vector $v$ for the isometry $T : r \rightarrow r + v$ which sends $\tilde{L}$ to canonical form is given as follows:
	
	
	\begin{align}
	v & = \frac{w}{\omega_0} & \text{if } k & = 0 \label{eq:transConFormI} \\
	v & = \frac{1}{\omega_1}(A w -\frac{1}{2}\frac{\omega_2}{\omega_1} w) & \text{if } k & = 1 \label{eq:transConFormII}
	\end{align}
	
	For the general case, see \cref{par:conFormTrans}.
\end{remark}

One can easily deduce that in Euclidean or Minkowski space, any covariantly non-constant OCT is non-degenerate. Hence non-degenerate CTs are the main interest of this article.

Some notation will be useful. The matrix $A$ will be called the \emph{parameter matrix} and the vector $w$ the \emph{axial vector} of the CT. When $k \geq 1$ in the above theorem, we will refer to the CT as an \emph{axial concircular tensor}.

Suppose $L$ is a non-degenerate CT in the canonical form given by \cref{thm:conTenCanForm}. We denote by $D$ the $A$-invariant subspace spanned by $w, Aw, \dotsc$. This subspace is either zero (if $w = 0$) or metrically non-degenerate. We will let $A_c := A|_{D^\perp}$, $A_d := A|_{D}$ and the central CT in $D^\perp$ with parameter matrix $A_c$ by $L_c$. Furthermore we define the following functions:

\begin{align}
p(z) & := \det (z I - L) \\
B(z) & := \det (z I - A_c) 
\end{align}

\noindent where the second determinant is evaluated in $D^\perp$.

The canonical forms for non-degenerate CTs can be enumerated by choosing a non-degenerate CT from \cref{thm:conTenCanForm} then choosing a metric-Jordan canonical form for the pair $(A|_{D^\perp}, g|_{D^\perp})$. The proofs of these canonical forms, which are given in \cref{sec:CtEunn}, can be omitted on first reading. Once these canonical forms are obtained, in \cref{sec:cTIrredCC,sec:cTIrredAC} we will calculate the characteristic polynomial for non-degenerate CTs in $\eunn$. Using this, for ICTs we can calculate the transformation from their canonical coordinates to Cartesian coordinates and the metric in canonical coordinates. Then in \cref{sec:CtClassRedEunn} we will show how to obtain the warped product decompositions induced by reducible OCTs.

\subsubsection{Spherical submanifolds of pseudo-Euclidean space}

In this section we assume $n \geq 3$. Denote the orthogonal projection $R$ onto the spherical distribution $r^\perp$ as follows:

\begin{align}
R & = I - \frac{r \otimes r^{\flat}}{r^{2}} & R^* & = I - \frac{r^{\flat} \otimes r}{r^{2}}
\end{align}

Then the general CT in $\eunn(\kappa)$ is obtained by restricting $A \in C^2_0(\eunn)$ to $\eunn(\kappa)$. It is given as follows in $\eunn$ in contravariant form (see \cref{prop:CtFormEunnKap}):

\begin{equation} \label{eq:CTGenEunnKap}
L = R A R^* =  A + \kappa^2 \bp{r, A r} r \odot r - 2 \kappa (A r \odot r) \qquad L^{ij} = R \indices{^i_l} A^{lk} R\indices{^j_k}
\end{equation}

The matrix $A$ is called the \emph{parameter matrix} of the CT. We denote by $L_c$ the central CT in $\eunn$ with parameter matrix $A$. Note that $L = R L_c R^*$. We will see later that several questions concerning $L$ can be related to similar ones concerning $L_c$.

The canonical forms for these CTs can be enumerated by choosing a metric-Jordan canonical form for the pair $(A, g)$. The proofs of these canonical forms, which are given in \cref{sec:CtEunnKap}, can be omitted on first reading. Once these canonical forms are obtained, in \cref{sec:cTIrredSphC} we will calculate the characteristic polynomial for CTs in $\eunn(\kappa)$ by making use of the solution to the similar problem in $\eunn$. Using this, for ICTs we can calculate the transformation from their canonical coordinates to Cartesian coordinates and the metric in canonical coordinates. Then in \cref{sec:CtClassRedEunnKap} we will show how to obtain the warped product decompositions induced by reducible OCTs by making use of the solution to the similar problem in $\eunn$.

\section{Canonical forms for Concircular tensors in pseudo-Euclidean space} \label{sec:CtEunn}

\subsection{Standard Model of pseudo-Euclidean space}

In this section we calculate the CVs and CTs for $\eunn$ in its standard vector space model. These results are well known \cite{Crampin2007,Benenti2005a}, but we include it here for completeness.

First we define the \emph{dilatational vector field}, $r$, to be the vector field satisfying for any $p \in \eunn$, $r_p = p \in T_p \eunn$. In Cartesian coordinates $(x^i)$, we have

\begin{equation}
r = \sum_{i} x^{i} \partial_{i}
\end{equation}

In the following proposition we calculate the general CV in $\eunn$ as done originally in \cite{Crampin2007}.

\begin{proposition}[Concircular vectors in $\eunn$ \cite{Crampin2007}] \label{prop:CvEunn}
	A vector $v \in \ve(\eunn)$ is a CV in $\eunn$ where $n > 1$ iff there exists $a \in C^0_0(\eunn)$ and $b \in C^1_0(\eunn)$ such that
	\begin{equation}
	v = a r + b
	\end{equation}
	\noindent where $r$ is the dilatational vector field.
\end{proposition}
\begin{proof}
	In $\eunn$ with canonical Cartesian coordinates $(x^{i})$, \cref{eq:defnCTp} becomes:
	
	\begin{equation}
	\pderiv{v^{i}}{x^{j}} = \phi ?\delta^{i}_{j}?
	\end{equation}
	
	This equation can be easily solved by observing the following:
	
	\begin{align}
	\pderiv{\phi}{x^{k}} ?\delta^{i}_{j}? = \spderiv{v^{i}}{x^{k}}{x^{j}} = \pderiv{\phi}{x^{j}} ?\delta^{i}_{k}? 
	\end{align}
	
	Thus taking $i = j \neq k$, we find that $\pderiv{\phi}{x^{k}} = 0$. Thus $\phi \in \R$ and we find that $v$ must have the form given by $v^{i} = \phi x^{i} + b^{i}$ where each $b^{i} \in \R$. 
\end{proof}

Then using \cref{cor:CTpSCC} we can deduce the general CT in $\eunn$:

\begin{proposition}[Concircular tensors in $\eunn$] \label{prop:CtFormEunn}
	$L$ is a concircular 2-tensor in $\eunn$ where $n > 1$ iff there exists $A \in C^2_0(\eunn)$, $w \in C^1_0(\eunn)$ and $m \in C^0_0(\eunn)$ such that:
	\begin{equation}
	L = A + 2 w \odot r + m r \odot r
	\end{equation}
	\noindent where $r$ is the dilatational vector field. The tensors $A$, $w$ and $m$ are uniquely determined by $L$.
\end{proposition}

\subsection{Parabolic Model of pseudo-Euclidean space} \label{sec:parabMod}

In order to obtain canonical forms for CTs it will be useful to work with a different model of $\eunn$. We will refer to it as the parabolic model of $\eunn$, to be introduced shortly. The main reason for working with this model is because it is a spherical submanifold of the ambient space in which the isometries of $\eunn$ are linearized (see for example, \cite[Appendix~D]{Rajaratnam2014}).

Let $\E^{n+2}_{\nu+1}(\infty)$ be the light cone in $\E^{n+2}_{\nu+1}$, i.e. the set of non-zero null vectors. We define the parabolic embedding of $\eunn$ in $\E^{n+2}_{\nu+1}$ with mean curvature vector $-a \in \E^{n+2}_{\nu+1}(\infty)$ by \cite{Tojeiro2007}

\begin{equation}
\eunn \cong \punn := \{p \in \E^{n+2}_{\nu+1}(\infty) : \bp{p, a} = 1 \}
\end{equation}

An explicit isometry with $\eunn$ is obtained by choosing $b \in \punn$, i.e. $b$ is lightlike and $\bp{a, b} = 1$. We let $V := \spa{a,b}^\perp$, note that $V \cong \eunn$, then for $x \in V$:

\begin{equation} \label{eq:parabEmbed}
\psi(x) = b + x  - \frac{1}{2}x^{2}a \in \punn
\end{equation}

The map $\psi$ gives an explicit isometry between $\punn$ and $\eunn$. By definition of $\punn$, it follows that $T_{p} \punn = p^{\perp} \cap a^{\perp} = \spa{p,a}^\perp$. Also note that for $x \in \punn$

\begin{align}
\psi^{-1}(x)  = x - \bp{x,b} a - \bp{x,a} b
\end{align}

An important reason for working with $\punn$ is the following (see \cite{Nolker1996} or \cite[Appendix~D]{Rajaratnam2014}):

\begin{proposition}[Isometry group of $\punn$] \label{prop:isoPunnII}
	The isometry group of $\punn$ is:
	
	\begin{equation}
	I(\punn) = \{T \in O_{\nu + 1} (n+2) \; | \; T a = a \}
	\end{equation}
	
	Furthermore suppose we fix an isometry with $\eunn$ via \cref{eq:parabEmbed} by fixing a subspace $V \subset a^\perp$ such that $V \simeq \eunn$, then for $p \in V$ and $\tilde{p} \in V^\perp$ we have the following Lie group isomorphism:
	
	\begin{equation}
	\phi : \begin{cases}
	O_{\nu}^{n}(V) \ltimes V & \rightarrow I(\punn) \\
	(B, v) & \mapsto  \phi(B,v)
	\end{cases}
	\end{equation}
	
	where
	
	\begin{equation} \label{eq:parabIso}
	\phi(B,v)(p + \tilde{p}) = \tilde{p} + B p + \bp{a,\tilde{p}} v - (\bp{B p, v} + \frac{1}{2}\bp{a, \tilde{p}}v^2))a
	\end{equation}
\end{proposition}
\begin{proof}
	See \cite[Appendix~D]{Rajaratnam2014} or \cite[lemma~6]{Nolker1996} which covers the case when $\eunn$ is Euclidean.
\end{proof}
\begin{remark}
	If $\psi : \eunn \rightarrow \punn$ is the standard embedding from \cref{eq:parabEmbed}, then $\psi$ is equivariant. In other words, if we let $Tp := Bp + v$ for $(B, v) \in O_{\nu}^{n}(V) \ltimes V$ as above, and $\hat{T} := \phi(B,v)$ then $\psi \circ T (p) = \hat{T} \circ \psi (p)$.
\end{remark}

We also have the following:

\begin{lemma}
	For $\bar{p} \in V$ and $X \in T_{\bar{p}} V$
	
	\begin{equation}
	\psi_* X = X - \bp{X,\bar{p}} a
	\end{equation}
	
	For $Y \in T_{\psi(\bar{p})} \punn$, the inverse of the above map is given by:
	
	\begin{equation}
	P_b : \begin{cases}
	T_{\psi(\bar{p})} \punn & \rightarrow T_{\bar{p}} V \\
	Y & \mapsto Y - \bp{Y,b} a
	\end{cases}
	\end{equation}
\end{lemma}
\begin{proof}
	The first statement is clear. First observe that $P_b \psi_*X = X$. Now,
	
	\begin{equation}
	\psi_* P_b Y = Y - \bp{Y,b} a - \bp{Y,\bar{p}} a
	\end{equation}
	
	Now $0 = \bp{Y,\psi(\bar{p})} = \bp{Y,b} + \bp{Y,\bar{p}}$. Thus $\psi_* P_b Y = Y$.
\end{proof}

Furthermore we denote by $P_1$ the orthogonal projector onto $T \punn$. It is given as follows for $r \in \E^{n+2}_{\nu+1}$

\begin{equation}
P_1 : \begin{cases}
T_{r} \E^{n+2}_{\nu+1} & \rightarrow T_{r} \E^{n+2}_{\nu+1} \\
V & \mapsto V - \bp{V,r} a - \bp{V,a} r
\end{cases}
\end{equation}

We will now calculate the CT in $\E^{n+2}_{\nu+1}$ which restricts to the most general CT in $\punn$. Due to \cref{cor:CTpSCC} we only need to examine how CVs restrict. By \cref{prop:CvEunn} and \cref{thm:CTsDim}, the general CV in $\E^{n+2}_{\nu+1}$ can be written 

\begin{equation}
v = c_0 r + \sum_{i=1}^{n} c_i a_i + c_{n+1} b + c_{n+2} a
\end{equation}

\noindent where each $c_i \in \R$, $a_1,\dotsc, a_n$ is a basis for $V$ and $r$ is the dilatational vector field in $\E^{n+2}_{\nu+1}$. Then

\begin{align}
P_b P_1 v & = P_b (\sum_{i=1}^{n} c_i  (a_i - \bp{a_i, r} a) + c_{n+1} (b - \bp{b, r} a - r)) \\
& = \sum_{i=1}^{n} c_i  a_i - c_{n+1} x
\end{align}

\noindent where $x$ is the dilatational vector field in $V$. Then using \cref{cor:CTpSCC} we have proven the following:

\begin{proposition} \label{prop:parabIndCT}
	Suppose $\punn$ is identified with $\eunn$ by the embedding in \cref{eq:parabEmbed}. Denote by $V = \spa{a,b}^\perp$, let $\tilde{A} \in C^2_0(V)$, $w \in C^1_0(V)$, and $m \in C^0_0(V)$. Define
	
	\begin{equation} \label{eq:parabConCT}
	A = \tilde{A} + m b \odot b - 2 w \odot b
	\end{equation}
	
	Then the restriction of $A$ to $V$, denoted $L$, via the embedding in \cref{eq:parabEmbed} is:
	\begin{equation}
	L = \tilde{A} + m r \odot r + 2 w \odot r
	\end{equation}
\end{proposition}

Note that $A$ is completely determined by the condition $A b = 0$. Now for $A \in C^2_0(\E^{n+2}_{\nu+1})$, define $A_b$ by 

\begin{equation}
(A_b)^{ij} := (P_b)\indices{^i_l} A^{lk} (P_b)\indices{^j_k}
\end{equation}

Note that $b$ is an eigenvector of $A_b$ with eigenvalue $0$. Also observe that

\begin{align}
P_1 P_b & = P_1 - w \otimes b^\flat + w \otimes b^\flat = P_1
\end{align}

The above equation shows that $A$ and $A_b$ induce the same CT on $\punn$. From the calculations proceeding \cref{eq:parabConCT} we see that 
\begin{equation}
\{a_1 - \bp{a_1, r} a, \dotsc, a_n - \bp{a_n, r} a, b - \bp{b, r} a - r\}
\end{equation}
is basis for the space of CVs on $\punn$. Thus it follows from \cref{cor:CTpSCC} and the proceeding calculations that $A,B \in C^2_0(\E^{n+2}_{\nu+1})$ induce the same CT on $\punn$ iff for some $b \in \punn$ we have

\begin{equation}
A_b = B_b
\end{equation}

Furthermore, one should note that if $b,c \in \punn$, then $(A_c)_b = A_b$. Hence it follows that if $A_b = B_b$ for some $b \in \punn$ then $A_c = B_c$ for all $c \in \punn$.

\subsection{Existence of Canonical forms}

In this section $A \in C^2_0(\E^{n+2}_{\nu+1})$. We are interested in finding canonical forms for the CT on $\punn$ induced by this tensor. As it was shown in the previous section, the induced CT depends only on $A_b$ for some $b \in \punn$. Hence our goal will be to find $\tilde{b} \in \punn$ such that $A_{\tilde{b}}$ is in a canonical form. Since the isometry with $\eunn$ (see \cref{eq:parabEmbed}) is fixed by a vector $b \in \punn$, we will then choose $T \in I(\punn)$ such that $T\tilde{b} = b$. This will transform $A_{\tilde{b}}$ to $(T_* A)_b$ which can be restricted to $\eunn$ using \cref{prop:parabIndCT} to obtain a canonical form for the original CT in $\eunn$.

To obtain the canonical choice of $b \in \punn$, first note that $A_b$ is completely determined by the fact that $A_b b = 0$. Secondly, note that since isometries of $\punn$ fix $a$, it follows that for each $l \geq 0$, $\bp{a,A^l a}$ are invariants of $A$. Although these are in general not invariants of the CT induced by $A$, they will play a significant role in the classification. Thirdly, since $a$ cannot be transformed by isometries, we will attempt to choose $b \in \punn$ such that $a$ is a basis vector in a metric-Jordan canonical basis for $A_b$. Since $\bp{a,b} = 1$, one can deduce that (using the metric-Jordan canonical form \cite[Appendix~C]{Rajaratnam2014}) in the simplest cases, $a,b$ lie in the same eigenspace of $A_b$ or $a$ generates a Jordan cycle ending in a constant multiple of $b$. These observations motivate our search for $b$.

For the following calculations, $b \in\punn$ is arbitrary and we let $\tilde{A} :=  A_b$. The following lemma will get us started:

\begin{lemma} \label{lem:tildAw}
	Suppose there is $k \in \N$ such that $\bp{a,A^l a} = 0$ for $ 0 \leq l < k$. Then for each $ 0 \leq l \leq k $
	
	\begin{equation} \label{eq:tildAw}
	\tilde{A}^l a = A^l a  - \sum_{j=0}^{l-1} \bp{b, A^{l-j} a} \tilde{A}^{j} a
	\end{equation}
	
	Furthermore, if $ 0 \leq l \leq k$ then
	\begin{align}
	\bp{a, \tilde{A}^l a} & = \bp{a, A^l a} \label{eq:tildAwInv}
	\end{align}
	So the constants $\bp{a, A^l a}$ are invariants of the CT on $\punn$ induced by $A$.
\end{lemma}
\begin{proof}
	We prove \cref{eq:tildAw} by induction. It clearly holds for $l = 0,1$. Now assume it holds for $l-1$, then
	
	\begin{align}
	\tilde{A}^l a & = \tilde{A}A^{l-1} a  - \sum_{j=0}^{l-2} \bp{b, A^{l-1-j} a} \tilde{A}^{j+1} a \\
	& = A^l a - a\bp{b, A^l a} - \sum_{j=0}^{l-2} \bp{b, A^{l-1-j} a} \tilde{A}^{j+1} a \\
	& = A^l a - a\bp{b, A^l a} - \sum_{j=1}^{l-1} \bp{b, A^{l-j} a} \tilde{A}^{j} a \\
	& = A^l a  - \sum_{j=0}^{l-1} \bp{b, A^{l-j} a} \tilde{A}^{j} a
	\end{align}
	
	Hence the first equation follows by induction.
	
	Suppose $ 0 \leq l < k$, then
	
	\begin{equation}
	\bp{a , \tilde{A}^l a} = - \sum_{j=0}^{l-1} \bp{b, A^{l-j} a} \bp{a, \tilde{A}^{j} a}
	\end{equation}
	
	Thus it follows by induction that $\bp{a , \tilde{A}^l a} = 0$. Thus $\bp{a, \tilde{A}^k a} = \bp{a, A^k a}$.
\end{proof}

Now, define $\omega_i$ by

\begin{equation}
\omega_i := \bp{ a , A^{i+1} a}
\end{equation}

We will also need the following lemma to calculate $\omega_i$ in $\eunn$.

\begin{lemma}
	Suppose $A$ has the form given by \cref{eq:parabConCT}, then
	
	\begin{equation} \label{eq:parabatilde}
	A^l a = \begin{cases}
	m b - w & l = 1 \\
	\bp{w,\tilde{A}^{l-2} w} b -\tilde{A}^{l-1} w &  l > 1
	\end{cases}
	\end{equation}
	
	and $\omega_i$ is given by \cref{eq:omegI}.
\end{lemma}

Using the above lemma we can also apply the definitions of index, sign and degeneracy of CTs in $\eunn$ from \cref{def:CtEunnInd} to CTs in $\punn$.

\subsubsection{Non-degenerate cases}

Now we consider the case where there exists a least $k \in \N$ such that $\bp{a,A^k a} \neq 0$. This will be the most important case for our interests. Motivated by special cases and the metric-Jordan canonical form of $\tilde{A}$ discussed earlier, we will try to find $b$ such that $a,\tilde{A} a,\dotsc,\tilde{A}^k a$ forms a skew-normal sequence with $\bp{a , A^k a} b = \tilde{A}^k a$. The following lemma describes $b$ provided it exists:

\begin{lemma}
	Suppose there is $k \in \N$ such that $\bp{a,A^l a} = 0$ for $ 0 \leq l < k$ and $\bp{a,A^k a} \neq 0$. Assume there exists a $b$ such that $\bp{a , A^k a} b = \tilde{A}^k a$ and $\bp{\tilde{A}^j a, \tilde{A}^k a} = 0$ for all $1 \leq j \leq k $. Then $b$ must satisfy the following equations for each $l \in \{0,\dotsc, k\}$
	
	\begin{equation} \label{eq:bAla}
	2 \bp{b, A^{l} a}  = \frac{\bp{A^l a , A^k a }}{\bp{ a , A^k a}} - \sum_{j=1}^{l-1} \bp{b, A^{l-j} a} \bp{b, A^{j} a}
	\end{equation}
\end{lemma}
\begin{proof}
	Suppose $0 < l \leq k$. Expanding $\tilde{A}^k a$ using \cref{eq:tildAw}, we have
	
	\begin{align}
	\bp{\tilde{A}^l a, \tilde{A}^k a} & = \bp{\tilde{A}^l a, A^k a} - \bp{b, A^l a} \bp{\tilde{A}^l a, \tilde{A}^{k-l} a} \\
	& \overset{\eqref{eq:tildAwInv}}{=} \bp{\tilde{A}^l a, A^k a} - \bp{b, A^l a} \bp{a , A^{k} a}
	\end{align}
	
	By imposing the condition $\bp{\tilde{A}^l a, \tilde{A}^k a} = 0$, the above equation implies that:
	
	\begin{equation} \label{eq:lemRel}
	\bp{\tilde{A}^l a, A^k a} - \bp{b, A^l a} \bp{a , A^{k} a} = 0
	\end{equation}
	
	Now expanding $\tilde{A}^l a$ using \cref{eq:tildAw}, the above equation becomes
	
	\begin{align}
	\bp{\tilde{A}^l a , A^k a } & = \bp{A^l a  - \sum_{j=0}^{l-1} \bp{b, A^{l-j} a} \tilde{A}^{j} a , A^k a } \\
	& = \bp{A^l a , A^k a } - \sum_{j=0}^{l-1} \bp{b, A^{l-j} a} \bp{ \tilde{A}^{j} a , A^k a } \\
	& = \bp{A^l a , A^k a } -  \bp{b, A^{l} a}\bp{ a , A^k a } -  \sum_{j=1}^{l-1} \bp{b, A^{l-j} a} \bp{ \tilde{A}^{j} a , A^k a } \\
	& \overset{\eqref{eq:lemRel}}{=} \bp{A^l a , A^k a } -  \bp{b, A^{l} a}\bp{ a , A^k a } -  \sum_{j=1}^{l-1} \bp{b, A^{l-j} a} \bp{b, A^{j} a} \bp{ a , A^k a} 
	\end{align}
	
	Equating the above equation with \cref{eq:lemRel} and solving for $\bp{b, A^{l} a}$ proves the result.
\end{proof}

Now we will use the above lemma and \cref{eq:tildAw} to construct a vector $b$ such that $\tilde{A}$ is in canonical form. First define a sequence $b_1,\dotsc,b_k$ of scalars recursively as follows:

\begin{equation} \label{eq:parabNonDeSeqI}
2 b_l := \frac{\bp{A^l a , A^k a }}{\bp{ a , A^k a}} - \sum_{j=1}^{l-1} b_{l-j} b_{j}
\end{equation}

Then define vectors $s_0,s_1,\dotsc,s_k$ as follows:

\begin{equation} \label{eq:parabNonDeSeqII}
s_l := A^l a  - \sum_{j=0}^{l-1} b_{l-j} s_{j}
\end{equation}

Then define $b$ by $ b \bp{ a , A^k a} := s_k$. The following lemma shows that this choice does work:

\begin{proposition} \label{prop:ParbconForm}
	The vectors $s_0,s_1,\dotsc,s_k$ form a skew-normal sequence with  $\bp{s_0, s_k} = \bp{ a , A^k a}$. If $\tilde{A}^l a$ are defined as in \cref{eq:tildAw} with the above vector $b$ then $\tilde{A}^l a = s_l$.
\end{proposition}
\begin{proof}
	The fact that $s_0,s_1,\dotsc,s_k$ form a skew-normal sequence follows verbatim from \cref{lem:tildAw} and the proceeding arguments by replacing $s_l \rightarrow \tilde{A}^l a$ and $ b_l \rightarrow \bp{b, A^l a}$.
	
	Suppose that $s_0,s_1,\dotsc,s_k$ form a skew-normal sequence where $\bp{s_0, s_k} = \bp{ a , A^k a}$. By definition of $s_l$, it follows that each $A^l a$ can be expanded in this basis as:
	
	\begin{equation}
	A^l a = s_l + \sum_{j=0}^{l-1} b_{l-j} s_{j}
	\end{equation}
	
	Thus
	
	\begin{equation}
	\bp{A^k a, a} \bp{b, A^l a} = \bp{s_k , A^l a} =  b_l \bp{A^k a, a}
	\end{equation}
	
	Hence $b_l = \bp{b, A^l a}$. Then it follows by definition of $s_l$ and $\tilde{A}^l a$ in  \cref{eq:tildAw} that $\tilde{A}^l a = s_l$.
\end{proof}

Now suppose $A$ is in the canonical form stated above. Let $V = \spa{a,b}^\perp$ where $b$ was chosen as above. Then $H = \spa{a,Aa,\dotsc, A^k a}$ is a non-degenerate $A$-invariant subspace (see \cite[Lemma~8.1.1]{Rajaratnam2014}). Hence $H^\perp$ is a non-degenerate $A$-invariant subspace complementary to $H$. We now mention more precisely what we mean by ``the'' canonical form:

\begin{definition} \label{defn:parabCTconForm}
	Suppose $L$ is a CT in $\punn$ with parameter matrix $A$ as above and index $k' := k - 1  \geq 0$, i.e. $L$ is \emph{non-degenerate}. The \emph{iso-canonical form} for $L$ is the metric-Jordan canonical form for $(A|_{H^\perp}, g|_{H^\perp})$ together with the index $k'$ and constant $\bp{a, A^{k'+1} a} \in \R \setminus \{0\}$.
\end{definition}

We will prove later on that this canonical form is uniquely determined by $L$. But for now we will examine this canonical form further. Let $\tilde{A} := A|_{H^\perp}$, then we can write:

\begin{equation}
A = \tilde{A} + \omega_0 b \odot b - 2 w \odot b
\end{equation}

\noindent where $w = \omega_0 b -A a$.

If $\omega_0 \neq 0$ then it follows that $w = 0$ and it follows by \cref{prop:parabIndCT} that the induced CT on $V$ is

\begin{equation}
\tilde{A} + \omega_0 r \odot r
\end{equation}

Thus after dividing by $\omega_0$ we get the central CT from \cref{thm:conTenCanForm}. If $\omega_0 = 0$, one can check that $w,\tilde{A} w,\dotsc,\tilde{A}^{k-2} w \in V$ form a skew-normal sequence with $\bp{w, \tilde{A}^{k-2} w} = \omega_{k-1}$. It follows by \cref{prop:parabIndCT} that the induced CT on $V$ is

\begin{equation}
\tilde{A} + 2 w \odot r
\end{equation}

This CT is a constant multiple of a (null) axial CT with the same index and sign from \cref{thm:conTenCanForm} (after an appropriate choice of basis).

\paragraph{Transformation to Canonical form: } \label{par:conFormTrans}  We now denote by $\tilde{b}$ the vector $b$ obtained above which puts $A$ into a canonical form.  The vector $b \in \punn$ is fixed by an isometry with $\eunn$ (see \cref{eq:parabEmbed}), furthermore we let $V = \spa{a,b}^\perp$. We can assume $A$ has the form given by \cref{eq:parabConCT}. The last problem is to choose $T \in I(\punn)$ such that $T \tilde{b} = b$. We can obtain a unique transformation if we require $T$ to induce a translation in $V$. Indeed, by \cref{eq:parabIso} the most general transformation of this type is

\begin{equation}
T = I - a \otimes (\frac{1}{2}v^2 a^\flat + v^\flat) + v \otimes a^\flat
\end{equation}

\noindent where $v \in V$ is arbitrary. The unique transformation with the above form satisfying $T \tilde{b} = b$ is obtained by taking

\begin{align}
v =  b - \tilde{b} + a \bp{\tilde{b}, b}
\end{align}

We now proceed to calculate $v$. First we can write

\begin{equation}
\tilde{b} = \frac{1}{\omega_{k-1}}\sum_{i=0}^{k} c_i A^i a
\end{equation}

Since $\bp{b, A^l a} = 0$ for any $l > 0$, we see that

\begin{equation}
\tilde{b} - a \bp{\tilde{b}, b} = \frac{1}{\omega_{k-1}}\sum_{i=1}^{k} c_i A^i a
\end{equation}

Since for $0 < l < k$, $\bp{a, A^l a} = 0$ it follows by \cref{eq:parabatilde} that $A^l a = -\tilde{A}^{l-1} w$. Thus

\begin{align}
v & = - \frac{1}{\omega_{k-1}}\sum_{i=1}^{k} c_i A^i a + b \\
& = \frac{1}{\omega_{k-1}}\sum_{i=1}^{k} c_i \tilde{A}^{i-1} w
\end{align}

where the last equation follows from the fact that $c_k = 1$. We have calculated the first four coefficients (which are sufficient for Euclidean and Minkowski space):



\begin{align}
c_k & = 1 \\
c_{k-1} & = - \frac{1}{2} \frac{\omega_{k}}{\omega_{k-1}} \\
c_{k-2} & = \frac{1}{16} \frac{(-8 \omega_{k-1} \omega_{k+1} + 6 \omega_{k}^2 )}{\omega_{k-1}^2} \\
c_{k-3} & = \frac{1}{16} \frac{(-8 \omega_{k-1}^2 \omega_{k+2} + 12 \omega_{k-1} \omega_{k} \omega_{k+1} - 5 \omega_{k}^3 )}{\omega_{k-1}^3}
\end{align}

In particular when $k = 1$ and $2$ respectively we have the following:

\begin{align}
v & = \frac{w}{\omega_0} \\
v & = \frac{1}{\omega_1}(\tilde{A} w -\frac{1}{2}\frac{\omega_2}{\omega_1} w)
\end{align}

Finally, we note that by equivariance of the map $\psi$ (see remark after \cref{prop:isoPunnII}), one only needs to apply the isometry $T : V \rightarrow V$ given by $r \mapsto r + v$ to send the induced CT in $V$ into canonical form. Hence in practice one does not need to work in $\punn$.

\subsubsection{Degenerate cases} \label{sec:degenCase}

We now consider the case where $\bp{a,A^l a} = 0$ for every $l \in \N$. First note that the dimension of the subspace spanned by $a, A a, \dotsc$ must be at most $n-1$ by non-degeneracy of the scalar product. So there exists a least $l \leq n-1$ such that $\{a,Aa,\dotsc,A^{l}a\} \subseteq a^\perp$ is a linearly independent set but $A^{l+1}a \in \spa{a,Aa,\dotsc,A^{l}a}$. Thus it follows that $A^{m}a \in \spa{a,Aa,\dotsc,A^{l}a}$ for all $m > l$. Also note by \cref{lem:tildAw} it follows that these properties are invariant under the transformation $A \rightarrow A_b$.

\begin{parts}
	\item $l = 0$ \\
	In this case $a$ is an eigenvector of $A$. After transforming $A$ to $A_b$ (if necessary), we can assume that $A a = 0$. Also $A b = 0$, then since $\bp{a,b} = 1$ it follows that $\spa{a,b}$ is a non-degenerate $A$-invariant subspace. Hence after identifying $\eunn \simeq \spa{a,b}^\perp$, it follows by \cref{prop:parabIndCT} that $A$ restricts to a Cartesian CT on $\eunn$.
	\item $l \geq 1$ \\
	Fix $b \in \punn$, let $V = \spa{a,b}^\perp$ and assume $A b = 0$. Then we can write:
	
	\begin{equation}
	A = \tilde{A} + 2 w \odot b
	\end{equation}
	
	Now note that for any $j \in \N$, $\bp{a, A^j a} = 0$. Suppose inductively that for all $1 \leq j \leq i$ that $A^j a \in V$ then
	
	\begin{equation}
	A A^i a = \tilde{A}A^i a \in V
	\end{equation}
	
	\noindent since $\bp{A^i a, w} = \bp{A^i a, A a} = 0$ and $\bp{A^i a, b} = 0$. Hence by induction for any $j \in \N$, $\bp{b, A^j a} = 0$. Thus $A a,\dotsc,A^l a,A^{l+1} a \in V$.
	
	In particular, when $l = 1$ we see that $w$ is a lightlike eigenvector of $\tilde{A}$. Then by \cref{prop:parabIndCT}, $A$ induces the following CT in $\eunn$
	
	\begin{equation}
	L = \tilde{A} - 2 w \odot r
	\end{equation}
	
	Observe that $w$ is a lightlike eigenvector of $L$ with non-constant eigenfunction. Thus $L$ is never an OC-tensor because lightlike eigenvectors of OC-tensors must have constant eigenfunctions.
	
	If $l > 1$, we see that $A a,A^2 a \in V$ are linearly independent orthogonal lightlike vectors. Thus this case can't occur in Euclidean or Minkowski case, so we ignore it.
\end{parts}

%

\subsection{Uniqueness of Canonical Forms}

In this section we will show that the canonical forms obtained in the previous section are uniquely determined by a given CT in $\punn$. As a consequence of this we will show that the different canonical forms divide the CTs into isometrically inequivalent classes. We will be working with the case when the CT is non-degenerate as the other cases are either straightforward or uninteresting.

Suppose $L$ and $M$ are CTs in $\punn$ with parameter matrices $A$ and $B$ respectively. We observed at the end of \cref{sec:parabMod} that $L = M$ iff for one (hence all) $b \in \punn$:

\begin{equation}
A_b = B_b
\end{equation}

Thus it follows that $L = T_*M$ for some $T \in I(\punn)$ iff for one (hence all) $b \in \punn$:

\begin{equation}
A_b = (T_*B)_b
\end{equation}

\begin{lemma} \label{lem:parabConUniq}
	Suppose $A_2$ is a parameter matrix, and $A_1 =  (A_2)_b $ for some $b \in \punn$. Assume each $A_i$ have the same index and admit a vector $b_i$ which transforms it to canonical form according to \cref{prop:ParbconForm}. Then $b_1 = b_2$.
\end{lemma}
\begin{proof}
	Let $A_0 = (A_2)_{b_2}$, then $A_1 = (A_0)_{b}$. Since $A_0$ is in canonical form, \\ $a,A_0 a, \cdots, A_0^k a$ forms an adapted cycle of generalized eigenvectors for $A_0$ with eigenvalue $0$. In this case $\bp{a, A_0^k a} \in \R \setminus \{0\}$.
	
	Let $b_1$ be the vector admitted by $A_1$ and let $A_3 := (A_1)_{b_1} = (A_0)_{b_1}$. Now by \cref{prop:ParbconForm} and \cref{lem:tildAw}, $b_1$ satisfies:
	
	\begin{equation} \label{eq:uncf}
	\bp{a, A_1^k a} b_1 = A_3^k a = A_0^k a  - \sum_{j=0}^{k-1} \bp{b_1 , A_0^{k-j} a} A_3^{j} a
	\end{equation}
	
	Since $A_3$ is in canonical form, it follows for each $l \in \{1,\cdots, k\}$, $\bp{b_1, A_0^{l} a}$ satisfies \cref{eq:bAla}. Then since $A_0$ is in canonical form, we have $\bp{b_1, A_0^{l} a} = 0$ for $l \in \{1,\cdots, k\}$. Thus \cref{eq:uncf} shows that
	
	\begin{equation}
	\bp{a, A_1^k a} b_1 = A_0^k a = \bp{a, A_1^k a} b_2
	\end{equation}
	
	Hence $b_1 = b_2$.
\end{proof}

In the following theorem we will show that the iso-canonical form defined in \cref{defn:parabCTconForm} for non-degenerate CTs is uniquely determined by the CT.

\begin{theorem}[Isometric Equivalence of CTs in $\eunn$]
	Suppose $L$ and $M$ are CTs in $\punn$ such that $M$ has an index $k \geq 0$. Then $L = T_*M$ for some $T \in I(\punn)$ iff $L$ and $M$ have the same iso-canonical form.
\end{theorem}
\begin{proof}
	Assume that $L = T_*M$ for some $T \in I(\punn)$. Then for some $b \in \punn$:
	
	\begin{equation}
	A_b = (T_*B)_b
	\end{equation}
	
	By the above equation and \cref{lem:tildAw} it follows that the index of $L$ is also $k$. Let $b_2$ be the vector which puts $B$ in canonical form given by \cref{prop:ParbconForm}. Then $T b_2$ sends $T_* B$ to canonical form. By \cref{lem:parabConUniq}, $T b_2$ is the vector obtained from \cref{prop:ParbconForm} which puts $A$ in canonical form. Let $\tilde{b} := T b_2$ then
	
	\begin{equation}
	A_{\tilde{b}} = (T_* B)_{\tilde{b}} = T_* (B_{b_2})
	\end{equation}
	
	Hence $B_{b_2}$ is isometric to $A_{\tilde{b}}$. Then it follows from the uniqueness of the metric-Jordan canonical form \cite[Appendix~C]{Rajaratnam2014} that $A_{\tilde{b}}$ and $B_{b_2}$ have the same iso-canonical form.
	
	Conversely suppose $L$ and $M$ have the same iso-canonical form. Then $A$ (resp. B) each admit a vector $b_1 \in \punn$ (resp. $b_2 \in \punn$) such that $A_{b_1}$ and $B_{b_2}$ have the same iso-canonical form. Then one can easily construct $T \in I(\punn)$ which transforms a metric-Jordan canonical basis of $B_{b_2}$ into $A_{b_1}$, so that $A_{b_1} = T_*(B_{b_2})$. Thus
	
	\begin{align}
	\Rightarrow  T (B_{b_2})^k a & = (A_{b_1})^k a \\
	\Rightarrow T b_2 & = b_1
	\end{align}
	
	Note that in the last equation we have used the fact that $\bp{a, B^k a} = \bp{a, A^k a}$. Then
	
	\begin{equation}
	A_{b_1}  = T_*(B_{b_2}) = (T_* B)_{b_1}
	\end{equation}
	
	Thus $L = T_*M$, which proves the converse.
\end{proof}

\paragraph{Geo-Canonical forms } We now give a geo-canonical form for non-degenerate CTs in $\punn$. Suppose $L$ is such a CT with index $k$ and parameter matrix $A$ in iso-canonical form. Then for $c \in \R$, $c L$ has parameter matrix $c A$ and

\begin{equation}
\bp{a , (c A)^{k+1} a} = c^{k+1} \bp{a, A^{k+1} a}
\end{equation}

Hence after an appropriate transformation $L \rightarrow c L$, we can assume

\begin{equation}
\bp{a, A^{k+1} a} = 
\begin{cases}
1  & 1 \text{ if $k$ is even} \\
\pm 1 & \text{ if $k$ is odd}
\end{cases}
\end{equation}

Note that when $k$ is odd, $c$ is only determined up to sign. Hence there are two possible geo-canonical forms in this case. Now, if $L$ is an axial CT, we can fix $d \in \R$ by requiring that $(A + d I)^{k} a \in \spa{a,b}$. This condition is satisfied in the iso-canonical form. If $L$ is central, we choose $d$ such that the real part of the smallest eigenvalue (see \cref{defn:ordC}) of $A|_{H^\perp}$ is zero.

%

\section{Canonical forms for Concircular tensors in Spherical submanifolds of pseudo-Euclidean space} \label{sec:CtEunnKap}
\subsection{Obtaining concircular tensors in umbilical submanifolds by restriction}

Let $\tilde{M}$ be a pseudo-Riemannian submanifold of $M$ with Levi-Civita connections $\tilde{\nabla}$ and $\nabla$ respectively. We say $\tilde{M}$ is an umbilical submanifold (see \cite{barrett1983semi} for more details) if there exists $H \in \ve(\tilde{M})^{\perp}$ (i.e. $H$ is orthogonal to $T \tilde{M}$) called the \emph{mean curvature normal} of $\tilde{M}$ such that

\begin{equation}
\nabla_{x}y = \tilde{\nabla}_{x}y + \bp{x,y}H
\end{equation}

\noindent for all $x,y \in \ve(\tilde{M})$. By generalizing an observation made in \cite{Crampin2003} one can deduce the following:

\begin{proposition}[Restriction of CTs to umbilical submanifolds \cite{Crampin2003}] \label{prop:CtRestUmb}
	Suppose $\tilde{M}$ is an umbilical submanifold of $M$ with mean curvature normal $H$ and $L$ is a concircular $r$-tensor on $M$ with conformal factor $C$ in covariant form. Then the pullback of $L$ to $\tilde{M}$ is a concircular $r$-tensor with conformal factor equal to the pullback of $C + r L(H)$, where in components, $L(H)_{i_{1},\dotsc,i_{r-1}} = L_{i_{1},\dotsc,i_{r-1}j}H^{j}$.
\end{proposition}

Since spherical submanifolds are umbilical submanifolds and $\eunn(\kappa)$ is a spherical submanifold (see for example \cite[Appendix~D]{Rajaratnam2014}), the above proposition allows us to obtain CTs on $\eunn(\kappa)$. We will do this in the following section.

\subsection{Concircular tensors in Spherical submanifolds of pseudo-Euclidean space}

In this section we study the CTs in $\eunn(\kappa)$ via the canonical embedding in $\eunn$. Let $r$ denote the dilatational vector field, we work on the subset of $\eunn$ for which $r^2 \neq 0$. Let $E := r^{\perp}$ and let $L$ be a CT on $M$. To obtain the CT on $\eunn(\frac{1}{r^2})$ (which is an integral manifold of $E$), we first let $R := I - \dfrac{r^{\flat} \otimes r}{r^{2}}$ where $I$ is the identity endomorphism then $L_{E} := L|_{E}$ is given as follows:

\begin{equation}
(L_{E})^{ij} = R \indices{^i_l} L^{lk} R\indices{^j_k}
\end{equation}

Now we will calculate the general CT on $\eunn(\kappa)$.

\begin{propMy}[Concircular tensors in $\eunn(\kappa)$] \label{prop:CtFormEunnKap}
	$\tilde{L}$ is a concircular tensor in $\eunn(\frac{1}{r^2})$ where $n > 2$ iff there exists $A \in C^2_0(\eunn)$ such that $\tilde{L}$ has the following form embedded in $\eunn$:
	
	\begin{equation}
	L = A_{E} =  A + \frac{\bp{r, A r}}{r^4} r \odot r - \frac{2}{r^2} (A  r \odot r)
	\end{equation}
	
	$A$ is uniquely determined by $\tilde{L}$. Furthermore $\tilde{L}$ is covariantly constant iff its a constant multiple of the metric on $\eunn(\frac{1}{r^2})$, i.e. $A = c \, G$ for some $c \in \R$ where $G$ is the metric of $\eunn$.
\end{propMy}
\begin{proof}
	Fix $\tilde{L} \in S^2(\eunn(\frac{1}{r^2}))$. Choose an orthonormal basis $a_1,\dotsc,a_n$ for $\eunn$. Let $R^* = I - \dfrac{r \otimes r^{\flat} }{r^{2}}$, then it follows from \cref{prop:CtRestUmb} that the vectors
	
	\begin{equation}
	R^{*}a_{i} = a_{i} - \frac{\bp{r,a_{i}}}{r^{2}} r \quad i = 1,\dotsc,n
	\end{equation}
	
	
	
	\noindent are CVs on $\eunn(\frac{1}{r^2})$. Furthermore one can check that these vectors are linearly independent. Thus by \cref{cor:CTpSCC} every CT can be written uniquely as a linear combination of symmetric products of the above CVs. Thus it follows that we can choose a unique $A \in C^2_0(\eunn)$ such that $\tilde{L} = A_{E}$ on $\eunn(\frac{1}{r^2})$. In $\eunn$, $A_{E}$ is given as follows:
	
	\begin{align}
	A_{E} & = R^{*}A R \\
	& = A + A(r^{\flat},r^{\flat}) \frac{r \odot r}{r^{4}} - \frac{2}{r^2} A(r^{\flat}) \odot r \\
	& = A + \bp{r, A r} \frac{r \odot r}{r^{4}} - \frac{2}{r^2} A r \odot r
	\end{align}
	
	Conversely by \cref{cor:CTpSCC} it follows that for any $A \in C^2_0(\eunn)$, $A_{E}$ corresponds to CT on $\eunn(\frac{1}{r^2})$.
	
	
	The last statement follows from \cref{prop:CtRestUmb}.
\end{proof}
\begin{remark}
	The general CT in $\eunn(\kappa)$ has been obtained in \cite[Section~3]{Thompson2005} with respect to certain canonical coordinates for these spaces. They use a different method for obtaining these tensors based on the theory developed in their article.
\end{remark}

For the remainder of this article we will always work with CTs in $\eunn(\kappa)$ via the tensor $L$ defined in $\eunn$ in the above proposition.

\begin{definition} \label{defn:sphCTconForm}
	Suppose $L$ is a CT in $\eunn(\kappa)$ with parameter matrix $A \in S^2(\eunn)$ as above. The iso-canonical form for $L$ is the metric-Jordan canonical form for $(A, g)$.
\end{definition}

Except for hyperbolic space $H_0^{n-1}$ and the space anti-isomorphic to it $S_{n-1}^{n-1}$, uniqueness of the iso-canonical form follows from the uniqueness of the metric-Jordan canonical form and the fact that $I(\eunn(\kappa)) = O(\eunn)$ \cite{barrett1983semi}. For $H_0^{n-1}$, $I(H_0^{n-1})$ is the subset of $O(\E^{n}_{1})$ that preserves time orientation \cite{barrett1983semi}. In this case, minor modifications of the proof of the uniqueness of the metric-Jordan canonical form will show that it holds true with $I(H_0^{n-1})$ in place of of $O(\E^{n}_{1})$. A similar argument goes for $S_{n-1}^{n-1}$. Hence we have proven the following:

\begin{theorem}[Isometric Equivalence of CTs in $\eunn(\kappa)$] \label{thm:EkisoEquivCTs}
	Suppose $L$ and $M$ are CTs in $\eunn(\kappa)$. Then $L = T_*M$ for some $T \in I(\eunn(\kappa))$ iff $L$ and $M$ have the same iso-canonical form.
\end{theorem}

\paragraph{Geo-Canonical forms } \label{par:sphGeoCan} By definition, the restriction of $G$ to $\eunn(\kappa)$ is the metric on $\eunn(\kappa)$. Hence we see that if $a \in \R \setminus \{0\}, b \in \R$ and $A \in C^2_0(\eunn)$, then $A$ and $a A + b G$ induce geometrically equivalent CTs on $\eunn(\kappa)$ (see \cref{prop:CtEquiv}). We now show how to obtain the geo-canonical forms. Suppose $\lambda_1,\dotsc,\lambda_k \in \C$ are the distinct eigenvalues of $A$. Let $\Abs{\cdot}$ denote the modulus of a complex number, then define:

\begin{equation}
\Abs{a} := \min_{i,j} \Abs{\lambda_i - \lambda_j} > 0
\end{equation}

Note that this quantity is invariant under geometric equivalence. By making the transformation  $\lambda_i \rightarrow \frac{\lambda_i}{\Abs{a}}$, we can assume $\Abs{a} = 1$. Furthermore we choose $b \in \R$ such that the real part of the smallest eigenvalue (see \cref{defn:ordC}) of $A$ is zero. Since its not possible to specify the sign of $a$, we conclude that there are (in general) two geo-canonical forms for CTs in $\eunn(\kappa)$. Although in practice one can often use more information from the metric-Jordan canonical form of $A$ to obtain a single geo-canonical form, as the following example shows:

\begin{example}[Separable coordinates in hyperbolic space] \label{ex:geoEquivHypSpa}
	Consider $H^{n-1} = \E^{n}_{1}(-1)$ with the standard metric:
	
	\begin{equation}
	g = \diag(-1,1,\dotsc,1)
	\end{equation}
	
	For $\lambda_1 < \cdots < \lambda_n \in \R$ define two linear operators $A_1$ and $A_2$ as follows:
	
	\begin{align}
	A_1 & = \diag(\lambda_1,\dotsc,\lambda_n) \\
	A_2 & = \diag(-\lambda_1,\dotsc,-\lambda_n)
	\end{align}
	
	These two operators are isometrically inequivalent since they have different metric-Jordan canonical forms. The timelike eigenvalue of the first is the smallest, while that of the second is the largest. Although $- A_2 = A_1$ and hence the CT on $H^{n-1}$ induced by these operators are geometrically equivalent. So, in $H^{n-1}$ we can work with inequivalent CTs (under change of sign) by working with those whose parameter matrix has a timelike eigenvalue which is less than or equal to $\lfloor \frac{n}{2} \rfloor$ spacelike eigenvalues.
	
	Thus the set of eigenvalues $\lambda_1 < \cdots < \lambda_n \in \R$ induce $\lceil \frac{n}{2} \rceil$ inequivalent separable coordinates in $H^{n-1}$; in contrast with the $n$ inequivalent separable coordinates in $\E^{n}_{1}$ induced by central CTs.
\end{example}


\section{Properties of Concircular tensors in Spaces of Constant Curvature} \label{sec:cTIrred}

In this section we will assume that each CT in $\eunn$ or $\eunn(\kappa)$ is in a canonical form listed in \cref{sec:sumRes}. Furthermore we will assume that the Cartesian coordinates are chosen such that the parameter matrix $A_c$ is in the complex metric-Jordan canonical form stated in \cref{thm:comMetJFor} (see \cite[Appendix~C]{Rajaratnam2014} for details). We now describe how to transform to real Cartesian coordinates such that $A_c$ obtains the real metric-Jordan canonical form (see \cite[Appendix~C]{Rajaratnam2014}). Suppose $\lambda \in \C \setminus \R$ and $(A,g)$ is given as follows:

\begin{align}
A & = J_k(\lambda) \oplus J_k(\conj{\lambda}) & g & = S_k \oplus S_k
\end{align}

\noindent in coordinates $(x^1 , \dotsc, x^k , \conj{x}^1 , \dotsc, \conj{x}^k)$. Define real coordinates $(s^1 , t^1 , \dotsc, s^k, t^k)$ implicitly as follows:

\begin{subequations} \label{eq:transReToCom}
	\begin{equation}
	x^j = \frac{1}{\sqrt{2}}(s^j -i  t^j)
	\end{equation}
	\begin{equation}
	\conj{x}^j = \frac{1}{\sqrt{2}}(s^j + i  t^j)
	\end{equation}
\end{subequations}

These coordinates were chosen so that the pair $(A,g)$ are in the real metric-Jordan canonical form in the real coordinates $(s^1 , t^1 , \dotsc, s^k, t^k)$ after applying the appropriate tensor transformation law. 

In Cartesian coordinates $(x^i)$, we will use the convention that $x_i := g_{ij} x^j$; this is the only case where the Einstein summation convention is used in this section.

We now list some generic facts about tensors and C-tensors that will be used. We first present some facts about $\binom{1}{1}$-tensors. In the following proposition, we use the notation $C^p$ to denote the differentiability class of a geometric object, where $p \in  \N \cup \{ \infty , \omega \}$, and $C^\omega$ denotes the analytic class.

\begin{proposition}
	Suppose $T$ is a $\binom{1}{1}$-tensor of class $C^p$ and fix $q \in M$.
	
	Let $\lambda_0$ be a simple eigenvalue of $T_q$. Then there exists a neighborhood of $q$ in which $T$ has a simple eigenfunction $\lambda$ with a corresponding eigenvector field which are both of class $C^p$, and $\lambda(q) = \lambda_0$.
	
	If $T_q$ has simple eigenvalues, then there exists a neighborhood of $q$ in which $T$ has simple eigenfunctions of class $C^p$, and $T$ admits a basis of eigenvector fields of class $C^p$.
\end{proposition} 
\begin{proof}
	The proof is an application of the implicit function theorem (see, for example \cite[Theorems~10.2.1-10.2.4]{Dieudonne2008}). Details can be found in \cite{Kazdan1998}, see also \cite{Lax2007}.
\end{proof}

The above proposition shows that Benenti tensors necessarily locally admit a smooth basis of eigenvectors with corresponding smooth eigenfunctions. The following proposition gives necessary and sufficient conditions to determine when a given Benenti tensor is an IC-tensor.

\begin{proposition} \label{prop:eigCTs}
	Suppose $L$ is a Benenti tensor in a neighbourhood $U$ of a point $p$. If the eigenfunctions of $L$ are not constant in $U$, then the eigenfunctions are functionally independent, i.e. $L$ is an IC-tensor in a dense open subset of $U$.
\end{proposition}
\begin{proof}
	This is a direct consequence of the torsionless property of these tensors. Since in this case there are coordinates $(q^{i})$ such that $L$ is diagonal and each eigenfunction $u^{i}(q^i)$. Then
	
	\begin{align}
	\d u^{1} \wedge \cdots \wedge \d u^{n} & = \deriv{u^{1}}{q^{1}} \d q^{1} \wedge \cdots \wedge \deriv{u^{n}}{q^{n}} \d q^{n} \\
	& = (\prod\limits_{i=1}^{n} \deriv{u^{i}}{q^{i}}) \d q^{1} \wedge \cdots \wedge \d q^{n}
	\end{align}
	
	Hence if $\d u^i \neq 0$ for each $i$, the eigenfunctions are functionally independent. If the $u^{i}$ are analytic functions of $q^{i}$, then by assumption it follows that $L$ is an IC-tensor in a dense open subset of $U$.
\end{proof}

\begin{propMy} \label{prop:CTmetric}
	Suppose $L$ is an OCT and $p(z) = \det (z I - L)$ is its characteristic polynomial. Suppose $u^{i}$ is a simple eigenfunction of L and $\d u^{i} \neq 0$, then the corresponding eigenform is given by:
	
	\begin{equation}
	\d u^{i} = - \frac{(\d p)|_{z=u^{i}}}{p'(u^{i})}
	\end{equation}
	
	\noindent where $\d p$ is the exterior derivative of $p$ with respect to the ambient coordinates and $p'$ is the partial derivative of $p$ with respect to $z$. Furthermore if $L$ is an IC-tensor, then the metric in the coordinates induced by the eigenfunctions of $L$ is:
	
	\begin{equation}
	g^{ij} =
	\begin{cases}
	(p'(u^{i}))^{-2} \scalprod{(\d p)|_{z=u^{i}}}{(\d p)|_{z=u^{i}}} & \text{if } i = j \\
	0 & \text{else}
	\end{cases}
	\end{equation}
\end{propMy}
\begin{proof}
	Since  $p(z) = (z-u^{i}) f(z)$ for a smooth function $f(z)$. By taking the exterior derivative, we get:
	
	\begin{align}
	\d p &= -f \d u^{i} + (z-u^{i}) \d f
	\end{align}
	
	Then by L'Hopital's rule, we find that:
	
	\begin{align}
	(\d p)|_{z=u^{i}} &= -p'(u^{i}) \d u^{i}
	\end{align}
	
	\noindent which can be solved for $\d u^{i}$ since $u^{i}$ is a simple eigenfunction. The fact that $L \d u^i = u^i \d u^i$ follows from the fact that $L$ is torsionless.
	
	To calculate the metric, first it follows that $g^{ij} = 0$ when $i \neq j$ since $L$ is self-adjoint and has simple eigenfunctions. For the remaining component:
	
	\begin{align}
	g^{ii} & = \bp{\d u^{i},\d u^{i}} \\
	& = (p'(u^{i}))^{-2} \scalprod{(\d p)|_{z=u^{i}}}{(\d p)|_{z=u^{i}}}
	\end{align}
\end{proof}
\begin{remark}
	The assumption that $L$ is a concircular tensor can be replaced with any symmetric contravariant tensor whose associated endomorphism is torsionless.
\end{remark}

The following lemma on determinants will be used several times.

\begin{lemma} \label{lem:detAplusTen}
	Suppose $T = A + v \otimes x$ where $A = [a_{1},...,a_{n}]$ is an $n \times n$ matrix, $v \in \Fi^{n}$ and $x \in \Fi^{n}$ (where $\Fi$ is $\R$ or $\C$). Then $\det T$ is given as follows:
	
	\begin{equation}
	\det T = \bigwedge \limits_{i=1}^{n} (a_{i} + x_{i}v) = \bigwedge \limits_{i=1}^{n} a_{i} + 
	\sum\limits_{i=1}^{n} a_{1} \wedge \cdots \wedge x_{i} v \wedge \cdots \wedge a_{n}
	\end{equation}
\end{lemma}
\begin{proof}
	The formula clearly holds for $n=1$, so inductively suppose the formula holds for $k = n-1$, then:
	
	\begin{align}
	\bigwedge \limits_{i=1}^{n} (a_{i} + x_{i}v) & = \bigwedge \limits_{i=1}^{n-1} (a_{i} + x_{i}v) \wedge (a_{n} + x_{n} v) \\
	& = (\bigwedge \limits_{i=1}^{n-1} a_{i} + 
	\sum\limits_{i=1}^{n-1} a_{1} \wedge \cdots \wedge x_{i} v \wedge \cdots \wedge a_{n-1}) \wedge (a_{n} + x_{n} v) \\
	& = \bigwedge \limits_{i=1}^{n} a_{i} + \sum\limits_{i=1}^{n-1} a_{1} \wedge \cdots \wedge x_{i} v \wedge \cdots \wedge a_{n} + \bigwedge \limits_{i=1}^{n-1} a_{i} \wedge x_{n} v \\
	& = \bigwedge \limits_{i=1}^{n} a_{i} + \sum\limits_{i=1}^{n} a_{1} \wedge \cdots \wedge x_{i} v \wedge \cdots \wedge a_{n}
	\end{align}
\end{proof}

In the following sections, we will obtain the following information. First we will calculate the characteristic polynomial for CTs in spaces of constant curvature. Using this, for ICTs we will calculate the transformation from the canonical coordinates they induce to Cartesian coordinates, and we will calculate the metric in canonical coordinates.

\subsection{Central Concircular tensors} \label{sec:cTIrredCC}

The following general lemma will be used to calculate the characteristic polynomial of central CTs.

\begin{lemma}[Determinant of Central Concircular tensors] \label{lem:detConTen}
	Suppose $L = A + r \otimes r^{\flat} $ is a central Concircular tensor, where $r^{i} = x^{i}$. Then,
	
	\begin{equation} \label{eq:CTDetCent}
	\det L = \bigwedge \limits_{i=1}^{n} a_{i} +
	\sum\limits_{i=1}^{n} a_{1} \wedge \cdots \wedge x_{i} r \wedge \cdots \wedge a_{n}
	\end{equation}
	
	Suppose U is a non-degenerate A-invariant subspace (hence $U^{\perp}$ is A-invariant), let $L_{u} = L|_{U}$ and $L_{u^{\perp}} = L|_{U^{\perp}}$, then:
	
	\begin{equation} \label{eq:CTDetCentInvSS}
	\det L = \det L_{u} \det A_{u^{\perp}} + \det A_{u} (\det L_{u^{\perp}} - \det A_{u^{\perp}})
	\end{equation}
\end{lemma}
\begin{proof}
	The first statement follows from Lemma~\ref{lem:detAplusTen} by taking $A \rightarrow A$, $r \rightarrow v$ and $r^{\flat} \rightarrow x$.
	
	Now for the second part, let $k = \dim U$, then in a basis adapted to the decomposition $V = U \obot U^{\perp}$, we have:
	
	\begin{equation}
	A =
	\begin{pmatrix}
	B & 0 \\ 
	0 & C
	\end{pmatrix} 
	\end{equation}
	
	\noindent where $B$ is a $k \times k$ matrix and $C$ is a $(n-k) \times (n-k)$ matrix. Furthermore $r = r_{b} + r_{c}$ where $r_{b} \in U$ and $r_{c} \in U^{\perp}$. The main fact we use is that for any square matrix, $T$, of the form:
	
	\begin{equation}
	\begin{pmatrix}
	A & B \\ 
	0 & C
	\end{pmatrix} 
	\end{equation}
	
	\noindent we have $ \det T = \det A \det C$. Thus:
	
	\begin{align}
	\det L & = \bigwedge \limits_{i=1}^{n} a_{i} +
	\sum\limits_{i=1}^{n} a_{1} \wedge \cdots \wedge x_{i} r \wedge \cdots \wedge a_{n} \\
	& = \bigwedge \limits_{i=1}^{k} b_{i} \wedge \bigwedge \limits_{i=1}^{n-k} c_{i} + (\sum\limits_{i=1}^{k} b_{1} \wedge \cdots \wedge x_{i} r_{b} \wedge \cdots \wedge b_{k}) \wedge \bigwedge \limits_{i=1}^{n-k} c_{i} \\
	& \qquad {} + \bigwedge \limits_{i=1}^{k} b_{i} \wedge (\sum\limits_{i=1}^{n-k} c_{1} \wedge \cdots \wedge x_{i} r_{c} \wedge \cdots \wedge c_{n-k}) \\
	& = (\bigwedge \limits_{i=1}^{k} b_{i} + \sum\limits_{i=1}^{k} b_{1} \wedge \cdots \wedge x_{i} r_{b} \wedge \cdots \wedge b_{k}) \wedge \bigwedge \limits_{i=1}^{n-k} c_{i} \\
	& \qquad {} + \bigwedge \limits_{i=1}^{k} b_{i} \wedge (\sum\limits_{i=1}^{n-k} c_{1} \wedge \cdots \wedge x_{i} r_{c} \wedge \cdots \wedge c_{n-k}) \\
	& = \det L_{u} \det A_{u^{\perp}} + \det A_{u} (\det L_{u^{\perp}} - \det A_{u^{\perp}})
	\end{align}
\end{proof}

Now consider the simplest case where $A = \diag(\lambda_{1},...,\lambda_{n})$. Then \cref{eq:CTDetCent} can be used to get the characteristic polynomial of L, which is:

\begin{equation} \label{eq:CTChPolCCI}
p(z) = \det (z I - L) = \prod\limits_{i=1}^{n} (z - \lambda_{i}) - \sum\limits_{i=1}^{n} x_{i}x^{i} \prod\limits_{j \neq i} (z - \lambda_{j})
\end{equation}

Now suppose $L$ is an ICT with eigenfunctions $(u^1,\dotsc, u^n)$, then from the above equation we have:

\begin{equation}
\prod\limits_{j=1}^{n}(u^{j} - \lambda_{i}) = p(\lambda_i) = - \varepsilon_{i} (x^{i})^{2} \prod\limits_{j \neq i} (\lambda_{i} - \lambda_{j})
\end{equation}

One can check that by assumption we must have $\lambda_i \neq \lambda_j$ if $i \neq j$. This will eventually be proven later. Thus we deduce the transformation from the coordinates $(u^1,\dotsc, u^n)$ to Cartesian coordinates to be:

\begin{equation} \label{eq:CTCoordsCCI}
(x^{i})^{2} = \varepsilon_{i} \frac{\prod\limits_{j=1}^{n}(u^{j} - \lambda_{i})}{\prod\limits_{j \neq i} (\lambda_{j} - \lambda_{i})}
\end{equation}

The derivation of the transformation to Cartesian coordinates follows that of \cite[section~5]{Crampin2003}. We will use this method for all other types of CTs as well. Now, it will be useful to write the characteristic polynomial in standard form:

\begin{propMy} \label{cor:CTChPolCent}
	Suppose $L$ is a central CT with parameter matrix $A = \diag(\lambda_{1},...,\lambda_{n})$ and arbitrary orthogonal metric. Write the characteristic polynomial of $A$ as:
	\begin{equation}
	B(z) = \det(zI - A) = \sum_{l=0}^{n} a_{l}z^{l}
	\end{equation}
	Then the characteristic polynomial of $L$ is:
	\begin{equation} \label{eq:CTChPolCCIExp}
	p(z) = \det(zI - L) = \sum_{l=0}^{n} (a_{l} - \sum_{j=0}^{n-1-l} a_{j+1+l} \bp{r, A^{j}r})z^{l}
	\end{equation}
\end{propMy}
\begin{proof}
	We will prove this formula by expanding \cref{eq:CTChPolCCI}. For the following calculations, if $a(z)$ is a polynomial in $z$, then $[z^{l}] a(z)$ is the coefficient of $z^l$ in this polynomial. First observe that
	
	\begin{align}
	[z^{l}] \prod\limits_{j} (z- \lambda_{j}) & = [z^{l}] [z\prod\limits_{j \neq i} (z- \lambda_{j}) - \lambda_{i}\prod\limits_{j\neq i} (z- \lambda_{j})] \\
	& =  [z^{l-1}]\prod\limits_{j \neq i} (z- \lambda_{j}) - \lambda_{i} [z^{l}]\prod\limits_{j\neq i} (z- \lambda_{j}) \\
	\Rightarrow & [z^{l-1}]\prod\limits_{j \neq i} (z- \lambda_{j}) = [z^{l}] \prod\limits_{j} (z- \lambda_{j}) + \lambda_{i} [z^{l}]\prod\limits_{j\neq i} (z- \lambda_{j})
	\end{align}
	
	We also have
	
	\begin{equation}
	[z^{n-1}]\prod\limits_{j \neq i} (z- \lambda_{j}) = 1
	\end{equation}
	
	We will prove inductively that
	
	\begin{align}
	[z^{l}]\prod\limits_{j \neq i} (z- \lambda_{j}) & = \sum_{j=0}^{n-1-l} \lambda_{i}^{j} a_{j+1+l}
	\end{align}
	
	Then by inductive hypothesis, we have
	
	\begin{align}
	[z^{l-1}]\prod\limits_{j \neq i} (z- \lambda_{j}) & = a_{l} + \lambda_{i} \sum_{j=0}^{n-1-l} \lambda_{i}^{j} a_{j+1+l} \\
	& = a_{l} + \sum_{j=1}^{n-l} \lambda_{i}^{j} a_{j+l} \\
	& = \sum_{j=0}^{n-l} \lambda_{i}^{j} a_{j+l}
	\end{align}
	
	Then
	
	\begin{align}
	[z^{l}]\sum\limits_{i=1}^{n} x_i x^{i} \prod\limits_{j \neq i} (z- \lambda_{j}) & = \sum\limits_{i=1}^{n} g_{ii}(x^{i})^{2} [z^{l}]\prod\limits_{j \neq i} (z- \lambda_{j}) \\
	& = \sum\limits_{i=1}^{n} g_{ii}(x^{i})^{2} \sum_{j=0}^{n-1-l} \lambda_{i}^{j} a_{j+1+l} \\
	& = \sum_{j=0}^{n-1-l} a_{j+1+l} \sum\limits_{i=1}^{n} g_{ii}(x^{i})^{2}\lambda_{i}^{j} \\
	& = \sum_{j=0}^{n-1-l} a_{j+1+l} \bp{r, A^{j}r}
	\end{align}
	
	Which together with \cref{eq:CTChPolCCI} proves the proposition.
\end{proof}

In the following theorem we collect a useful limiting procedure for dealing with Jordan blocks. It has been proven by \citeauthor{Kalnins1984} in \cite{Kalnins1984} for general dimensions. We have independently verified it only for dimensions less than three. The details of this verification are only partially included in the following proof, which can be omitted without loss of continuity.

\begin{theorem}[\cite{Kalnins1984}] \label{thm:KalLim}
	Let $A_0 := J_{n}^T(\lambda_1)$ and $g_0 := \varepsilon S_n$. For $n \leq 3$, there exists a sequence of diagonal matrices $A := \diag(\lambda_1,\dotsc,\lambda_n)$, $g := \diag(a_1,\dotsc,a_n)$ and transformation matrices $\Lambda$ such that
	
	\begin{align}
	\Lambda^{-1} A \Lambda & \rightarrow A_0 & \Lambda^T g \Lambda & \rightarrow g_0
	\end{align}
\end{theorem}
\begin{proof}
	First consider the following definitions:
	\begin{align} 
	?\Lambda^i_j?  & := \epsilon^{j-1}_{i+1-j} = \prod_{l=2}^{j}(\epsilon_{i-1}^{1}- \epsilon_{l-2}^{1}) &
	\epsilon^{k}_{l} & := \begin{cases}
	0 & \text{if } l \leq 0 \\
	1 & \text{if } k \leq 0
	\end{cases}
	\label{eq:limproc} \\
	a_{i} & := \frac{\varepsilon}{\prod_{k \neq i} (\epsilon^{1}_{i-1}- \epsilon^{1}_{k-1})}
	\end{align}
	
	Note that $\epsilon^{k}_{l}$ is of order $k$ if $k,l > 0$. Finally let $\lambda_i := \lambda_1 + \epsilon_{i-1}^{1}$. Then the conclusion follows by direct calculation if for each $i = 2,\dotsc,n$, $\epsilon_{i}^{1} \rightarrow 0$.
\end{proof}

Now suppose $L$ is a central CT with parameter matrix $A = J_k^T(0)$. We will use the above theorem to obtain this CT as a limit of central CTs with parameter matrix $A = \diag(0,\lambda_2,\dotsc,\lambda_k)$. The characteristic polynomial of these CTs is given by \cref{eq:CTChPolCCIExp}. In order to obtain the characteristic polynomial for a CT with $A = J_k^T(0)$ we will use the fact that the characteristic polynomial of $J_k^T(0)$ is $z^k$. Then starting with $A = \diag(0,\lambda_2,\dotsc,\lambda_k)$, by \cref{eq:CTChPolCCIExp} we have:

\begin{align}
p(z) & = \sum_{l=0}^{k} (a_{l} - \sum_{j=0}^{k-1-l} a_{j+1+l} \bp{r, A^{j}r})z^{l} \\
& \rightarrow z^k - \sum_{l=0}^{k-1} \bp{r, A^{k-1-l}r} z^{l} \\
& = z^k - \sum_{l=0}^{k-1} \bp{r, A^{k-1-l}r} z^{l} \\
& = z^k - \varepsilon \sum_{l=0}^{k-1} \sum_{i=1}^{l+1} x^i x^{l + 2-i} z^{l}
\end{align}


Thus we have proven part of the following:

\begin{propMy} \label{cor:CTChPolCentDe}
	Suppose $L$ is a central CT with parameter matrix $A = J_k^T(0)$ and  metric $g = \varepsilon S_k$.
	Then the characteristic polynomial of $L$ is:
	\begin{equation}
	p(z) = \det(zI - L) = z^k - \varepsilon \sum_{l=0}^{k-1} \sum_{i=1}^{l+1} x^i x^{l + 2-i} z^{l}
	\end{equation}
	
	Furthermore the following are true:
	\begin{itemize}
		\item $L$ has no constant eigenfunctions.
		\item If $T(z) = \frac{p(z)}{B(z)}$ and $k \leq 3$, then $\bp{\d T, \d T}  = 4\deriv{}{z} T(z)$
	\end{itemize}
\end{propMy}
\begin{proof}
	We first prove the case where $A$ is a real Jordan block. To prove that $L$ has no constant eigenfunctions, we differentiate an equation preceding this proposition to obtain:
	
	\begin{equation}
	\nabla p = - 2 \sum_{l=0}^{k-1} z^{l} A^{k-1-l}r
	\end{equation}
	
	
	\noindent from which we see that $\bp{e_k, \nabla p} = - 2 \varepsilon z^{k-1} x^1$. Thus $L$ cannot have a constant eigenfunction. The equation for $\bp{\d T, \d T}$ is proven as follows. When $A = \diag(0,\lambda_2,\dotsc,\lambda_k)$ one can easily prove the formula using \cref{eq:CTChPolCCI}. Then the formula for $A = J_k^T(0)$ follows by applying the limiting technique in \cref{thm:KalLim} used above. Finally, for the case of a complex Jordan block, i.e. $A = J_k^T(\lambda)$ where $\lambda \in \C$, note that these proofs hold by replacing $A \rightarrow A - \lambda I$ and $z \rightarrow z + \lambda$.
\end{proof}

Now one can use the second part of \cref{lem:detConTen} to obtain the characteristic polynomial of any central CT in $\eunn$. Indeed, suppose $L$ is a central CT with parameter matrix

\begin{align}
A & = J_k^T(0) \oplus \diag(\lambda_{k+1},\dotsc,\lambda_n) & g & = \varepsilon_0 S_k \oplus \diag(\varepsilon_{k+1},\dotsc,\varepsilon_n)
\end{align}

We can apply \cref{lem:detConTen} with $U$ equal to the subspace corresponding to $J_k^T(0)$, then

\begin{align}
p(z) = \det (zI - L) & = (\prod\limits_{i=k+1}^{n} y_{i}) \left (z^{k} - \varepsilon_0  \sum\limits_{l=0}^{k-1} \left (\sum\limits_{i=1}^{l+1} x^{i} x^{l+2-i} \right ) z^{l} \right ) \label{eq:CTChPolCCIII} \\
& \qquad {} - z^{k}(\sum\limits_{i=k+1}^{n} x_{i}x^{i} \prod\limits_{j=k+1, j \neq i}^{n} y_{j})
\end{align}

When $L$ is an ICT, we can obtain a transformation from canonical coordinates to Cartesian coordinates. Our formula is motivated by one in \cite{Kalnins1984} and is given as follows:

\begin{subequations} \label{eq:CTCoordsCCII}
	\begin{equation} \label{eq:CTCoordsCCIIa}
	\sum\limits_{i=1}^{l+1} x^{i} x^{l+2-i} = \frac{- \varepsilon_0 }{l! } (\deriv{}{z})^{l}(\frac{p(z)}{B_{u^\perp}(z)}) \big \lvert_{z = 0}  \quad l = 0,\dotsc, k-1
	\end{equation}
	\begin{align}
	(x^{i})^{2}  & =  -\varepsilon_i \frac{p(\lambda_i)}{B'(\lambda_i)} & i = k+1,...,n \label{eq:CTCoordsCCIIb}
	\end{align}
\end{subequations}

The following lemma will be used to obtain the metric in canonical coordinates adapted to an ICT defined in a space of constant curvature.

\begin{lemma} \label{lem:CctTForm}
	Suppose $L$ is a central CT with parameter matrix $A$. Let
	
	\begin{equation}
	T(z) = \frac{p(z)}{B(z)}
	\end{equation}
	
	Then $\bp{\d T, \d T} = 4\deriv{}{z} T(z)$.
\end{lemma}
\begin{proof}
	We prove this by induction. The base cases are given by \cref{cor:CTChPolCentDe}. Suppose $U$ is a non-degenerate invariant subspace of $A$ such that $L_u$ has the form given by \cref{cor:CTChPolCentDe} and $U^\perp$ satisfies the induction hypothesis. 
	
	By \cref{eq:CTDetCentInvSS} we can write:
	
	\begin{equation}
	p(z) = p_u(z)B_{u^\perp}(z) + B_{u}(z)(p_{u^\perp}(z) - B_{u^\perp}(z))
	\end{equation}
	
	Then
	
	\begin{equation}
	\d p = B_{u^\perp} \d p_u + B_{u} \d p_{u^\perp}
	\end{equation}
	
	Thus from the above equation, we have:
	
	\begin{align}
	\frac{\d p}{B} & = \frac{\d p_u}{B_u} + \frac{\d p_{u^\perp}}{B_{u^\perp}} \\
	\Rightarrow \d T & = \d T_u  + \d T_{u^\perp} \\
	\Rightarrow \bp{\d T, \d T} & = \bp{\d T_u, \d T_u} + \bp{\d T_{u^\perp}, \d T_{u^\perp}} \\
	& = 4\deriv{}{z} T_u(z) + 4\deriv{}{z} T_{u^\perp}(z) \\
	& = 4\deriv{}{z} T(z)
	\end{align}
\end{proof}

\paragraph{Examples } We end this section with some separable coordinate systems induced by central ICTs which can be analyzed fairly easily. These examples are a natural generalization of those presented in \cite[section~5]{Crampin2003}.

\begin{example}[Generalization of elliptic coordinates to $\eunn$] \label{ex:genElipCoord}
	Our first example is the central CT in $\eunn$ with parameter matrix $A = \diag(\lambda_{1},...,\lambda_{n})$ and orthogonal metric $g = (-1,\dotsc,-1,1,\dotsc,1)$. This CT is easiest to analyze if we assume $\lambda_{1} < \lambda_{2} < \cdots < \lambda_{n}$. Recall from \cref{eq:CTChPolCCI}, that the characteristic polynomial of $L$ is:
	
	\begin{equation}
	p(z) = \det (z I - L) = \prod\limits_{i=1}^{n} y_{i} - \sum\limits_{i=1}^{n} x_{i}x^{i} \prod\limits_{j \neq i} y_{j}
	\end{equation}
	
	Using the above formula, one can show that $L$ has no constant eigenfunctions (see the proof of \cref{cor:CTChPolCentDe}). Then by \cref{prop:eigCTs}, this CT is an ICT near any point where the eigenfunctions of $L$ are simple. We will now show that $L$ is an ICT in a dense subset of $\eunn$. First note that
	
	\begin{equation} \label{eq:chPolCCIsgn}
	p(\lambda_{i}) = -\varepsilon_{i}(x^{i})^{2} \prod\limits_{j \neq i} (\lambda_{i} - \lambda_{j})
	\end{equation}
	
	Assume each $x^{i} \neq 0$, then from Equation~\ref{eq:chPolCCIsgn}, we find that $\sgn p(\lambda_{i}) = \varepsilon_{i} (-1)^{n+1-i}$. Also since the coefficient of leading degree of $p(z)$ is $z^{n}$, we find that $\lim\limits_{z \rightarrow \infty} p(z) = 1$ and $\lim\limits_{z \rightarrow -\infty} p(z) = (-1)^{n}$. Since by assumption we have that $\varepsilon_{n} = 1$, we can use the intermediate value theorem to deduce the following about the roots of $p(z)$. If $\nu = 0$ (i.e. in Euclidean space), there are n distinct roots $u^{1},...,u^{n}$ satisfying:
	
	\begin{equation}
	\lambda_{1} < u^{1} < \lambda_{2} < u^{2} \cdots < \lambda_{n} < u^{n}
	\end{equation}
	
	If $\nu > 0$ then there are n distinct roots $u^{1},...,u^{n}$ satisfying:
	
	\begin{equation} \label{eq:CCIeigdom}
	u^{1} < \lambda_{1} < u^{2} \cdots < u^{\nu} < \lambda_{\nu} < \lambda_{\nu+1} < u^{\nu+1} < \lambda_{\nu+2} < u^{\nu+2} \cdots < \lambda_{n} < u^{n}
	\end{equation}
	
	Hence $L$ is an IC-tensor on an open dense subset of $\eunn$; because of this property one could consider the induced separable coordinates to be a generalization of elliptic coordinates. Since $p(\lambda_{i}) = \prod\limits_{j=1}^{n}(\lambda_{i}-u^{j})$, by Equation~\eqref{eq:chPolCCIsgn}, we can obtain the Cartesian coordinates in terms of the separable coordinates $u^{1},...,u^{n}$
	
	\begin{equation}
	(x^{i})^{2} = \varepsilon_{i} \frac{\prod\limits_{j=1}^{n}(u^{j} - \lambda_{i})}{\prod\limits_{j \neq i} (\lambda_{j} - \lambda_{i})}
	\end{equation}
	
	By using \cref{eq:CCIeigdom} and \cref{prop:IctMetricEunn}, one can check that in the separable coordinates $(u^1,\dotsc,u^n)$, for $1  \leq i \leq \nu$, $\sgn g^{ii} = \frac{(-1)^{n-i+1}}{(-1)^{n-i}} = -1$. Hence $\partial_{1},\dotsc,\partial_\nu$ are timelike vector fields and the remaining ones are spacelike.
\end{example}

We now show that if we relax the condition that $ \lambda_{1} < \cdots < \lambda_{n}$ in the above example then the coordinate system may no longer be defined on a dense subset of $\eunn$. Although one should note that in $\E^n$ that condition was not restrictive. The simplest case occurs in $\E^2_1$.

\begin{example} \label{ex:CCTCoordMink}
	Consider a central CT $L$ in $\E^2_1$ with parameter matrix $A = \diag(\lambda_{1},\lambda_{2})$ where $\lambda_1 > 
	\lambda_2$ and orthogonal metric $g = \diag(-1,1)$. Denote Cartesian coordinates by $(t,x)$. In this case the characteristic polynomial of $L$, $p(z)$, given by \cref{eq:CTChPolCCIExp} reduces to:
	
	\begin{equation}
	p(z) = z^2 + (2(t^2 - x^2 ) - \lambda_1 - \lambda_2)z - 2 t^2 \lambda_2 + 2 x^2 \lambda_1 + \lambda_1 \lambda_2
	\end{equation}
	
	One can calculate the discriminant of this polynomial to be:
	
	\begin{equation}
	4\, \left(  \left( t -x \right)^{2} + \frac{\lambda_2 - \lambda_1}{2} \right)  \left(  \left( t + x \right)^{2} + \frac{\lambda_2 - \lambda_1}{2} \right)
	\end{equation}
	
	If we define new Cartesian coordinates $(y^1,y^2 )$ by: 
	
	\begin{align}
	y^1 & := \sqrt{2}(t-x) & y^2 & := \sqrt{2}(t+x)
	\end{align}
	
	and we let $e := \sqrt{\lambda_1 - \lambda_2}$, then $L$ is a Benenti tensor on the following connected regions:
	
	\begin{center}
		\begin{tabular}{|c|c|}
			\hline Region & $(u^1 , u^2)$ \\ 
			\hline N & $y^1 > e , y^2 < - e$ \\  
			\hline E & $y^1, y^2 > e$ \\ 
			\hline S & $y^1 < -e , y^2 > e$ \\ 
			\hline W & $y^1 , y^2 > -e$ \\ 
			\hline C & $\Abs{y^1},\Abs{y^2} < e$ \\
			\hline 
		\end{tabular}
	\end{center}
	
	%
	%
	
	Hence the regions are separated by the lightlike lines $\Abs{y^i} = e$. Thus as claimed the associated separable coordinate systems aren't defined on a dense subset. 
	
	
	One can also find the coordinate domains as follows. Suppose $L$ is an ICT with eigenfunctions $u^1 < u^2$. Then by requiring that the metric in these coordinates given by \cref{prop:IctMetricEunn} to be Lorentzian, one finds the following constraints:
	
	\begin{subequations} \label{eq:constEqns}
		\begin{align}
		u^1 < u^2 < \lambda_2 < \lambda_1 \\
		\lambda_2 < \lambda_1 < u^1 < u^2 \\
		\lambda_2 < u^1 < u^2 < \lambda_1 \\
		u^1 < \lambda_2 < \lambda_1 < u^2
		\end{align}
	\end{subequations}
	
	The above inequalities shown that in the subset where $L$ is a Benenti tensor, if the eigenfunctions transition from one coordinate domain to another then one of the eigenfunctions must take the value $\lambda_1$ or $\lambda_2$. Hence the transition manifolds are solutions of $p(\lambda_i) = 0$, i.e. by \cref{eq:CTChPolCCI} where $(x^i)^2 = 0$. In this case, the eigenfunctions of $L$ can be readily calculated:
	
	\begin{align}
	t = 0 & \Rightarrow \lambda_1 , \lambda_2 + x^2 \\
	x = 0 & \Rightarrow \lambda_1 - t^2, \lambda_2
	\end{align}
	
	Using the values of the eigenfunctions on these subsets and their possible ranges given in \cref{eq:constEqns} one can deduce the following:
	
	\begin{center}
		\begin{tabular}{|c|c|}
			\hline $(y^1 , y^2)$ & $(u^1 , u^2)$ \\ 
			\hline E, W & $u^1 < u^2 < \lambda_2 < \lambda_1$ \\ 
			\hline N,S  & $\lambda_2 < \lambda_1 < u^1 < u^2$ \\ 
			\hline C  & $\lambda_2 < u^1 < u^2 < \lambda_1$ \\ 
			\hline 
		\end{tabular}
	\end{center}
	
	Together with \cref{eq:CTCoordsCCI}, this completes the analysis of these coordinate systems.
\end{example}

Even in three dimensions, the above analysis becomes much more difficult. This is because in three dimensions one can show that the discriminant is an eight degree polynomial in the coordinates with many terms. However, we note two simplifications that could be made for the general case. First by transferring to a geometrically equivalent CT, we could assume one of the eigenvalues of $A$ is zero. Secondly since the characteristic polynomial of $L$, given by \cref{eq:CTChPolCCI} only depends on the quantities $(x^i)^2$ and not $x^i$ explicitly, one can restrict the analysis to the quadrant where each $x^i > 0$ without losing generality. This symmetry is a consequence of the non-uniqueness of the chosen basis, in particular due to the fact that if $v$ is an eigenvector of $A$ then so is $-v$.


\subsection{Axial Concircular tensors} \label{sec:cTIrredAC}

\begin{propMy} \label{prop:CTDetIrredAxial}
	Let $L$ be an axial CT with parameter matrix $A  = J_{k}(0)^T$ and metric $g = \varepsilon S_k$. Then
	
	\begin{equation} \label{eq:CTDetIrredAxial}
	p(z) = \det (z I - L) = z^{k}+\sum\limits_{l=2}^{k} \sum\limits_{i=1}^{l-1} x^{k+1+i-l}x^{k+1-i} z^{k-l} - 2 \varepsilon \sum\limits_{i=1}^{k} x^{k-i+1} z^{k-i}
	\end{equation}
	
	Furthermore the following are true:
	\begin{itemize}
		\item $L$ has no constant eigenfunctions.
		\item If $k \leq 3$, then $\bp{\d p, \d p}  = 4 \varepsilon \deriv{}{z} p(z)$.
	\end{itemize}
\end{propMy}
\begin{proof}
	We first outline how one proves the above formula for $p(z)$. It is sufficient to calculate $\det L$ when $L$ has the parameter matrix $A  = J_{k}(\lambda)^T$. Let $\tilde{A} = [\tilde{a_{1}},...,\tilde{a_{n}}] := A + \varepsilon r \otimes e_{k}$. Then applying \cref{lem:detAplusTen} to $L = \tilde{A} + e_1 \otimes r^\flat$ gives:
	
	\begin{equation}
	\det L = \bigwedge \limits_{i=1}^{n} \tilde{a_{i}} +
	\sum\limits_{i=1}^{n} \tilde{a_{1}} \wedge \cdots \wedge x_{i} e_{1} \wedge \cdots \wedge \tilde{a_{n}}
	\end{equation}
	
	After expanding $r$ and $e_1$ in the basis $\{a_1,\dotsc,a_k\}$ and simplifying, the result then follows by a straightforward but tedious calculation.
	
	Suppose the above formula for $p(z)$ holds. We now show that $L$ has no constant eigenfunctions. The constant term of $\d p$ is:
	
	\begin{equation}
	-2 \varepsilon \sum\limits_{i=1}^{k} z^{k-i} \d x^{k-i+1}
	\end{equation}
	
	If $\lambda \in \R$ satisfies $p(\lambda) \equiv 0$, then the above form must be identically zero. A contradiction, hence $L$ has no constant eigenfunctions.
	
	The formula involving $\bp{\d p, \d p}$ can be checked manually for the cases $k \leq 3$.
\end{proof}

%

The following proposition will reduce the calculation of the characteristic polynomial for general axial concircular tensors to cases already considered.

\begin{propMy}[Determinant of Axial Concircular tensors] \label{prop:CTDetAxial}
	Suppose $L$ is an axial CT in canonical form given as follows:
	
	\begin{align}
	L & = A + e_{1} \otimes r^{\flat} + r \otimes e_{1}^{\flat} \\
	A & = A_d \oplus A_c
	\end{align}
	
	\noindent where $A_d = J_{k}^T(\lambda)$. Then $p(z) = \det(z I - L)$ is given as follows:
	
	\begin{equation} \label{eq:CTChPolAxial}
	p(z) = p_d(z) B(z) + \varepsilon (p_c(z) - B(z)) 
	\end{equation}
\end{propMy}
\begin{proof}
	First note that it is sufficient to calculate $\det L$. Write $r = r_d + r_c$ adapted to the decomposition $\eunn = D \obot D^\perp$ where $D$ is the $A$-invariant subspace generated by $e_1$. Then
	
	\begin{equation}
	L = L_d + A_c + e_{1} \otimes (r_c)^{\flat} + r_c \otimes e_{1}^{\flat} 
	\end{equation}
	
	\noindent where $L_d$ is $L$ restricted to $D$ and $A_c$ is $A$ restricted to $D^\perp$. Let $\tilde{L} = L_d + A_c + e_{1} \otimes (r_c)^{\flat}$, then applying \cref{lem:detAplusTen} to $L = \tilde{L} + \varepsilon r_c \otimes e_k$ gives:
	
	\begin{equation} \label{eq:prop:CTDetAxI}
	\det L = \det \tilde{L} + \varepsilon \tilde{L}_{1} \wedge \cdots \wedge r_c \wedge \cdots \wedge \tilde{L}_{n} 
	\end{equation}
	
	\noindent where $r_c$ appears in the $k$th position. Note that in block diagonal form
	
	\begin{equation}
	\tilde{L} = \begin{pmatrix}
	L_d & e_1 \otimes (r_c)^\flat \\ 
	0 & A_c
	\end{pmatrix} 
	\end{equation}
	
	Then after applying \cref{lem:detAplusTen} once more, we get
	
	\begin{align}
	\tilde{L}_{1} \wedge \cdots \wedge r_c \wedge \cdots \wedge \tilde{L}_{n}  & = \bigwedge \limits_{i=1}^{k-1} (L_d)_i \wedge r_c \wedge (\sum\limits_{i=k+1}^{n} a_{k+1} \wedge \cdots \wedge x_{i} e_{1} \wedge \cdots \wedge a_{n}) \\
	& = - \bigwedge \limits_{i=1}^{k-1} (L_d)_i \wedge e_1\wedge (\sum\limits_{i=k+1}^{n} a_{k+1} \wedge \cdots \wedge x_{i} r_c \wedge \cdots \wedge a_{n}) \\
	& = - \bigwedge \limits_{i=1}^{k-1} a_i \wedge e_1 \wedge (\sum\limits_{i=k+1}^{n} a_{k+1} \wedge \cdots \wedge x_{i} r_c \wedge \cdots \wedge a_{n}) \\
	& = (-1)^k  e_1 \wedge \cdots \wedge e_k \wedge (\sum\limits_{i=k+1}^{n} a_{k+1} \wedge \cdots \wedge x_{i} r_c \wedge \cdots \wedge a_{n}) \\
	& = (-1)^k (\det (L_c) - \det (A_c))
	\end{align}
	
	\noindent where the second last equation follows by expanding $e_1$ in the basis $\{ a_1,\dotsc,a_k \}$. The result then follows by \cref{eq:prop:CTDetAxI}.
\end{proof}

One can use \cref{prop:CTDetAxial} to obtain the characteristic polynomial of any axial CT in $\eunn$. This is done as in the example in the discussion following \cref{cor:CTChPolCentDe}. As an example, we will calculate the Cartesian coordinates for a non-null axial CT (i.e. $k = 1$). Indeed, suppose $L$ is a non-null axial ICT with eigenfunctions $(u^1,\dotsc, u^n)$. Let $A_c = \diag(\lambda_2,\dotsc,\lambda_{n})$, then from \cref{eq:CTChPolAxial} and \cref{eq:CTDetIrredAxial}, we see that

\begin{equation}
p(z) = \det (z I - L) = (\prod\limits_{i=2}^{n} y_{i})(z - 2 \varepsilon x^{1}) - \varepsilon (\sum\limits_{i=2}^{n} x_{i}x^{i} \prod\limits_{j=2, j \neq i}^{n} y_{j})
\end{equation}

\noindent where $y_i = z - \lambda_i$. Since $p(z) = \prod\limits_{i=1}^{n} (z-u^i)$, we can deduce the transformation from the coordinates $(u^1,\dotsc, u^n)$ to Cartesian coordinates as follows. By evaluating $p(\lambda_i)$, we get 

\begin{align}
(x^{i})^{2} & = -\varepsilon_i \varepsilon \frac{\prod\limits_{j=1}^{n} (u^{j}-\lambda_{i})}{\prod\limits_{j \geq 2, j \neq i} (\lambda_{j}-\lambda_{i})} & i = 2,...,n \label{eq:CTCoordsACI}
\end{align}

By taking the coefficient of $z^{n-1}$ of $p(z)$, we get:

\begin{equation} \label{eq:CTCoordsACII}
x^{1} = \frac{\varepsilon}{2}(u^{1}+ \cdots + u^{n} -\lambda_2 - \cdots - \lambda_n)
\end{equation}

In conclusion, we note that this procedure can be generalized for $k \geq 2$.

Observe that \cref{eq:CTChPolAxial} holds for a central CT if we define $p_d(z) \equiv 1$ in this case. We will use \cref{eq:CTChPolAxial} and \cref{lem:CctTForm} to obtain the metric in canonical coordinates for some ICTs in $\eunn$. We have the following:

\begin{proposition}[ICT metrics in $\eunn$] \label{prop:IctMetricEunn}
	Suppose $L$ is an ICT in Euclidean or Minkowski space in canonical form with eigenfunctions $(u^1,\dotsc,u^n)$. Then the metric in adapted coordinates is orthogonal and
	
	\begin{equation}
	g_{ii} = \frac{\varepsilon}{4} \frac{p'(u^i)}{B(u^i)} = \frac{\varepsilon}{4}\frac{\prod\limits_{j \neq i}(u^{i}-u^{j})}{\prod\limits_{j=1}^{n-k} (u^{i}-\lambda_{j})}
	\end{equation}
	
	\noindent where $\varepsilon$ is the sign associated with $L$ and $\lambda_{1},\dotsc,\lambda_{n-k}$ are the roots of $B(z)$.
\end{proposition}
\begin{remark}
	The above formula likely holds in general (see \cite{Kalnins1984}) but we haven't verified it for null axial CTs when $k > 3$.
\end{remark}
\begin{proof}
	
	Let $T(z) := \frac{p(z)}{B(z)}$, $S(z) = p_d(z)$ and $\tilde{T}(z) := \frac{p_c(z)}{B(z)}$, then \cref{eq:CTChPolAxial} implies:
	
	\begin{equation}
	\d T = \varepsilon \d \tilde{T} + \d S
	\end{equation}
	
	Also recall that in these spaces, the index $k \leq 3$. Hence
	
	\begin{align}
	\bp{\d T, \d T} & = \d T (\nabla T) \\
	& = \bp{\d \tilde{T}, \d \tilde{T}}+ \bp{\d S, \d S} \\
	& = 4\deriv{}{z} \tilde{T}(z) + 4\varepsilon \deriv{}{z} S(z)  \quad \text{by \cref{lem:CctTForm} and \cref{prop:CTDetIrredAxial} }   \\ 
	& = 4\varepsilon \deriv{}{z} ( \varepsilon\tilde{T}(z) + S(z)) \\
	& \overset{\eqref{eq:CTChPolAxial} }{=} 4\varepsilon \deriv{}{z}\frac{p(z)}{B(z)}
	\end{align}
	
	Thus we have the following:
	\begin{align}
	\frac{\scalprod{(\d p)|_{z=u^{i}}}{(\d p)|_{z=u^{i}}}}{B(u^{i})^{2}} & = 4 \varepsilon \deriv{}{z}\frac{p(z)}{B(z)}\bigg\lvert_{z=u^{i}} \\
	& = 4 \varepsilon \frac{p'(u^{i})}{B(u^{i})}
	\end{align}
	
	From \cref{prop:CTmetric} we have:
	
	\begin{align}
	g^{ii} & = \frac{\scalprod{(\d p)|_{z=u^{i}}}{(\d p)|_{z=u^{i}}}}{p'(u^{i})^{2}} \\
	& = 4 \varepsilon \frac{B(u^{i})}{p'(u^{i})} \\
	& = 4 \varepsilon\frac{\prod\limits_{j=k+1}^{n} (u^{i}-\lambda_{j})}{\prod\limits_{j \neq i}(u^{i}-u^{j})}
	\end{align}
\end{proof}
\begin{remark}
	The above technique for calculating the metric is based on Moser's calculation of the metric for sphere-elliptic coordinates in \cite[P.~179-180]{Moser2011}.
\end{remark}

\begin{corollary}
	Suppose $L$ is a non-degenerate CT in Euclidean or Minkowski space in canonical form. Then the points at which a real eigenvalue of $A_c$ is an eigenvalue of $L$ are singular, i.e. $L$ cannot be an ICT in any neighborhood of these points.
\end{corollary}


\subsection{Concircular tensors in Spherical Submanifolds of pseudo-Euclidean space} \label{sec:cTIrredSphC}

In this section we treat the case of CTs defined on $\eunn(\kappa)$. We will be able to reduce most calculations to similar ones involving central CTs. The following proposition will allow us to do this.

\begin{propMy}[Determinant of Spherical CTs]
	Suppose $L = R L_c R^*$ is a CT in $\eunn(\frac{1}{r^2})$, the following holds:
	
	\begin{equation} \label{eq:CTChPolEunnKap}
	p(z) = \det (z R -  L + \frac{r \otimes r^{\flat}}{r^{2}} ) =  r^{-2} (B(z) - p_c(z)) 
	\end{equation}
\end{propMy}
\begin{proof}
	It is sufficient to prove that:
	
	\begin{equation} \label{eq:CTDetEunnKap}
	\det (L + \frac{r \otimes r^{\flat}}{r^{2}} ) =  r^{-2} (\det L_c - \det A) 
	\end{equation}
	
	Observe that:
	
	\begin{align}
	L + \frac{r \otimes r^{\flat}}{r^{2}} & = A R + \frac{[(r \cdot A \cdot r) + r^2]}{r^{4}} r \otimes r^{\flat} -  \frac{1}{r^2} r \otimes r^{\flat} \cdot A \\
	& = A R + r \otimes d
	\end{align}
	
	\noindent for some vector $d$ and
	
	\begin{align}
	A R & =  A  - \frac{1}{r^2} A r \otimes r^{\flat}
	\end{align}
	
	Let $b_i$ be the columns of $A R$, then by \cref{lem:detAplusTen} we have
	
	\begin{align}
	\det (L + \frac{r \otimes r^{\flat}}{r^{2}} ) & = \bigwedge \limits_{i=1}^{n} b_{i} + 
	\sum\limits_{i=1}^{n} b_{1} \wedge \cdots \wedge d_{i} r \wedge \cdots \wedge b_{n}
	\end{align}
	
	Now observe that
	
	\begin{align}
	0 = \det L & = \bigwedge \limits_{i=1}^{n} b_{i} + 
	\sum\limits_{i=1}^{n} b_{1} \wedge \cdots \wedge (d_{i} - \frac{x_i}{r^2}) r \wedge \cdots \wedge b_{n}
	\end{align}
	
	Thus
	
	\begin{align}
	\det (L + \frac{r \otimes r^{\flat}}{r^{2}} ) & = \frac{1}{r^2}\sum\limits_{i=1}^{n} b_{1} \wedge \cdots \wedge  x_i r \wedge \cdots \wedge b_{n} \label{eq:CTDetEunnKapI}
	\end{align}
	
	Now, again using \cref{lem:detAplusTen}, we have:
	
	\begin{align}
	b_{1} \wedge \cdots \wedge r \wedge \cdots \wedge b_{n} & = (-1)^{i-1} r \wedge b_{1} \wedge \cdots \wedge \hat{b}_{i} \wedge \cdots \wedge b_{n} \\
	& = (-1)^{i-1} r \wedge (a_{1} \wedge \cdots \wedge \hat{a}_{i} \wedge \cdots \wedge a_{n} - r^{-2} \sum\limits_{j \neq i} a_{1} \wedge \cdots \wedge x_{j} A r \wedge \cdots \wedge a_{n})
	\end{align}
	
	Note that the term $\hat{b}_{i}$, means $b_{i}$ is missing from the product. Also note that for $i \neq j$
	
	\begin{equation}
	(-1)^{i-1} x_i r \wedge a_{1} \wedge \cdots \wedge x_{j} A r \wedge \cdots \wedge a_{n} = - (-1)^{j-1} x_j r \wedge a_{1} \wedge \cdots \wedge x_{i} A r \wedge \cdots \wedge a_{n}
	\end{equation}
	
	Thus
	
	\begin{align}
	\det (L + \frac{r \otimes r^{\flat}}{r^{2}} ) & \overset{\eqref{eq:CTDetEunnKapI}}{=} r^{-2}\sum\limits_{i=1}^{n} b_{1} \wedge \cdots \wedge  x_i r \wedge \cdots \wedge b_{n} \\
	& = r^{-2} \sum\limits_{i=1}^{n} a_{1} \wedge \cdots \wedge  x_i r \wedge \cdots \wedge a_{n} \\
	& =  r^{-2} (\det (A + r \otimes r^{\flat}) - \det A) 
	\end{align}
\end{proof}

Using \cref{eq:CTChPolEunnKap}, for ICTs, the transformation from canonical coordinates to Cartesian coordinates can be calculated using the standard method. Indeed, if $L$ is an ICT in $\eunn(\frac{1}{r^2})$ with parameter matrix:

\begin{align}
A & = J_k^T(0) \oplus \diag(\lambda_{k+1},\dotsc,\lambda_n) & g & = \varepsilon_0 S_k \oplus \diag(\varepsilon_{k+1},\dotsc,\varepsilon_n)
\end{align}

Then by a calculation almost identical to the one used to derive \cref{eq:CTCoordsCCIIa,eq:CTCoordsCCIIb}, one obtains the following now using \cref{eq:CTChPolEunnKap}:

\begin{subequations} \label{eq:CTCoordsCCIII}
	\begin{equation} \label{eq:CTCoordsCCIIIa}
	\sum\limits_{i=1}^{l+1} x^{i} x^{l+2-i} = \frac{r^2 \varepsilon_0 }{l! } (\deriv{}{z})^{l}(\frac{p(z)}{B_{u^\perp}(z)}) \big \lvert_{z = 0}  \quad l = 0,\dotsc, k-1
	\end{equation}
	\begin{align}
	(x^{i})^{2}  & =  r^2\varepsilon_i \frac{p(\lambda_i)}{B'(\lambda_i)} & i = k+1,...,n \label{eq:CTCoordsCCIIIb}
	\end{align}
\end{subequations}

The transformation from canonical coordinates $(u^1,\dotsc,u^{n-1})$ to Cartesian coordinates are obtained by noting that $p(z) = \prod\limits_{i=1}^{n-1} (z-u^{i})$.

\begin{example}[Circular coordinates] \label{ex:coordEunKap}
	Let $M = \E^2_{\nu} (\kappa)$ where $\kappa = \pm 1$. Consider the CT in $M$ with parameter matrix:
	
	\begin{align}
	A & = \diag(0,1) & g & = \diag(\kappa_1, \varepsilon) \quad \kappa_1, \varepsilon \in \{-1,1\}
	\end{align}
	
	Then by \cref{eq:CTCoordsCCIIIa,eq:CTCoordsCCIIIb}, Cartesian coordinates $(x,y)$ are given by:
	
	\begin{align}
	x^{2} & = \kappa \kappa_1 u \\
	y^{2} & = \kappa \varepsilon (1 - u)
	\end{align} 
	
	We now show how to obtain the standard parameterizations of these coordinates. First note that by metric-Jordan canonical form theory, there are three isometrically inequivalent cases\footnote{Note that these cases additionally depend on $\nu$.}:
	
	\begin{parts}
		\item $\kappa_1 = \kappa$ and $\varepsilon = \kappa$, thus $g = \diag(\kappa, \kappa)$
		
		If we take $u = \cos^2(t)$, then we obtain:
		
		\begin{align}
		x^{2} & = \cos^2(t) \\
		y^{2} & = \sin^2(t)
		\end{align}
		
		\item $\kappa_1 = \kappa$ and $\varepsilon = -\kappa$, thus $g = \diag(\kappa, -\kappa)$ 
		
		If we take $u = \cosh^2(t)$, then we obtain:
		
		\begin{align}
		x^{2} & = \cosh^2(t) \\
		y^{2} & = \sinh^2(t)
		\end{align}
		
		\item $\kappa_1 = -\kappa$ and $\varepsilon = \kappa$, thus $g = \diag(-\kappa, \kappa)$ 
		
		If we take $u = -\sinh^2(t)$, then we obtain:
		
		\begin{align}
		x^{2} & = \sinh^2(t) \\
		y^{2} & = \cosh^2(t)
		\end{align}
	\end{parts}
	
	Although the last two cases are geometrically equivalent, it will be useful to distinguish them when we move on to reducible CTs.
\end{example}

Also using \cref{eq:CTChPolEunnKap}, one can obtain the metric in ICT induced coordinates.

\begin{proposition}[ICT metrics in $\eunn(\kappa)$] \label{prop:IctMetricEunnKap}
	Suppose $L$ is an ICT in $\eunn(\frac{1}{r^2})$ with eigenfunctions $(u^1,\dotsc,u^{n-1})$. Then the metric in adapted coordinates is orthogonal and
	
	\begin{equation}
	g_{ii} = \frac{- r^2}{4} \frac{p'(u^i)}{B(u^i)} = \frac{- r^2}{4}\frac{\prod\limits_{j \neq i}(u^{i}-u^{j})}{\prod\limits_{j=1}^{n} (u^{i}-\lambda_{j})}
	\end{equation}
	
	\noindent where $\lambda_{1},\dotsc,\lambda_n$ are the roots of $B(z)$.
\end{proposition}
\begin{proof}
	We reduce this calculation to the corresponding one for $L_c$ using \cref{eq:CTChPolEunnKap}. We assume that $L$ is an ICT with eigenfunctions $(u^1 ,\dotsc, u^{n-1})$ in some neighborhood in $\eunn(\frac{1}{r^2})$.
	
	Now if we let $\tilde{d}$ denote the exterior derivative on the sphere, note that
	
	\begin{equation}
	\tilde{\d} p = R^* \d p
	\end{equation}
	
	Now we make the following observation. 
	
	\begin{equation}
	\bp{\d p , r^\flat } = \nabla_r p = 0
	\end{equation}
	
	%
	
	This can be proven, for example, by using \cref{eq:CTDetCent} and the fact that $r$ is a CV. Note that the above equation also implies that $\bp{\d p_c  , r^\flat } = - 2 r^2 p $.
	
	Hence we see that
	
	\begin{equation}
	\bp{\tilde{\d} p , \tilde{\d} p} = \bp{\d p , \d p}
	\end{equation}

	Thus at a root $z = u^i$, we have
	
	\begin{equation}
	\bp{\tilde{\d} p , \tilde{\d} p} = r^{-4} \bp{\d p_c , \d p_c }
	\end{equation}
	
	Then at $z = u^i$ we have
	\begin{align}
	\frac{\bp{\tilde{\d} p, \tilde{\d} p}}{B^2} & = \frac{r^{-4} \bp{\d p_c , \d p_c }}{B^2} \\
	& \overset{\ref{lem:CctTForm}}{=} 4 r^{-4} \deriv{}{z} \frac{p_c(z)}{B(z)}\bigg\lvert_{z=u^{i}} \\
	& = - 4 r^{-2} \deriv{}{z} \frac{p(z)}{B(z)}\bigg\lvert_{z=u^{i}} \\
	& = - 4 r^{-2} \frac{p'(u^i)}{B(u^i)}
	\end{align}
	
	Thus \cref{prop:IctMetricEunnKap} follows from the above equation and  \cref{prop:CTmetric}.
\end{proof}


\section{Classification of reducible concircular tensors} \label{sec:CtClassRed}

In this section, we will show how to find a warped product which ``decomposes''\footnote{This amounts to partially diagonalizing these CTs.} a given reducible OCT defined in a space of constant curvature. First we will prove a generic result which will allow us to construct reducible OCTs. Then in the next two sections, we will apply this result to pseudo-Euclidean space, then to spherical submanifolds of pseudo-Euclidean space.

The following proposition will give us a useful characterization of reducible OCTs in terms of their irreducible part. Its proof, which is based on theorem~6.1 in \cite{Rajaratnam2014a}, can be omitted without lose of continuity.

\begin{propMy}[Characterization of Reducible OCTs] \label{prop:CTredChar}
	Suppose $L \in S^{2}(M)$ is an orthogonal tensor. Then $L$ is a reducible OCT iff there exists a warped product decomposition $M = M_{0} \times_{\rho_{1}} M_{1} \times \cdots \times_{\rho_{k}} M_{k}$ with adapted contravariant metric $G = \sum_{i=0}^{k} G_{i}$ such that $L$ has the following contravariant form:
	\begin{equation}
	L = \tilde{L} + \sum_{i=1}^{k} \lambda_i G_{i}
	\end{equation}
	\noindent where each $\lambda_i \in \R$ and $\tilde{L} \in \hat{S}^{2}(M_{0})$ is the canonical lift (see \cite{Rajaratnam2014a}) of an ICT $\tilde{L} \in S^{2}(M_{0})$ satisfying the following equation on $M_{0}$ for each $i > 0$
	\begin{equation} \label{eq:CTredExt}
	\tilde{L}(\d \log \rho_{i}) = \d ( \lambda_i  \log \rho_{i} + \frac{1}{2} \tr{\tilde{L}})
	\end{equation}
\end{propMy}
\begin{proof}
	Suppose $L$ is an OCT. Let $D_1,\dotsc, D_l$ be the eigenspaces of $L$ associated with constant eigenfunctions and let $M = M_{0} \times_{\rho_{1}} M_{1} \times \cdots \times_{\rho_{k}} M_{k}$ be a warped product adapted to $(\bigcap\limits_{i=1}^{l} D_{i}^{\perp}, D_1,\dotsc, D_l)$ which exists by theorem~6.1 in \cite{Rajaratnam2014a}. We define $\tilde{L}$ to be the restriction of $L$ to $M_0$; it follows by theorem~6.1 that $\tilde{L}$ is an ICT in $M_0$. It also follows by theorem~6.1 that we can assume 
	
	\begin{equation} \label{eq:OCTwarpFn}
	\rho_{i}^{2} = \prod\limits_{a} \Abs{\lambda_{i}-\lambda_{a}}
	\end{equation}
	
	\noindent where $a$ ranges over all eigenfunctions of $\tilde{L}$. If $\dim M_0 = 0$, i.e. $L$ induces a pseudo-Riemannian product, the conclusion follows. Otherwise, since $\lambda_i$ is constant and because $\tilde{L}$ is torsionless, we see that on $M_0$

	\begin{align}
	\tilde{L}(\d \log \rho_{i}) & =  \frac{1}{2} \sum\limits_{a} \lambda_{a} \d \log \Abs{\lambda_{i}-\lambda_{a}} \\
	& =  \frac{1}{2} \sum\limits_{a} \lambda_{a} \frac{\d \lambda_{a}}{\lambda_{a}-\lambda_{i}} \\
	& = \frac{\lambda_i}{2} \sum\limits_{a} \frac{\d \lambda_{a}}{\lambda_{a}-\lambda_{i}} + \frac{1}{2} \sum\limits_{a} \d \lambda_{a} \\
	& = \d ( \lambda_i  \log \rho_{i} + \frac{1}{2} \tr{\tilde{L}})
	\end{align}
	
	Conversely, it is easily checked that if $\tilde{L}$ is an ICT and $\rho_i$ satisfies the above equation, then $c \rho_i$ must satisfy \cref{eq:OCTwarpFn} for some $c \in \R^+$. Hence it follows that $L$ defined in the statement is torsionless and then by theorem~6.1 in \cite{Rajaratnam2014a} that $L$ is a reducible OCT.
\end{proof}

In the following sections we will use the above proposition to classify reducible OCTs in spaces of constant curvature. But first we will need the following definition.

\begin{definition}
	Suppose $L$ is a CT in $M$ and let $N = N_{0} \times_{\rho_{1}} N_{1} \times \cdots \times_{\rho_{k}} N_{k}$ be a local warped product decomposition of $M$ passing through $\bar{p} \in N \subseteq M$. We say $L$ is \emph{decomposable} in this warped product if for each $p \in N$ and $i > 0$, $T_p N_i$ is an invariant subspace for $L$.
\end{definition}

\subsection{In pseudo-Euclidean space} \label{sec:CtClassRedEunn}

We first need to review the standard warped product decompositions of $\eunn$. All other warped product decompositions of $\eunn$ can be built up from the standard ones. Our exposition is based on the article by \citeauthor{Nolker1996} \cite{Nolker1996}. More details are given in \cite[Appendix~D]{Rajaratnam2014} where the standard warped products of spaces of constant curvature are given, generalizing results originally given in \cite{Nolker1996}.

Consider the following decomposition $\eunn = V_{0} \obot V_{1}$ of $\eunn$ into nontrivial (hence non-degenerate) subspaces. Choose $a \in V_0 \setminus \{0\}$ and $\bar{p} \in V_0$ such that $\bp{a, \bar{p}} = 1$. Denote $\kappa := a^2$ and $\epsilon := \sgn \kappa$. We have two types of warped products:

\begin{description}
	\item[non-null warped decomposition] If $\kappa \neq 0$, let $W_0 := V_0 \cap a^\perp$ and $W_1 := W_0^\perp$. Let $c = \bar{p} - \frac{a}{\kappa}$ and
	
	\begin{equation}
	N_1 = c + \{ p \in W_1 \: | \: p^{2} = \frac{1}{\kappa} \}
	\end{equation}
	
	\item[null warped decomposition] If $\kappa = 0$, then $a$ is lightlike, so fix another lightlike vector $b \in V_0$ such that $\bp{a,b} = 1$, let $W_0 := V_0 \cap \spa{a,b}^\perp$ and $W_1 := V_1$. Let
	
	\begin{equation}
	N_1 = \bar{p} + \{ p - \frac{1}{2} p^{2} a  \: | \: p \in W_1 \}
	\end{equation}
\end{description}

In each case, we say that $N_1$ is the sphere determined by $(\bar{p},V_1, a)$. For $i = 0,1$, let $P_i : \eunn \rightarrow W_{i}$ be the orthogonal projection. Let

\begin{equation}
N_{0} = \{p \in V_{0} | \bp{a, p} > 0  \}
\end{equation}

\begin{equation}
\rho : \begin{cases}
N_{0} & \rightarrow \R_{+} \\
p_{0} & \mapsto \bp{a, p_{0}}
\end{cases}
\end{equation}

Then the following holds:

\begin{theorem}[Standard Warped Products in $\eunn$ \cite{Nolker1996}] \label{thm:WPDecomps}
	The map
	
	\begin{equation}
	\psi : \begin{cases}
	N_{0} \times _{\rho} N_{1} & \rightarrow \eunn \\
	(p_{0},p_{1}) & \mapsto p_{0} + \rho(p_{0})(p_{1} - \overline{p})
	\end{cases}
	\end{equation}
	
	\noindent is an isometry onto the following set:
	
	\begin{equation}
	\Ima(\psi) = \begin{cases}
	\{ p \in \eunn \: | \: \sgn (P_{1}p)^{2} = \epsilon \} & \text{non-null case} \\
	\{ p \in \eunn \: | \: \bp{a,p} > 0 \} & \text{null case}
	\end{cases}
	\end{equation}
	
	Furthermore, the following equation holds:
	
	\begin{equation} \label{eq:wpdecpSp}
	\psi(p_{0},p_{1})^{2} = p_{0}^{2}
	\end{equation}
\end{theorem}
\begin{proof}
	See \cite[Appendix~D]{Rajaratnam2014}.
\end{proof}

In fact, for $(p_{0},p_{1}) \in N_{0} \times N_{1}$, $\psi$ has one of the following forms, first if $\psi$ is non-null:

\begin{equation} \label{eq:WPDecNn}
\psi(p_{0},p_{1}) = P_{0}p_{0} + \bp{a, p_{0}} (p_{1}-c)
\end{equation}

and if $\psi$ is null:

\begin{equation} \label{eq:WPDecN}
\psi(p_{0},p_{1}) = P_{0}p_{0} + (\bp{b, p_{0}} - \frac{1}{2}\bp{a, p_{0}}(P_{1}p_1)^{2})a + \bp{a,p_{0}}b + \bp{a, p_{0}} P_{1}p_1
\end{equation}


The above forms are obtained from the equation for $\psi$ from the above theorem by expanding $p_0$ in an appropriate basis. The warped product decomposition $\psi$ is completely determined by the fact that $\psi(\bar{p},\bar{p}) = \bar{p}$, $N_1$ is a spherical submanifold of $\eunn$ with $\bar{p} \in N_1$, $T_{\bar{p}} N_1 = V_1$ and mean curvature normal $-a$ at $\bar{p}$ (see \cite{Nolker1996} or \cite[Appendix~D]{Rajaratnam2014}). The point $\bar{p}$ was restricted so that the warped product is in \emph{canonical form} (see \cite[Appendix~D]{Rajaratnam2014}); we will make this assumption throughout this article. We call $\psi$ the warped product decomposition (of $\eunn$) determined by $(\bar{p}; V_0 \obot V_1 ; a)$; often we omit the point $\bar{p}$ as it doesn't enter calculations, in this case the warped product is assumed to be in canonical form.

We note that the warped products with multiple spherical factors can be obtained using the standard ones described above. Indeed, suppose $\phi_1 : N_0' \times_{\rho_1} N_1 \rightarrow \eunn$ is the warped product decomposition determined by $(\bar{p}; V_0 \obot V_1 ; a_1)$ as above. Since $V_0$ is pseudo-Euclidean, consider a warped product decomposition, $\phi_2 : \tilde{N}_0 \times_{\rho_2} N_2 \rightarrow V_0$, determined by $(\bar{p}; \tilde{V}_0 \obot \tilde{V}_1 ; a_2)$ with $V_0 \cap W_0^\perp \subset \tilde{W}_0$ (hence $a_1 \in \tilde{W}_0$). Note that $\tilde{W}_0$ is the subspace $W_0$ from the above construction for $\phi_2$. Let $N_0 := N_0' \cap \tilde{N}_0$, then one can check that the map $\psi$ defined by:

\begin{equation}
\psi : \begin{cases}
N_{0} \times_{\rho_1} N_1 \times_{\rho_2} N_2 & \rightarrow \eunn \\
(p_0 , p_1 , p_2) & \mapsto \phi_1(\phi_2(p_0, p_2), p_1)
\end{cases}
\end{equation}

\noindent is a warped product decomposition of $\eunn$. We illustrate this construction with an example.

\begin{example}[Constructing multiply warped products]
	Suppose $\phi_1$ and $\phi_2$ are given as follows:
	
	\begin{align}
	\phi_1(p_{0}',p_{1}) & = P_{0}'p_{0}' + \bp{a_1, p_{0}'} (p_{1}-c_1) \\
	\phi_2(\tilde{p}_{0},p_{2}) & = \tilde{P}_{0}\tilde{p}_{0} + \bp{a_2, \tilde{p}_{0}} (p_{2}-c_2)
	\end{align}
	
	Observe that $\rho_1(\phi_2(\tilde{p}_{0},p_{2})) = \rho_1(\tilde{p}_{0})$, which follows from the above equation for $\phi_2$ and the fact that $a_1 \in \tilde{W}_0$. Then, 
	
	\begin{align}
	\psi(p_0 , p_1 , p_2) & = \phi_1(\phi_2(p_0, p_2), p_1) \\
	& = P_{0}'\phi_2(p_0, p_2) + \bp{a_1, \phi_2(p_0, p_2)} (p_{1}-c_1) \\
	& = P_{0}'\tilde{P}_{0}p_0 + \bp{a_2, p_{0}} (p_{2}-c_2) + \bp{a_1, p_0} (p_{1}-c_1)
	\end{align}
	
	\noindent where $P_{0}'\tilde{P}_{0}$ is the orthogonal projector onto $\tilde{W}_0 \cap W_0 = \tilde{V}_0 \cap \spa{a_1,a_2}^\perp$.
\end{example}

This procedure can be repeated as many times as necessary to obtain more general warped products. In general, for some $\bar{p} \in \eunn$, suppose we have a decomposition $T_{\bar{p}} \eunn = \bigobot\limits_{i=0}^{k} V_{i}$ into non-trivial subspaces (hence non-degenerate) with $k \geq 1$ and linearly independent pair-wise orthogonal vectors $a_1,\dotsc,a_k \in V_0 \setminus \{0\}$. Furthermore we will assume the warped product is in canonical form, so $\bar{p} \in V_0$ and $\bp{a_i, \bar{p}} = 1$ for each $i$. This data determines a warped product decomposition $\psi$, having the following form \cite[Appendix~D]{Rajaratnam2014}:

\begin{equation} \label{eq:genPsiEqn}
\psi : \begin{cases}
N_{0} \times _{\rho_{1}} N_{1} \times \cdots  \times _{\rho_{k}} N_{k} & \rightarrow \eunn \\
(p_{0},...,p_{k}) & \mapsto p_{0} + \sum\limits_{i=1}^{k} \rho_{i}(p_{0})(p_{i} - \overline{p})
\end{cases}
\end{equation}

\noindent where $\rho_i(p_0) = \bp{a_i , p_0}$ and $N_i$ is the sphere determined by $(\bar{p},V_i, a_i)$. This general formula is originally from \cite[theorem~7]{Nolker1996}. We call $\psi$ the warped product decomposition (of $\eunn$) determined by $(\bar{p}; \bigobot\limits_{i=0}^{k} V_{i}; a_{1},...,a_{k})$. One can more generally let some of the $a_i$ be zero, this results in Cartesian products as done in \cite{Nolker1996}. Since we assume the $a_i$ are non-zero, we say additionally that $\psi$ is a \emph{proper} warped product decomposition. Finally, note that the properties of the more general warped product decompositions of $\eunn$ can be deduced from \cref{thm:WPDecomps}.

Now suppose $N = N_{0} \times_{\rho_{1}} N_{1} \times \cdots \times_{\rho_{k}} N_{k}$ is a warped product and $\tilde{L}$ is a CT in $N_0$. We say $\tilde{L}$ can be \emph{extended to a CT in $N$} if $\tilde{L}$ satisfies \cref{eq:CTredExt} for each $i$ with some $\lambda_i \in \R$. Assuming $\tilde{L}$ is an OCT, then \cref{prop:CTredChar} allows one to define a CT on $N$ which restricts to $\tilde{L}$ on $N_0$. The following lemma will be our main tool for classifying reducible concircular tensors.

\begin{lemma} \label{prop:extCT}
	Fix a proper warped product decomposition $(V_{0} \obot V_{1}; a)$ of $\eunn$ and let $?L^i_j? = ?A^i_j? + m x^{i} x_{j} + w^{i} x_{j} + x^{i}w_{j} $ be a concircular tensor in $N_{0}$. Then $L$ can be extended to concircular tensor in $\eunn$ decomposable in this warped product iff $a$ is an eigenvector of $A$ orthogonal to $w$.
\end{lemma}
\begin{proof}
	First observe
	
	\begin{align}
	v^{k}\nabla_{k} \tr{L} & = v^{k}\nabla_{k} (m  x_{i} x^{i} + 2x^{i}w_{i}) \\
	& = m [(v^{k}\nabla_{k} x_{i}) x^{i} + x_{i} (v^{k}\nabla_{k} x^{i})] + 2 [(v^{k}\nabla_{k} w_{i}) x^{i} + w_{i} (v^{k}\nabla_{k} x^{i})] \\
	& = m (v_{i} x^{i} + x_{i}v^{i}) + 2 v^{i}w_{i} \\
	& = 2m v^{i}x_{i} + 2 v^{i} w_{i}
	\end{align}
	
	Hence $\nabla^{i} \tr{L} = 2 (m x^{i} + w^{i})$. Now let $\rho = a^{i} x_{i} = \bp{a, x} > 0$, then one can similarly show that 
	
	\begin{equation}
	\nabla^{i} \log  \rho = \frac{a^{i}}{\rho}
	\end{equation}
	
	Then,
	
	\begin{align}
	?L^i_j? \nabla^{j} \log  \rho - \frac{1}{2} \nabla^{i} \tr{L} & = \frac{1}{\rho} (?A^i_j?a^{j} + m x^{i} x_{j}a^{j} + w^{i} x_{j}a^{j} + x^{i}w_{j}a^{j}) - m x^{i} - w^{i} \\
	& = \frac{1}{\rho} (?A^i_j?a^{j} + x^{i}w_{j}a^{j}) + \frac{1}{\rho} (m x^{i} \rho + w^{i} \rho) - m x^{i} - w^{i} \\
	& = \frac{1}{\rho} (?A^i_j?a^{j} + x^{i}w_{j}a^{j})
	\end{align}
	
	By definition, $L$ can be extended to a CT decomposable in this warped product iff $?L^i_j? \nabla^{j} \log  \rho - \frac{1}{2} \nabla^{i} \tr{L} \in \spa {\nabla^{i} \log  \rho}$. The above equation implies that this happens iff $a$ is an eigenvector of $A$ and $a \in w^{\perp}$.
\end{proof}

We now use the above lemma to construct reducible CTs in $\eunn$.

\begin{propMy}[Constructing Reducible CTs in $\eunn$] \label{prop:CtRedConstrEn}
	Fix a proper warped product decomposition $(V_{0} \obot V_{1}; a)$ of $\eunn$ and let $\tilde{L} = \tilde{A} + m \tilde{r} \odot \tilde{r} + 2 \tilde{r} \odot \tilde{w} $ be a concircular tensor in $N_{0}$ (in contravariant form) which can be extended to a concircular tensor L in $\eunn$ via the above lemma. Since $N_{0} \subset V_{0} \subset \eunn$, we can consider $\tilde{L}$ to be a tensor in $\eunn$. Then L is given as follows:
	
	\begin{equation}
	L = A + m r \odot r + 2 r \odot \tilde{w}
	\end{equation}
	
	\noindent where as a linear operator, $A = \tilde{A} + \lambda I_{V_{1}}$, where $\lambda$ is the eigenvalue of $\tilde{A}$ associated with $a$ and $I_{V_{1}}$ is the identity on $V_{1}$.
\end{propMy}

\begin{proof}
	Throughout the proof, $G$ is the contravariant metric for $\eunn$ and this metric adapted to the warped product is given as follows:
	
	\begin{equation}
	G = G' + \frac{1}{\rho^{2}} G_{1}
	\end{equation}
	
	\textbf{The non-null case:} In this case $\kappa_{1} := a^{2} = \pm 1$. Let $m := \dim V_{0}$ and choose an orthonormal basis for $V_0$, $\{a_{1},...,a_{m}\}$ with $a_{m} = a$.
	
	First note that for $p = (p_{0},p_{1}) \in N_{0} \times N_{1}$ and $v = (v_{0},v_{1}) \in T_{p} (N_{0} \times N_{1})$, \cref{eq:WPDecNn} implies that
	
	\begin{equation}  \label{eq:nonNullWPPush}
	\psi_{*}v = P_{0}v_{0} + \bp{a, v_{0}}(p_{1}-c) + \bp{a, p_{0}}v_{1}
	\end{equation}
	
	Hence we observe the following:
	
	\begin{align}
	\psi_{*}p_{0} & = P_{0}p_{0} + \bp{a, p_{0}}(p_{1}-c) \label{eq:nonNullPushR} \\
	& = \psi(p_{0},p_{1})
	\end{align}
	
	and
	
	\begin{equation} \label{eq:nonNullPushBas}
	\psi_{*}a_{i} = a_{i} \quad \text{for } i = 1,...,m-1
	\end{equation}
	
	Now let $\tilde{L} = \tilde{A} + m \tilde{r} \odot \tilde{r} + 2 \tilde{w} \odot \tilde{r}$ be a concircular tensor in $N_{0}$ satisfying $\tilde{A}a = \lambda a$ for some $\lambda$ and $\bp{a,\tilde{w}} = 0$. Then from \cref{prop:extCT} we know that $\psi_{*}(\tilde{L} + \frac{\lambda}{\rho^{2}} G_{1})$ is a concircular tensor in $\eunn$. We now calculate $\psi_{*}(\tilde{L} + \frac{\lambda}{\rho^{2}} G_{1})$ explicitly.
	
	First note that
	
	\begin{equation}
	\tilde{A} = A_{0} + \lambda \kappa_1 a \odot a
	\end{equation}
	
	\noindent where $A_{0} a = 0$ and so $\psi_* A_0 = A_0$ by \cref{eq:nonNullPushBas}. Let $G$ be the contravariant metric for $\eunn$ and $G_{0} $ be the restriction of $G$ to $W_{0}$, then
	
	\begin{align}
	G & = G' + \frac{1}{\rho^{2}} G_{1} \\
	& = G_{0} + \kappa_1 a \odot  a + \frac{1}{\rho^{2}} G_{1}
	\end{align}
	
	Thus
	
	\begin{equation}
	\frac{1}{\rho^{2}} G_{1} = G - G_{0} - \kappa_1 a \odot a
	\end{equation}
	
	Let $G_{V_{1}}$ be the restriction of $G$ to $V_{1}$, then
	
	\begin{align}
	\psi_{*}(\tilde{A} + \frac{\lambda}{\rho^{2}} G_{1}) & = \psi_{*}(A_{0} + \lambda \kappa_1 a \odot a + \lambda( G - G_{0} - \kappa_1 a \odot a)) \\
	& = \psi_{*}(A_{0} + \lambda( G - G_{0})) \\
	& = A_{0} + \lambda( G - G_{0}) \\
	& = \tilde{A} + \lambda G_{V_{1}}
	\end{align}
	
	\noindent where the second last equality follows from \cref{eq:nonNullPushBas} and the fact that $\psi$ is an isometry.
	
	\cref{eq:nonNullPushR} implies that $\psi_{*}\tilde{r} = r$, also \cref{eq:nonNullPushBas} together with the fact that $\bp{a,\tilde{w}} = 0$ implies that $\psi_{*}\tilde{w} = \tilde{w}$. Thus we conclude that
	
	\begin{equation}
	\psi_{*}(\tilde{L} + \frac{\lambda}{\rho^{2}} G_{1}) = A + m r \odot r + 2 r \odot \tilde{w}
	\end{equation}
	
	\noindent where as a linear operator, $A = \tilde{A} + \lambda I_{V_{1}}$ where $I_{V_{1}}$ is the identity on $V_{1}$.
	
	\textbf{The null case:} In this case $a$ is a lightlike vector. Let $m := \dim V_{0}$ and choose a basis $\{a_{1},...,a_{m-2},a,b\}$ for $V_0$ where $\{a_{1},...,a_{m-2}\}$ is an orthonormal basis for $W_0$ and $a,b$ are as in the null warped product decomposition.
	
	First note that for $p = (p_{0},p_{1}) \in N_{0} \times N_{1}$ and $v = (v_{0},v_{1}) \in T_{p} (N_{0} \times N_{1})$, \cref{eq:WPDecN} implies that
	
	\begin{align}
	\psi_{*}v & = P_{0}v_{0} + (\bp{b, v_{0}} - \frac{1}{2}\bp{a, v_{0}}(P_{1}p_1)^{2} - \bp{a, p_{0}}\bp{P_{1}p_{1}, P_{1}v_{1}})a + \bp{a,v_{0}}b \nonumber \\
	& \qquad + \bp{a, v_{0}} P_{1}p_{1} + \bp{a, p_{0}} P_{1}v_{1} \label{eq:nullWPPush}
	\end{align}
	
	Hence we observe the following:
	
	\begin{align}
	\psi_{*}p_{0} & = P_{0}p_{0} + (\bp{b, p_{0}} - \frac{1}{2}\bp{a, p_{0}}(P_{1}p_1)^{2})a + \bp{a,p_{0}}b  + \bp{a, p_{0}} P_{1}p_{1} \label{eq:nullPushR} \\
	& = \psi(p_{0},p_{1})
	\end{align}
	
	and
	
	\begin{align}
	\psi_{*}a_{i} & = a_{i} \quad i = 1,...,m-2 \label{eq:nullPushBas} \\
	\psi_{*}a & = a
	\end{align}
	
	Now let $\tilde{L} = \tilde{A} + m \tilde{r} \odot \tilde{r} + 2 \tilde{w} \odot \tilde{r}$ be a concircular tensor on $N_{0}$ satisfying $\tilde{A}a = \lambda a$ for some $\lambda$ and $\bp{a,\tilde{w}} = 0$. Then from \cref{prop:extCT} we know that $\psi_{*}(\tilde{L} + \frac{\lambda}{\rho^{2}} G_{1})$ is a concircular tensor in $\eunn$. We now calculate $\psi_{*}(\tilde{L} + \frac{\lambda}{\rho^{2}} G_{1})$ explicitly.
	
	Since $\tilde{A}a = \lambda a$, $\tilde{A}$ can be decomposed in contravariant form as follows:
	
	\begin{equation}
	\tilde{A} = A_{0} + 2 \lambda a \odot b
	\end{equation}
	
	\noindent where $A_{0} a = 0$ and so $\psi_* A_0 = A_0$ by \cref{eq:nullPushBas}. Let $G$ be the contravariant metric for $\eunn$ and $G_{0} $ be the restriction of $G$ to $W_{0}$, then we see that
	
	\begin{equation}
	\frac{1}{\rho^{2}} G_{1} = G - G_{0} - 2 a \odot b
	\end{equation}
	
	Let $G_{V_{1}}$ be the restriction of $G$ to $V_{1}$, then
	
	\begin{align}
	\psi_{*}(\tilde{A} + \frac{\lambda}{\rho^{2}} G_{1}) & = \psi_{*}(A_{0} +  2\lambda a \odot b + \lambda( G - G_{0} -2 a \odot b)) \\
	& = \psi_{*}(A_{0} + \lambda( G - G_{0})) \\
	& = A_{0} + \lambda( G - G_{0}) \\
	& = A_{0} + 2 \lambda a \odot b + \lambda G_{V_{1}} \\
	& = \tilde{A} + \lambda G_{V_{1}}
	\end{align}
	
	\noindent where the third equality follows from \cref{eq:nullPushBas} and the fact that $\psi$ is an isometry.
	
	\cref{eq:nullPushR} implies that $\psi_{*}\tilde{r} = r$, also \cref{eq:nullPushBas} together with the fact that $\bp{a,\tilde{w}} = 0$ implies that $\psi_{*}\tilde{w} = \tilde{w}$. Thus we conclude that
	
	\begin{equation}
	\psi_{*}(\tilde{L} + \frac{\lambda}{\rho^{2}} G_{1}) = A + m r \odot r + 2 r \odot \tilde{w}
	\end{equation}
	
	\noindent where as a linear operator, $A = \tilde{A} + \lambda I_{V_{1}}$ where $I_{V_{1}}$ is the identity on $V_{1}$.
\end{proof}

\begin{remark}
	Note that even though the extended CT, $L$, can be naturally extended to all of $\eunn$. It is the extension of $\tilde{L}$ only for the subset $\Ima(\psi)$ of $\eunn$ given by \cref{thm:WPDecomps}, which is in general not a dense subset of $\eunn$.
\end{remark}

The following corollary will be useful in the sequel.

\begin{corollary} \label{cor:WPsphMet}
	Fix a proper warped product decomposition $\psi$ determined by the data $(V_{0} \obot V_{1}; a)$ with $\kappa_1 := a^2 = \pm 1$. Let $\tilde{r} = P_1 r$ be the dilatational vector in $W_1$ and $G_1$ be the metric in $W_1$. Write the metric adapted to the warped product as $G = G' + \frac{1}{\rho^{2}} \tilde{G}$, then:
	
	\begin{equation}
	\psi_* \tilde{G} = \kappa_1 \tilde{r}^2  (G_1 - \frac{1}{\tilde{r}^2} \tilde{r} \odot \tilde{r})
	\end{equation}
\end{corollary}
\begin{proof}
	Let $G$ be the contravariant metric for $\eunn$ and $G_{0} $ (resp. $G_1$) be the restriction of $G$ to $W_{0}$ (resp. $W_1$), then recall that
	
	\begin{align}
	\frac{1}{\rho^{2}} \tilde{G} & = G - G_{0} - \kappa_1 a \odot a
	\end{align}
	
	Hence the above equation together with \cref{eq:nonNullPushBas} implies that
	
	\begin{align}
	\psi_* \tilde{G}  & = \rho^2 (G - G_{0} - \kappa_1  \psi_* (a \odot a)) \\
	& = \rho^2 (G_1 - \kappa_1  \psi_* (a \odot a))
	\end{align}

	Let $\tilde{p}_1 = p_{1}-c \in W_1(\kappa_1)$ then $\tilde{r} =  P_1 r = \bp{a, p_0} \tilde{p}_1$. Then by \cref{eq:nonNullPushR}
	
	\begin{align}
	\psi_{*}a  & = \kappa_1 \tilde{p}_1 \\
	& = \kappa_1 \frac{\tilde{r}}{\bp{a, p_0}} \\
	& = \kappa_1 \frac{\tilde{r}}{\rho}
	\end{align}
	
	Thus since $\tilde{r}^2 = \frac{\rho^2}{\kappa_1} $, we have:
	
	\begin{align}
	\psi_* \tilde{G}  & = \rho^2 (G_1 - \kappa_1  \psi_* (a \odot a)) \\
	& = \rho^2 (G_1 - \kappa_1 \frac{1}{\rho^2} \tilde{r} \odot \tilde{r}) \\
	& = \kappa_1 \tilde{r}^2  (G_1 - \frac{1}{\tilde{r}^2} \tilde{r} \odot \tilde{r})
	\end{align}
\end{proof}

We now present some examples which show how to use the above proposition (\cref{prop:CtRedConstrEn}) to construct warped products which decompose a given reducible CT.

\begin{example} \label{ex:wpRedI}
	Let $M = \eunn$ where $n \geq 3$. Consider the central CT $L$ with parameter matrix $A = \varepsilon e \odot e$, where $\varepsilon := e^2 = \pm 1$.
	
	Let $W := e^\perp$ and $P$ be the orthogonal projection onto $W$. Choose $\bar{p} \in \eunn$ such that $(P \bar{p})^2 \neq 0$, WLOG we assume $(P \bar{p})^2 = \pm 1$. We construct a warped product passing through $\bar{p}$ which decomposes $L$.
	
	Let $\kappa_1 := \sgn (P \bar{p})^2$ and take $a := \kappa_1 P \bar{p} \in W$. Let $V_1 = W \cap a^\perp$ and $V_0 = V_1^\perp = \R e \obot \R a$. Note that $a$ was chosen so that the initial data $(\bar{p} ; V_0 \obot V_1; a)$ is in canonical form and also note that $\kappa_1 = a^2$. Let $\psi : N_{0} \times_{\rho} N_{1} \rightarrow \eunn$ be the warped product in \cref{thm:WPDecomps} determined by this initial data.
	
	Now let $\tilde{A} := \varepsilon e \odot e + 0 a \odot a \in C^2_0(N_0)$, then by construction we have that:
	
	\begin{equation}
	A = \tilde{A} + 0 I_{V_1}
	\end{equation}
	
	Let $\tilde{L}$ be the central CT in $N_0$ with parameter matrix $\tilde{A}$ and suppose the contravariant metric in the warped product decomposes as $G = G' + \frac{1}{\rho^{2}} G_{1}$. The above proposition shows that:
	
	\begin{equation}
	\psi_*(\tilde{L} + 0 \frac{1}{\rho^{2}} G_{1}) = L
	\end{equation}
	
	\noindent for all points in the image of $\psi$, which includes $\bar{p}$. Hence this warped product decomposition decomposes $L$. Note that this warped product was constructed so that $\tilde{A}$ has simple eigenvalues and so $\tilde{L}$ is no longer reducible.
	
	In the following we replace $N_1$ with $N_1 - c_1$ so that $N_1$ is a central hyperquadric. Then by \cref{eq:WPDecNn}, we have for $(p_0, p) = (\kappa_1 x a + y e ,p) \in N_{0} \times N_{1}$
	
	\begin{equation}
	\psi(p_{0},p) = x p + y e
	\end{equation}
\end{example}

The above example will be applied to construct separable coordinates in \cref{sec:constSepCoord}, see \cref{ex:oblProCoords}. We now give a non-Euclidean variation of the above example.

\begin{example} \label{ex:wpRedII}
	Let $M = \eunn$ where $n \geq 3$. Consider the central CT $L$ with parameter matrix $A = a \odot a$ with $a^2 = 0$ and $a \neq 0$.
	
	Let $W = a^\perp$. Choose $\bar{p} \notin W$, WLOG we assume $\bp{\bar{p}, a} = \pm 1$. We now construct a warped product passing through $\bar{p}$ which decomposes $L$.
	
	If $\bp{\bar{p}, a} = -1$, then set $a := - a$, so we can assume $\bp{\bar{p}, a} = 1$. Define $b$ as follows:
	
	\begin{eqnarray}
	b := \bar{p} - \frac{\bar{p}^2}{2} a
	\end{eqnarray}
	
	
	Note that $b$ is a lightlike vector satisfying $\bp{a,b} = 1$. Define $V_1 = a^\perp \cap b^\perp$ and $V_0 = \spa{a,b}$. Note that $b$ was chosen so that the initial data $(\bar{p} ; V_0 \obot V_1; a)$ is in canonical form. Let $\psi : N_{0} \times_{\rho} N_{1} \rightarrow \eunn$ be the warped product in \cref{thm:WPDecomps} determined by this initial data.
	
	Note that $\{b,a\}$ forms a cycle of generalized eigenvectors for $A$ and $A|_{V_1} = 0 I_{V_1}$. Hence by the above proposition, $(\psi^{-1})_* L$ is decomposable in this warped product. Also by \cref{thm:WPDecomps}, $\bar{p} \in \Ima(\psi)$. Also, the restriction of $(\psi^{-1})_* L$ to $N_0$, $\tilde{L}$, is a central CT with 2D parameter matrix $a \odot a$.
	
	In the following we replace $N_1$ with $P_1(N_1 - \bar{p})$ so that $N_1 = V_1$ is a vector space. Then by \cref{eq:WPDecN}, we have for $(p_0, p) = (x b + y a,p) \in N_{0} \times N_{1}$
	
	\begin{equation}
	\psi(p_{0},p) = x(b + p - \frac{1}{2} p^2 a) + y a
	\end{equation}
\end{example}

\paragraph{General Construction} \label{par:CtRedConstrEn} We will show how to use \cref{prop:CtRedConstrEn} to construct a warped product which decomposes an interesting class\footnote{This class includes all reducible OCTs in Euclidean and Minkowski space.} of non-degenerate reducible CTs. This construction generalizes the above examples. First we need a preliminary definition. Suppose $A$ is a linear operator on a vector space. We say that a vector $v$ is a \emph{proper generalized eigenvector of $A$} if $(A - \lambda I)^k v = 0$ for some $\lambda \in \C$ and $k > 1$.

Let $L = A + m r \odot r + 2 r \odot w$ be a non-degenerate CT in $\eunn$ in the canonical form given by \cref{thm:conTenCanForm}. We let the subspace $D$ and the matrix $A_c$ be as in the remarks following the theorem. We assume that each real generalized eigenspace of $A_c$ admits at most one proper generalized eigenvector. We lose no generality when working in Euclidean or Minkowski space \cite[Appendix~C]{Rajaratnam2014}.

Now let $W_{1},\dotsc,W_{k}$ be the multidimensional (real) eigenspaces of $A_c$ with corresponding eigenvalues $\lambda_1,\dotsc,\lambda_k$. The following construction is based on the metric-Jordan canonical form of $A_c$, see \cref{thm:comMetJFor} or \cite[theorem~C.3.7]{Rajaratnam2014}.

\begin{parts}
	\item $W_{i}$ is a non-degenerate subspace \\
	Choose a unit vector $a_{i} \in W_{i}$ and define $V_{i} := W_{i} \cap a_{i}^{\perp}$. The pair $(V_{i}, a_{i})$ determine a sphere.
	\item $W_{i}$ is a degenerate subspace \\
	Consider the metric-Jordan canonical form for $A_c$. By assumption there must be a single cycle $v_1,\dotsc,v_r$ of generalized eigenvectors with $v_r \in W_i$ being a lightlike eigenvector. Let $a_i := v_r$ and $V_{i} := W_{i} \cap v_1^{\perp}$, note that $V_i$ is non-degenerate.
\end{parts}

Now let $V_0 := \cap_{i=1}^{k} V_{i}^{\perp}$ and $\tilde{A} := A|_{V_0}$. By construction, the data $(\bigobot\limits_{i=0}^{k} V_{i}; a_{1},...,a_{k})$, determines a warped product decomposition $\psi : N_0 \times_{\rho_1} N_1 \cdots \times_{\rho_k} N_k \rightarrow \eunn$ in canonical form. By repeatedly applying \cref{prop:CtRedConstrEn} we see that $L$ is decomposable in the warped product decomposition induced by $\psi$, with the following properties:

\begin{itemize}
	\item $((\psi^{-1})_* L)|_{N_0} = \tilde{A} + m \tilde{r} \odot \tilde{r} + 2 \tilde{r} \odot w $ where $\tilde{r}$ is the dilatational vector field in $N_0$
	\item $\tilde{A}|_{D^\perp}$ only has eigenspaces of dimension one, i.e. each Jordan block of $\tilde{A}|_{D^\perp}$ has a distinct eigenvalue.
	\item For each $i > 0$, $T N_i$ is an eigenspace of $(\psi^{-1})_* L$ with constant eigenfunction $\lambda_i$
\end{itemize}


\paragraph{On Completeness } We will end this section by showing that the above construction is complete, meaning that the restriction of $(\psi^{-1})_* L$ to the geodesic factor $N_0$ no longer has constant eigenfunctions.

We also note here that with an appropriate choice of $a_1,\dotsc,a_k$ we can choose warped product decompositions to cover all of $\eunn$ except for a union of closed submanifolds with dimension strictly less than $n$. \Cref{ex:wpRedI,ex:wpRedII} give more details on how to do this, see also \cref{thm:WPDecomps}. In other words, for the non-degenerate CTs considered above, there exists a warped product decomposition $\psi : N_0 \times_{\rho_1} N_1 \cdots \times_{\rho_k} N_k \rightarrow \eunn$ such that $\Ima(\psi)$ is a dense subset of $\eunn$. Although the cost of this is that the factors $N_i$ may no longer be connected.

The following lemma shows that the classification of reducible CTs given above is complete for central CTs.

\begin{lemma}[Reducible central CTs] \label{prop:CtCentRed}
	Let $L$ be a central CT with parameter matrix $A$. Suppose that each real generalized eigenspace of $A$ has at most one proper generalized eigenvector. Then $A$ has a real eigenspace $\tilde{E}_\lambda$ with dimension $m > 1$ iff $L$ has a non-degenerate eigenspace $E_\lambda$ (defined on a dense subset of $\eunn$) with constant eigenfunction $\lambda$ and dimension $m-1$.
\end{lemma}
\begin{proof}
	It was proven above that under the hypothesis, if $A$ has a real eigenspace with dimension $m > 1$ then $L$ has a non-degenerate eigenspace $E_\lambda$ with dimension $m-1$. We will now prove the converse.
	
	To prove the converse, we simply have to prove that if all real eigenspaces of $A$ are at most one dimensional then $L$ has no non-degenerate eigenspaces with constant eigenfunctions defined on open subsets of $\eunn$. It is sufficient to show that $L$ has no constant eigenfunctions defined on open subsets of $\eunn$.
	
	We prove this by induction. The base cases are given by \cref{cor:CTChPolCentDe}. Suppose $U$ is a non-degenerate invariant subspace of $A$ such that $L_u$ has the form given by \cref{cor:CTChPolCentDe} and $U^\perp$ satisfies the induction hypothesis. By \cref{eq:CTDetCentInvSS} we can write:
	
	\begin{equation}
	p(z) = p_u(z)B_{u^\perp}(z) + B_{u}(z)(p_{u^\perp}(z) - B_{u^\perp}(z))
	\end{equation}
	
	Then
	
	\begin{equation}
	\d p = B_{u^\perp} \d p_u + B_{u} \d p_{u^\perp}
	\end{equation}
	
	By the induction hypothesis, $L_{u^\perp}$ has no constant eigenfunctions. Suppose $\lambda$ is a constant eigenfunction of $p$, then by \cref{cor:CTChPolCentDe} and the above equation, it follows that
	
	\begin{equation}
	B_{u^\perp}(\lambda)  = B_{u}(\lambda) = 0
	\end{equation}
	
	If $B_{u}$ has no real roots, we reach a contradiction. Otherwise, by construction $A$ must have a real eigenspace with dimension $m > 1$, a contradiction. Hence we conclude that $L$ has no constant eigenfunctions which proves the claim by induction.
\end{proof}

Since a multidimensional eigenspace of an OCT has a constant eigenfunction, the above proposition allows us to classify these eigenspaces when the CTs considered induce an OCT on some subset of $\eunn$. For completeness, we will show that the hypothesis of the above proposition is the most general for classifying OCTs.

\begin{propMy}
	Let $L$ be a central CT with parameter matrix $A$. Suppose $A$ has a real generalized eigenspace with multiple proper generalized eigenvectors, then $L$ is not an OCT.
\end{propMy}
\begin{proof}
	WLOG we can assume that that this generalized eigenspace of $A$ is associated with the eigenvalue zero. First we have
	
	\begin{align}
	L & = A + r \odot r \\
	L^2 & = A^2 + A r \odot r + r^2 r \odot r
	\end{align}
	
	By hypothesis, $\dim N(L) \geq 1$. We also have that $\dim N(A^2) \geq 4$. The above equation shows that the range of $L^2$ is spanned by $\{r, Ar\}$ and the range of $A^2$ (on a dense subset of $\eunn$), hence we see that $\dim N(L^2) \geq 1 + \dim N(L)$. This implies that $L$ is not point-wise diagonalizable on some dense subset of $\eunn$ (see for example \cite{friedberg2003linear}).
\end{proof}

In fact one can show that if $A = J_2(0) \oplus J_2(0)$, then the associated central CT has a $2$-cycle of generalized eigenvectors associated with eigenvalue zero.

The following lemma is the analogue of \cref{prop:CtCentRed} for axial CTs. Its proof is also analogous and reduces to \cref{prop:CtCentRed} with the help of \cref{eq:CTChPolAxial} and \cref{prop:CTDetIrredAxial}. 

\begin{lemma}[Reducible axial CTs]
	Let $L$ be an axial CT with parameter matrix $A$. Suppose that each real generalized eigenspace of $A_c$ has at most one proper generalized eigenvector. Then $A_c$ has a real eigenspace $\tilde{E}_\lambda$ with dimension $m > 1$ iff $L$ has a non-degenerate eigenspace $E_\lambda$ (defined on a dense subset of $\eunn$) with constant eigenfunction $\lambda$ and dimension $m-1$.
\end{lemma}

In conclusion we have the following theorem which summarizes our classification:

\begin{thmMy}[Classification of Reducible CTs in $\eunn$] \label{thm:classRedCtEunn}
	Let $L$ be a non-degenerate CT in $\eunn$ such that each real generalized eigenspace of $A_c$ has at most one proper generalized eigenvector. Then $L$ is reducible iff $A_c$ has a multidimensional real eigenspace. If $L$ is reducible, then there exists an explicitly constructible warped product decomposition $\psi : N_0 \times_{\rho_1} N_1 \cdots \times_{\rho_k} N_k \rightarrow \eunn$ such that the following hold:
	\begin{itemize}
		\item $L$ is decomposable in the warped product $N_0 \times_{\rho_1} N_1 \cdots \times_{\rho_k} N_k$.
		\item The restriction of $(\psi^{-1})_* L$ to $N_0$ has no constant eigenfunctions.
		\item $\Ima(\psi)$ is an open dense subset of $\eunn$.
	\end{itemize}
\end{thmMy}

\subsection{In Spherical submanifolds of pseudo-Euclidean space} \label{sec:CtClassRedEunnKap}

In this section we show how the problem of classifying reducible CTs in $\eunn(\kappa)$ can be reduced to the same problem in $\eunn$; we will assume $n > 2$ to avoid trivial cases. First we will need to obtain the warped product decompositions of $\eunn(\kappa)$. The following proposition shows that any proper warped product decomposition of $\eunn$ in canonical form restricts to a warped product decomposition of $\eunn(\kappa)$. Its proof is straightforward consequence of \cref{eq:wpdecpSp}; see \cite[Appendix~D]{Rajaratnam2014} for more details.

\begin{theorem}[Restricting Warped products to $\eunn(\kappa)$] \label{thm:eunnKapRestWP}
	Let $\psi$ be a proper warped product decomposition of $\eunn$ associated with $(\bar{p}; \bigobot\limits_{i=0}^{k} V_{i}; a_{1},...,a_{k})$ in canonical form. Suppose $\kappa^{-1} := \bar{p}^{2} \neq 0$ and let $N' := N_{0}(\kappa) \times_{\rho_{1}} N_{1} \times \cdots \times_{\rho_{k}} N_{k}$. Then $\phi : N' \rightarrow \eunn(\kappa)$ defined by $\phi := \psi|_{N'}$ is a warped product decomposition of $\eunn(\kappa)$ passing through $\bar{p}$.
	
	
\end{theorem}
\begin{remark}
	Sometimes $N_{0}(\kappa)$ may not be connected, for more details on this see \cite[Appendix~D]{Rajaratnam2014}.
\end{remark}

Now we show how to restrict a reducible CT in $\eunn$ to one in $\eunn(\kappa)$.

\begin{propMy}[Restricting Reducible CTs to $\eunn(\kappa)$] \label{thm:eunnKapRestCT}
	Let $\psi : N_0 \times_{\rho_1} N_1 \cdots \times_{\rho_k} N_k \rightarrow \eunn$ be a proper warped product decomposition in canonical form and let $\bar{p} \in \Ima(\psi)$ as in the above theorem. Suppose $L_c$ is a reducible central CT in $\eunn$ satisfying
	\begin{equation}
	L_c = \psi_{*} (\tilde{L}_c + \sum_{i=1}^{k} \lambda_{i} G_{i})
	\end{equation}
	
	\noindent where $G_i$ is the restriction of $G$ to $T N_i$, $\lambda_i \in \R$ and $\tilde{L}_c$ is a CT in $N_0$. Let $\phi := \psi|_{N'}$ be the induced warped product decomposition of $\eunn(\kappa)$ as in the above theorem. Then if we let $L$ (resp. $\tilde{L}$) be the restriction of $L_c$ (resp. $\tilde{L}_c$) to $\eunn(\kappa)$ (resp. $N_0(\kappa)$), then
	
	\begin{equation}
	L = \phi_{*} (\tilde{L} + \sum_{i=1}^{k} \lambda_{i} G_{i})
	\end{equation}
\end{propMy}
\begin{proof}
	Let $\tilde{r}$ (resp. $r$) be the dilatational vector field in $N_0$ (resp. $\eunn$). We will use the fact that $\psi_{*} \tilde{r} = r$; this can be deduced from the proof of \cref{prop:CtRedConstrEn} or \cref{eq:genPsiEqn}. We let $R^{*} = I - \dfrac{r \otimes r^{\flat}}{r^{2}}$ be the orthogonal projection onto $T \eunn(\kappa)$ with a similar definition for $\tilde{R}^{*}$ with respect to $T N_0(\kappa)$. In the following, given $L \in S^2(\eunn)$, we denote by $R^*L$ the restricted tensor given by $(R^*L)^{ij} = R \indices{^i_l} L^{lk} R\indices{^j_k}$.
	
	Using the fact that $\psi$ is an isometry and $\psi_{*} \tilde{r} = r$, one can show that $R^{*} \circ \psi_{*} = \psi_{*} \circ \tilde{R}^{*}$. Also note that $\tilde{R}^{*}G_{i} = G_{i}$. Thus
	
	\begin{align}
	R^{*}L_c & = R^{*}\psi_{*} (\tilde{L}_c + \sum_{i=1}^{k} \lambda_{i} G_{i}) \\
	& = \psi_{*} (\tilde{R}^{*}\tilde{L}_c + \sum_{i=1}^{k} \lambda_{i} \tilde{R}^{*}G_{i}) \\
	& = \psi_{*} (\tilde{R}^{*}\tilde{L}_c + \sum_{i=1}^{k} \lambda_{i} G_{i})
	\end{align}
	
	By evaluating the above equation in $N_0(\kappa) \times_{\rho_1} N_1 \cdots \times_{\rho_k} N_k$, one obtains the desired result.
\end{proof}

Now we show how to apply the above results to obtain a warped product decomposition in which a given CT in $\eunn(\kappa)$ is decomposable. Let $L$ be a non-trivial CT in $\eunn(\kappa)$, then there is a unique central CT, $L_c$, such that $L = R^* L_c$ . As described in the previous section, provided $L_c$ is reducible, we can choose a warped product decomposition of $\eunn$, $\psi$, such that $L_c = \psi_{*} (\tilde{L}_c + \sum_{i=1}^{k} \lambda_{i} G_{i})$ satisfying the hypothesis of the above proposition. Thus the above proposition gives a warped product decomposition $\phi$ which decomposes $L$, and is obtained by an appropriate restriction of $\psi$. We now give some examples of this procedure to obtain the standard spherical coordinates.

\begin{example}[Spherical Coordinates I] \label{ex:wpRedIII}
	Let $M = \eunn(\kappa)$ where $\kappa = \pm 1$ and $n \geq 3$. Consider the CT $L$ in $\eunn(\kappa)$ induced by $A = \varepsilon e \odot e$ with $\varepsilon := e^2 = \pm 1$. Let $P$ be the orthogonal projector onto $e^\perp$ and choose $\bar{p} \in \eunn(\kappa)$ such that $(P \bar{p})^2 = \pm 1$. By \cref{ex:wpRedI} there is a warped product decomposition $\psi : N_{0} \times_{\rho} N_{1} \rightarrow \eunn$ passing through $\bar{p}$ which decomposes $L_c := A + r \odot r$. For $(p_0, p) = (x \kappa_1 a + y e,p) \in N_{0} \times N_{1}$, we have
	
	\begin{equation}
	\psi(p_{0},p) = x p + y e
	\end{equation}
	
	To obtain a warped product decomposition of $\eunn(\kappa)$, by \cref{thm:eunnKapRestWP} we need to restrict $\psi$ to $N_{0}(\kappa) \times N_{1}$. Let $\phi$ be the induced warped product decomposition of $\eunn(\kappa)$, then it follows by \cref{thm:eunnKapRestCT} that $L$ is decomposable in this warped product. We now give the standard forms of this warped product by parameterizing $(x,y)$ as in \cref{ex:coordEunKap} while enforcing $x = \bp{a, p_0} > 0$ and $N_0(\kappa)$ to be connected. We have three different cases:
	
	%
	%
	
	\begin{parts}
		\item $\kappa_1 = \kappa$ and $\varepsilon = \kappa$
		
		\begin{equation}
		\phi : \begin{cases}
		\left (0, \pi \right ) \times _{\sin} N_1 & \rightarrow \E_\nu^{n}(\kappa) \\
		(t,p) & \mapsto \sin(t) p + \cos(t) e
		\end{cases}
		\end{equation}
		\item $\kappa_1 = \kappa$ and $\varepsilon = -\kappa$ 
		
		\begin{equation}
		\phi : \begin{cases}
		\R \times _{\cosh} N_1 & \rightarrow \E_\nu^{n}(\kappa) \\
		(t,p) & \mapsto \cosh(t) p + \sinh(t) e
		\end{cases}
		\end{equation}
		\item $\kappa_1 = -\kappa$ and $\varepsilon = \kappa$
		
		\begin{equation}
		\phi : \begin{cases}
		\R^+ \times _{\sinh} N_1 & \rightarrow \E_\nu^{n}(\kappa) \\
		(t,p) & \mapsto \sinh(t) p + \cosh(t) e
		\end{cases}
		\end{equation}
	\end{parts}
	
	Note that even though there is only one inequivalent coordinate system on $\E^2_\nu (\kappa)$, the last two warped products are inequivalent. This is due to the fact that $a^2 = \kappa_1$ is different in these cases and $N_0 = \{p \in V_{0} | \bp{a, p} > 0  \}$.
\end{example}

The following example considers spherical coordinates that only occur in non-Euclidean spheres.

\begin{example}[Spherical Coordinates II] \label{ex:wpRedIV}
	Let $M = \eunn(\kappa)$ where $\kappa = \pm 1$ and $n \geq 3$. We now consider the CT $L$ in $\eunn(\kappa)$ induced by $A = a \odot a$ with $a^2 = 0$ and $a \neq 0$. This example proceeds similarly to the first. Fix $\bar{p} \in \eunn(\kappa)$ such that $\bp{a , \bar{p}} = 1$. By \cref{ex:wpRedII} there is a warped product decomposition $\psi : N_{0} \times_{\rho} N_{1} \rightarrow \eunn$ passing through $\bar{p}$ which decomposes $L_c := A + r \odot r$. For $(p_0, p) = (x b + y a,p) \in N_{0} \times N_{1}$, we have
	
	\begin{equation}
	\psi(p_{0},p) = x(b + p - \frac{1}{2} p^2 a) + y a
	\end{equation}
	
	Restricting $\psi$ to $N_{0}(\kappa) \times N_{1}$ forces:
	
	\begin{equation}
	\kappa = p_0^2  =  2 x y
	\end{equation}
	
	Let $\phi$ be the warped product decomposition of $\eunn(\kappa)$ induced by $\psi$ as in \cref{thm:eunnKapRestWP}. Again, it follows by \cref{thm:eunnKapRestCT} that $L$ is decomposable in this warped product. We now give $\phi$ with the standard parameterization of $N_0(\kappa)$, by enforcing $x = \bp{a, p_0} > 0$ and $N_0(\kappa)$ to be connected. These conditions are all satisfied if we take $x = \frac{1}{\sqrt{2}} \exp(t)$. Then we have the following:
	
	\begin{equation}
	\phi : \begin{cases}
	\R \times _{\frac{1}{\sqrt{2}} \exp} \E_{\nu-1}^{n-2} & \rightarrow \E_{\nu}^{n}(\kappa) \\
	(t,p) & \mapsto \frac{1}{\sqrt{2}} \exp(t)  (b + p  - \frac{1}{2}p^{2}a) + \frac{\kappa}{\sqrt{2}} \exp(-t) a
	\end{cases} 
	\end{equation}
	
	Also note that if $\nu = - \kappa = 1$, then $\phi$ is an isometry onto a connected component of $\E_{1}^{n}(-1) \simeq H^{n-1}$.
\end{example}

In conclusion we have the following theorem which summarizes our classification:

\begin{thmMy}[Classification of Reducible CTs in $\eunn(\kappa)$]
	Let $L$ be a non-trivial CT in $\eunn(\kappa)$ with $n > 2$ such that each real generalized eigenspace of $A$ has at most one proper generalized eigenvector. Then $L$ is reducible iff $A$ has a multidimensional real eigenspace. If $L$ is reducible, then there exists an explicitly constructible warped product decomposition $\psi : N_0 \times_{\rho_1} N_1 \cdots \times_{\rho_k} N_k \rightarrow \eunn(\kappa)$ such that the following hold:
	\begin{enumerate}
		\item $L$ is decomposable in the warped product $N_0 \times_{\rho_1} N_1 \cdots \times_{\rho_k} N_k$.
		\item \label{it:clasRedCnstEi} The restriction of $(\psi^{-1})_* L$ to $N_0$ has no constant eigenfunctions.
		\item \label{it:clasReddens} $\Ima(\psi)$ is an open dense subset of $\eunn(\kappa)$.
	\end{enumerate}
\end{thmMy}
\begin{proof}
	We give the proof of \cref{it:clasRedCnstEi}. First suppose $\lambda$ is a constant eigenfunction of $L$, then one can naturally lift $\lambda$ to a constant function on $\eunn$. Let $p(z)$ be the characteristic polynomial of $L$ having the form given by \cref{eq:CTChPolEunnKap}. Then since $\lied{p}{r} = 0$ (see the proof of \cref{prop:IctMetricEunnKap}), we must have $p(\lambda) = 0$ on some open subset of $\eunn$. Then the proof of  \cref{prop:CtCentRed} holds verbatim by \cref{eq:CTChPolEunnKap}, which proves the result.
	
	\Cref{it:clasReddens} follows from the construction of $\psi$ (see \cref{thm:eunnKapRestCT}) and \cref{thm:classRedCtEunn}.
\end{proof}

\section{Applications and Examples} \label{sec:appNEx}

In this section we show how to apply the theory developed in this article to solve some of the motivating problems stated in the introduction. First, in \cref{sec:enumIneqCoord} we show how to enumerate the isometrically inequivalent separable coordinates in a given space of constant curvature. Then in \cref{sec:constSepCoord} we show how to construct separable coordinate systems by way of examples. Finally, in \cref{sec:BEKMsep} we show how to explicitly execute the BEKM separation algorithm in general. We also give the details of executing the BEKM separation algorithm for the Calogero-Moser system.

\subsection{Enumerating inequivalent separable coordinates} \label{sec:enumIneqCoord}

In this section we show how one can use the theory developed in this article to enumerate the isometrically inequivalent separable coordinate systems on a given space of constant curvature. For dimensions greater than two, this problem is recursive as described in \cite[section~6.2]{Rajaratnam2014a}. This recursive nature was originally discovered by Kalnins et al. and is discussed more concretely in \cite{Kalnins1986}. So one will also have to enumerate the separable coordinate systems on spherical submanifolds of the underlying space and then construct the separable coordinates systems using warped products (see the beginning of \cref{sec:sumRes} and also \cite[section~6.2]{Rajaratnam2014a}).

The main step is to enumerate the geometrically inequivalent CTs, so we will focus on this. To do this, one has to enumerate the canonical forms summarized in \cref{sec:sumRes} together with the metric-Jordan canonical forms for $A_c$ and take into account geometric equivalence. We illustrate this idea with some examples.

\begin{example}[Central CTs]
	Let $L$ be a central CT with parameter matrix $A$. In this case, we essentially have to enumerate the different metric-Jordan canonical forms for $A$. Fix $\lambda_1 < \cdots < \lambda_n \in \R$.
	
	In Euclidean space there is only one central CT we can build from these parameters; it is given by the parameter matrix $A = \diag(\lambda_1,\dotsc,\lambda_n)$ and it induces the well known elliptic coordinate system (see \cref{ex:genElipCoord}).
	
	In Minkowski space there are $n$ (geometrically inequivalent) central CTs we can build from these parameters, they are given as follows:
	
	\begin{align}
	A & = J_{-1}(\lambda_1) \oplus J_{1}(\lambda_2)  \oplus \cdots J_{1}(\lambda_n) \\
	\vdots \\
	A & = J_{1}(\lambda_1) \oplus J_{1}(\lambda_2) \oplus \cdots J_{-1}(\lambda_n)
	\end{align}
	
	They differ by the eigenvalue of $A$ which is timelike. Similarly there are $n-1$ central CTs built only using $\lambda_2 < \cdots < \lambda_n$ with parameter matrix of the form:
	
	\begin{equation}
	A = J_{\pm 2}(\lambda_2) \oplus J_{1}(\lambda_3) \oplus \cdots J_{1}(\lambda_n)
	\end{equation}
	
	Now consider the case where $A$ has a two dimensional eigenspace, the rest being simple. Using $\lambda_2 < \cdots < \lambda_n$, in Euclidean space there are $n - 1$ central CTs depending on which $\lambda_i$ corresponds to the two dimensional eigenspace\footnote{When $n = 3$ the two different cases induce the oblate and prolate spheroidal coordinate systems.}. Each of these cases in Euclidean space induce $n-1$ different cases in Minkowski space depending on which $\lambda_i$ becomes timelike, hence there are a total of $(n-1)^2$ cases in Minkowski space.
	
	Finally we note that in Minkowski space $A$ can have two complex conjugate eigenvalues, then since the corresponding real Jordan block is distinguishable from the other real eigenvalues of $A$, a similar analysis applies. In general one would have to order the complex eigenvalues (see \cref{defn:ordC}).
\end{example}

Enumerating inequivalent axial CTs can largely be reduced to the same problem for central CTs. For example, in Euclidean space there is only one type of axial CT if all the eigenvalues of $A_c$ are distinct. We conclude with CTs in spherical submanifolds of pseudo-Euclidean space as they are somewhat different.

\begin{example}[CTs in $\eunn(\kappa)$]
	Let $L$ be the CT in $\eunn(\kappa)$ with parameter matrix $A$. Fix $\lambda_1 < \cdots < \lambda_n \in \R$. In this case there are sometimes less geometrically inequivalent CTs then isometrically inequivalent ones.
	
	In the Euclidean sphere there is only one CT we can build from these parameters, it is given by the parameter matrix $A = \diag(\lambda_1,\dotsc,\lambda_n)$ and it induces the sphere-elliptic coordinate system.
	
	Now suppose the ambient space is Minkowski space. Then we only need to consider $\lceil \frac{n}{2} \rceil$ cases given by (see \cref{ex:geoEquivHypSpa}):
	
	\begin{align}
	A & = J_{-1}(\lambda_1) \oplus J_{1}(\lambda_2)  \oplus \cdots J_{1}(\lambda_n) \\
	\vdots \\
	A & = J_{1}(\lambda_1) \oplus J_{1}(\lambda_2) \oplus \cdots \oplus J_{-1}(\lambda_{\lceil \frac{n}{2} \rceil}) \oplus \cdots J_{1}(\lambda_n)
	\end{align}
	
	Note that only the first $\lceil \frac{n}{2} \rceil$ eigenvalues of $A$ are made timelike.
	
	Most of the other cases can be deduced from the first example if one desires. Although we illustrate one difference with an example. For the Euclidean sphere $\E^3(1)$, fix $\lambda_1 < \lambda_2 \in \R$ and consider the CT induced by the following parameter matrices:
	
	\begin{align}
	A_1 & = \diag(\lambda_1, \lambda_1, \lambda_2) \\
	A_2 & = \diag(\lambda_1, \lambda_2, \lambda_2)
	\end{align}
	
	Note that $-A_2$ has the same form as $A_1$, specifically the smallest eigenvalue of $-A_2$ is repeated. Hence in considering parameter matrices with two dimensional eigenspaces, we only need to enumerate those with the form given by $A_1$, where the smaller eigenvalue is repeated.
\end{example}

We have described how to enumerate the geometrically inequivalent CTs in spaces of constant curvature. One should note though, that in non-Euclidean spaces a given CT could induce different coordinate systems on disjoint connected subsets of the space (see \cref{ex:CCTCoordMink}). Hence in these cases, more work has to be done to enumerate the isometrically inequivalent separable coordinate systems.

\subsection{Constructing separable coordinates} \label{sec:constSepCoord}

In a two dimensional Riemannian manifold, all non-trivial CTs are Benenti tensors. Hence in this case, one can enumerate all isometrically inequivalent separable coordinates simply by enumerating the geometrically inequivalent CTs. The latter problem can be solved in pseudo-Euclidean space using \cref{thm:conTenCanForm}. In \cref{tab:E2} we have done this for $\E^2$ and included the standard transformations from separable to Cartesian coordinates.

\begin{table}[h]
	\caption{Separable Coordinate Systems in $\E^2$}
	\label{tab:E2}
	\centering
	\medskip
	\begin{tabular}{| l | l | l |}
		\hline
		1. Cartesian coordinates &  $L = d \odot d $ & $x \: d + y \: e$ \\
		2. Polar coordinates &   $L =  r \odot r $ & $\rho \cos \theta \: d + \rho \sin \theta \: e$ \\
		3. Elliptic coordinates &  $L = d \odot d + a^{-2} r \odot r $ & $a\cos \phi \cosh \eta \: d + a\sin \phi \sinh \eta \: e$  \\
		4. Parabolic coordinates &  $L = 2 r \odot d $ & $\frac{1}{2}(\mu^{2} - \nu^{2}) \: d + \mu\nu \: e$ \\
		\hline
	\end{tabular} \\[1.5pt]
	The vectors $d,e$ form an orthonormal basis for $\E^2$ and $a > 0$.
\end{table}

We now show how one obtains the coordinate formula in \cref{tab:E2} from formulas we have already calculated. For elliptic coordinates, take Cartesian coordinates $(x,y)$ on $\E^2$ and let $L$ be the central CT with parameter matrix $A = \diag(\lambda_1, \lambda_2)$ where $\lambda_2 > \lambda_1$. Then the transformation from canonical coordinates $(u^1 , u^2)$ to Cartesian coordinates $(x,y)$ read (see \cref{eq:CTCoordsCCI}):

\begin{align}
x^2 & = \frac{(\lambda_1 - u^1)(\lambda_1 - u^2)}{(\lambda_2 - \lambda_1)} & y^2 & = \frac{(\lambda_2 - u^1)(\lambda_2 - u^2)}{(\lambda_1 - \lambda_2)}
\end{align}

We can obtain the standard parameterization of elliptic coordinates as follows. Note that $L = \lambda_1 G + (\lambda_2 - \lambda_1) \tilde{L}$ where $\tilde{L} = e \odot e + (\lambda_2 - \lambda_1)^{-1} r \odot r $ is geometrically equivalent to $L$. The eigenfunctions of $\tilde{L}$, $(\tilde{u}^1 , \tilde{u}^2)$, are related to those of $L$ by $u^i = \lambda_1 + (\lambda_2 - \lambda_1) \tilde{u}^i$. Letting $a^2 := \lambda_2 - \lambda_1$ and substituting this expression for $u^i$ in the above equation gives:

\begin{align}
x^2 & = a^2 \tilde{u}^1 \tilde{u}^2 & y^2 & = a^2 (1-\tilde{u}^1) (\tilde{u}^2 - 1)
\end{align}

Then making the transformation $\tilde{u}^1 = \cos^2 \phi$ and $\tilde{u}^2 = \cosh^2 \eta$, we obtain the formula in \cref{tab:E2}.

The formula for parabolic coordinates follow similarly from \cref{eq:CTCoordsACI,eq:CTCoordsACII}, after taking $u^1 = - \nu^{2}$ and $u^2 = \mu^{2}$ assuming $u^1 < u^2$.

We end with a few more examples to further illustrate the theory. The first example shows how to obtain coordinates which diagonalize a Benenti tensor which is not an ICT.

\begin{example}[Spherical coordinates in $\Si^2$] \label{ex:sphCoord}
	Fix $d \in \Si^2$ and let $L$ be the CT induced in $\Si^2$ by restricting $d \odot d$. As we observed earlier, $L$ is necessarily a Benenti tensor. In \cref{ex:wpRedIII} it was shown that a warped product which decomposes $L$ is given by:
	
	\begin{equation}
	\psi(\phi, p) = \cos \phi \: d + \sin \phi \: p
	\end{equation}
	
	\noindent where $p \in d^\perp(1)$, i.e. $p \in \Si^2 \cap d^\perp$ and $\phi \in (0, \pi)$. Since $d^\perp(1)$ is the unit circle we obtain coordinates on it by taking $p = \cos \theta \: e + \sin \theta \: f$ where $e,f$ is an orthonormal basis for $d^\perp$. Then the above equation becomes:
	
	\begin{equation}
	\psi(\phi, p) = \cos \phi \: d + \sin \phi (\cos \theta \: e + \sin \theta \: f)
	\end{equation}
	
	Furthermore, since $\psi$ is a warped product decomposition with warping function $\sin \phi$, it follows from \cref{ex:wpRedIII} that the metric is:
	
	\begin{equation}
	g = (\d \phi)^2 + \sin^2 \phi (\d \theta)^2
	\end{equation}
\end{example}

\begin{example}[Oblate/Prolate spheroidal coordinates in $\E^3$] \label{ex:oblProCoords}
	Fix a unit vector $d \in \E^n$, $c \neq 0$ and consider the following CT in $\E^n$:
	
	\begin{equation} \label{eq:exCTobPro}
	L = c \: d \odot d + r \odot r
	\end{equation}
	
	It follows from \cref{ex:wpRedI} that a warped product $\psi$ which decomposes $L$ is given as follows: Let $e \in d^\perp$ be a unit vector, then for  $(p_0, p) = ( x d + y e ,p) \in N_{0} \times N_{1}$
	
	\begin{equation}
	\psi(p_{0},p) = x d + y p
	\end{equation}
	
	Observe that $N_0 \simeq \E^2$ and $L$ induces a Benenti tensor, $\tilde{L}$, on $N_0$ which has the form given by \cref{eq:exCTobPro}. If we let $a := \sqrt{\Abs{c}}$, then using \cref{tab:E2} we can take coordinates on $N_0$ which diagonalize $\tilde{L}$ yielding the following maps.
	
	\begin{equation} \psi(p_{0},p) = 
	\begin{cases}
	c > 0 &  a\cos \phi \cosh \eta \: d + a\sin \phi \sinh \eta \: p \\
	c < 0 &  a\sin \phi \sinh \eta \: d + a\cos \phi \cosh \eta \: p
	\end{cases}
	\end{equation}
	
	Also $N_1$ is the unit sphere in $d^\perp$, hence $N_1 \simeq \Si^{n-2}$. We can obtain separable coordinates for $\E^n$ by taking any separable coordinates for $\Si^{n-2}$ on $N_1$ \cite{Rajaratnam2014a}. For example, if $c > 0$ and $n = 3$, we obtain prolate spheroidal coordinates:
	
	\begin{equation}
	\psi(p_{0},p) = a\cos \phi \cosh \eta \: d + a\sin \phi \sinh \eta \: (\cos \theta \: e + \sin \theta \: f)
	\end{equation}
	
	\noindent where $e,f$ is any orthonormal basis for $d^\perp$. Also note that using \cref{prop:IctMetricEunn} and the fact that $\psi$ is a warped product decomposition with warping function $a\sin \phi \sinh \eta$, one can obtain the following expression for the metric:
	
	\begin{equation}
	g = a^2(\sinh^2 \eta + \sin^2 \phi)((\d \phi)^2 + (\d \eta)^2) + a^2 \sin^2 \phi \sinh^2 \eta (\d \theta)^2
	\end{equation}
	
	Finally note that oblate spheroidal coordinates can be obtained by taking $c < 0$.
\end{example}

\begin{example}[Product coordinates in $\E^4$]
	Consider the decomposition $\E^n = V \obot W$ into non-trivial subspaces. Let $\tilde{G}$ denote the induced contravariant metric in $V$ and consider the following CT in $\E^n$:
	
	\begin{equation}
	L = \tilde{G}
	\end{equation}
	
	Observe that the warped product $\psi : V \times_1 W \rightarrow \E^n$ given by $(q,p) \rightarrow q + p$ is adapted to the eigenspaces of $L$. We can construct separable coordinates by parameterizing $q$ (resp. $p$) with separable coordinates on $V$ (resp. $W$). For example, if $\dim V = \dim W = 2$, by taking polar (resp. elliptic) coordinates on $V$ (resp. $W$) from \cref{tab:E2}, we have the following separable coordinates on $\E^4$:
	
	\begin{equation}
	\psi(q,p) = \rho \cos \theta \: b + \rho \sin \theta \: c + a\cos \phi \cosh \eta \: d + a\sin \phi \sinh \eta \: e
	\end{equation}
	
	\noindent where $b,c$ (resp. $d,e$) is an orthonormal basis for $V$ (resp. $W$).
\end{example}

Extending the above analysis one can prove that there are eleven classes of isometrically inequivalent separable coordinate systems in $\E^3$.

\subsection{The BEKM separation algorithm} \label{sec:BEKMsep}

In this section we show how to execute the BEKM separation algorithm (see \cite[section~6.3]{Rajaratnam2014a} for details) in spaces of constant curvature using the classification of CTs given in this article.

In order to execute this algorithm in $\eunn$ we will need the \gls{kbd} equation in $\eunn$ and in $\eunn(\kappa)$. Fix a function $V \in \F(\eunn)$ and suppose $n > 1$. Then if $L$ is the general CT in $\eunn$ given by \cref{eq:CTGenEunn} and $K_e := \tr{L}G - L$ is its KBDT, then the KBD equation in $\eunn$ is:

\begin{equation}
\d (K_e \d V) = 0
\end{equation}

We will often refer to the above equation as just the \emph{KBD equation}.

It will be convenient to evaluate the KBD equation in $\eunn(\kappa)$ via its embedding in $\eunn$. Then if $\tilde{L}$ is the general CT in $\eunn(\kappa)$ given in $\eunn$ by \cref{eq:CTGenEunnKap}, let $L := r^2 \tilde{L}$ and $K_s := \tr{L}R - L$, then the KBD equation in $\eunn(\kappa)$ (embedded in $\eunn$) is:

\begin{equation} \label{eq:KBDsph}
\d (K_s \d V) = 0
\end{equation}

We will often refer to the above equation as the \emph{spherical KBD equation}. We will show how this equation is derived in \cref{par:sphKBDeq}.

We should also mention here that we carry out the BEKM separation algorithm slightly differently than described in \cite[section~6.3]{Rajaratnam2014a}. We construct warped products which decompose reducible OCTs such that the induced CT on the geodesic factor is an ICT as opposed to a Benenti tensor. This allows one to simultaneously construct separable coordinates while carrying out the algorithm, as illustrated by the following example.

\subsubsection{Example: Calogero-Moser system} \label{ex:CalMosSys}

We first present an example which separates in several different coordinate systems and hence provides a good example for the BEKM separation algorithm. Our example is the Calogero-Moser system, which will be defined shortly. Another advantage of this example is that its separability properties have been studied by several different authors \cite{Horwood2005,Rauch-Wojciechowski2005,Waksjo2003,Benenti2000,Calogero1969}, hence it allows one to compare and contrast different methods. Finally we mention that we obtained this example from \cite{Waksjo2003} where an algorithm equivalent to the BEKM separation algorithm was used to study this example.

The $n$-dimensional Calogero-Moser system is given by the following natural Hamiltonian \cite{Calogero2008}:

\begin{equation} \label{eq:CMHam}
\tag{CM}
H\left( p,q\right) =\frac{1}{2}\sum_{i=1}^{n}\left(p_{i}^{2}+\omega ^{2}q_{i}^{2}\right) + \sum_{1 \leq i < j \leq n} \frac{g^{2}}{(q_{i}-q_{j})^{2}}
\end{equation}

We will take $\omega = 0, g = 1$ for convenience. In this case this Hamiltonian models $n$ point particles moving on a line acted on by forces depending on their relative distances. We can write the potential $V$ as follows:

\begin{equation}
V = \sum_{i} \bp{r,a_{i}}^{-2}
\end{equation}

\noindent where $a_{i} = e_{k} - e_{l}$ for some $k,l \in \{1,\dotsc,n\}$ with $e_i := \partial_i$. Furthermore we let

\begin{equation}
d = \frac{1}{\sqrt{n}} \sum_{i=1}^{n} e_{i}
\end{equation}

We can obtain solutions to the KBD equation by using the following result.

\begin{propMy} \label{prop:extCTKBDsol}
	Suppose $L = A + m r \odot r + 2 w \odot r$ is a CT in $\eunn$ and let $\tilde{L}$ be the restriction of $L$ to $\eunn(\kappa)$. Let $a$ be a covariantly constant vector and let $V := \bp{r,a}^{-2}$. If $a$ is an eigenvector of $A$ orthogonal to $w$ then $V$ satisfies the KBD equation with $L$ in $\eunn$. If $a$ is an eigenvector of $A$ then the restriction of $V$ to $\eunn(\kappa)$ satisfies the KBD equation with $\tilde{L}$ in $\eunn(\kappa)$.
\end{propMy}
\begin{proof}
	We first consider the case in $\eunn$. Under these hypothesis it follows by \cref{prop:extCT} that if $\rho := \Abs{\bp{r,a}}$, then we have:
	
	\begin{equation}
	L(\d \log \rho) = \d ( \lambda  \log \rho + \frac{1}{2} \tr{L})
	\end{equation}
	
	\noindent for some $\lambda \in \R$. From the above equation one can check that $L$ satisfies the KBD equation with $V$. A similar proof holds for the case in $\eunn(\kappa)$, but now the above equation with $\tilde{L}$ follows either by restriction of the one in the ambient space or by \cref{thm:eunnKapRestCT} together with \cref{eq:CTredExt} from \cref{prop:CTredChar}.
\end{proof}
\begin{remark}
	This result comes from the connection between extending KTs into warped products and the separation of the Hamilton-Jacobi equation for natural Hamiltonians \cite{Rajaratnam2014a}. One can show that the commuting integrals can be explicitly calculated; this is a consequence of the fact that $L$ is torsionless.
\end{remark}
\begin{remark}
	One can naturally construct separable potentials from the above proposition. For example if $a_1,\dotsc,a_n$ is an orthonormal basis for $\eunn$ then the above proposition implies that the following potential is separable in generalized elliptic coordinates (see \cref{ex:genElipCoord}):
	\begin{equation}
	V = \sum_{i} k_i \bp{r,a_{i}}^{-2}
	\end{equation}
	
	\noindent for some $k_i \in \R$. In fact this potential is clearly multi-separable. Furthermore we can also obtain a multi-separable potential on $\eunn(\kappa)$ by restriction.
\end{remark}

Now returning to the Calogero-Moser system, we construct the most general solution to the KBD equation that one can construct using the above proposition:

\begin{propMy} \label{prop:KBDsolCM}
	If $V$ is the potential of the Calogero-Moser system given by \cref{eq:CMHam}, then the following CT is a solution of the KBD equation:
	
	\begin{equation} \label{eq:KBDsolCM}
	L = c \, d \odot d + 2 w \, d \odot r + m \, r \odot r
	\end{equation}
	
	\noindent where $c,w,m \in \R$. Furthermore the restriction of the above CT to $\Si^{n-1}$ is a solution of the spherical KBD equation.
\end{propMy}

\begin{proof}
	Consider the vectors $b_{i} := e_{1} - e_{i}$ for $i \neq 1$. We construct the most general CT for which each vector $b_i$ is an eigenvector of $A$ and orthogonal to $w$. Observe that none of them are orthogonal, they span an $n-1$ dimensional subspace and 
	\begin{equation}
	\cap_{i} b_{i}^{\perp} = (\oplus_{i} \spa{b_{i}})^{\perp} = \spa{d}
	\end{equation}
	
	Now suppose $A$ is a self-adjoint operator such that each $b_{i}$ is an eigenvector of $A$. Then it follows that $A$ must have $d^{\perp}$ as an eigenspace, hence $A = k I + c d \odot d$ for some $k,c \in \R$. Thus up to equivalence the above form of $L$ satisfies our requirements, and it follows by \cref{prop:extCTKBDsol} that $L$ satisfies the KBD equation with $V$.
	
	The second statement on the spherical KBD equation follows by a similar argument using \cref{prop:extCTKBDsol}.
\end{proof}
\begin{remark}
	It follows by a straightforward calculation that the CT stated in the above proposition is the most general solution of the KBD equation. Similarly when $n = 3$ one can check that the solution to the spherical KBD equation given in the above proposition is the most general.
\end{remark}

\paragraph{Canonical forms} We obtain the canonical forms according to \cref{thm:conTenCanForm} for the CTs given by \cref{eq:KBDsolCM}. First the constants $\omega_i$ from \cref{eq:omegI} are given as follows:

\begin{align}
\omega_0 & = m \\
\omega_1 & = w^2 
\end{align}

Note that in Euclidean space, one only needs to calculate $\omega_0$ and $\omega_1$ to carry out the classification. We now consider the cases given by \cref{thm:conTenCanForm}:

\begin{parts}
	\item \textbf{Elliptic}: $\omega_0 \neq 0$ \\
	By applying the translation given by \cref{eq:transConFormI} and changing to a geometrically equivalent CT one obtains:
	
	\begin{equation} \label{eq:CMCTGeoI}
	L = c d \odot d + r \odot r
	\end{equation}
	
	\noindent for some $c \in \R$.
	\item \textbf{Parabolic}: $\omega_0 = 0, \omega_1 \neq 0$ \\
	By applying the translation given by \cref{eq:transConFormII} and changing to a geometrically equivalent CT one obtains:
	
	\begin{equation} \label{eq:CMCTGeoII}
	L = 2 d \odot r
	\end{equation}
	
	\item \textbf{Cartesian}: $\omega_0 = 0, \omega_1 = 0, c \neq 0$ \\
	
	In this case after changing to a geometrically equivalent CT, we have:
	
	\begin{equation} \label{eq:CMCTGeoIII}
	L = d \odot d
	\end{equation}
\end{parts}

Hence the three geometrically inequivalent solutions of the KBD equation for the Calogero-Moser potential are given by \cref{eq:CMCTGeoI,eq:CMCTGeoII,eq:CMCTGeoIII}. Note that we can obtain these CTs from \cref{eq:KBDsolCM} with an appropriate choice of parameters, hence there is no need to apply any isometries.

\paragraph{Determining Separability} We analyze these solutions further to find separable coordinates. We will obtain a compete analysis for the case $n \leq 3$ for purposes of illustration. For the following analysis, we fix unit vectors $a \in d^\perp$ and $e \in \d^\perp \cap a^\perp$.

We define $N_1$ to be the unit sphere in $d^\perp$:

\begin{equation}
N_1 = \{ p \in d^\perp \: | \: p^{2} = 1 \}
\end{equation}

Note if $d^\perp = \R a$, then we take $N_1 = \{a\}$. When $\dim N_1 = 1$, we take coordinates on it as follows:

\begin{equation}
\sigma(\theta) = \cos(\theta) a + \sin(\theta) e
\end{equation}

\begin{parts}
	\item Elliptic with $c \neq 0$ \\ 
	When $n > 2$, this CT is reducible and a warped product decomposition $\psi$ which decomposes this CT is given by \cref{ex:wpRedI}. First define $N_0$ as follows:
	
	\begin{align}
	N_{0} & = \{p \in \R d \obot \R a \: | \: \bp{a, p} > 0  \}
	\end{align}
	
	For $(p_0, p) = ( x a + y d ,p) \in N_{0} \times N_{1}$, $\psi$ is given as follows (see \cref{ex:wpRedI}):
	
	\begin{equation}
	\psi(p_{0},p) = x p + y d
	\end{equation}
	
	Note that this equation also holds when $n = 2$, but in this case $\psi$ is not a warped product decomposition. To separate $V$, we have to apply the BEKM separation algorithm with $V$ restricted to $N_1$ on $N_1$. Although it will be more convenient to use the spherical KBD equation in $d^\perp$, see the next section for more details.
	
	When $n \leq 3$, no additional steps are needed since in this case $\dim N_1 \leq 1$. Indeed, by \cref{ex:genElipCoord} $L$ restricted to $N_0$ is an ICT (in a dense subset) hence $L$ has simple eigenfunctions (locally), and so one obtains separable coordinates for $V$ by taking elliptic coordinates on $N_0$ \cite{Rajaratnam2014a}. When $c < 0$ we obtain oblate spheroidal coordinates and when $c > 0$ we obtain prolate spheroidal coordinates; see \cref{ex:oblProCoords} for more details.
	
	\item Parabolic \\
	When $n > 2$, then proceeding as in \cref{ex:wpRedI} (see also \cref{par:CtRedConstrEn}), one observes that the same warped product $\psi$ as in the above case decomposes this CT. When $n \leq 3$, with similar arguments as in the above case, one finds that $L$ locally has simple eigenfunctions, and one obtains separable coordinates for $V$ by taking parabolic coordinates on $N_0$ \cite{Rajaratnam2014a}. The resulting coordinate system is often called rotationally symmetric parabolic coordinates.
	\item Spherical: Elliptic with $c = 0$ \\
	
	In this case, one can check that the following warped product, $\psi$, decomposes $L$. For $(p_0, p) = ( \rho a ,p) \in \R^+ a \times \Si^{n-1}$, $\psi$ is given as follows:
	
	\begin{equation}
	\psi(p_{0},p) = \rho p
	\end{equation}
	
	Now observe that even when $n = 3$, $L$ does not have simple eigenfunctions; in contrast with the previous two cases. To fill the multidimensional eigenspace of $L$ corresponding to $r^\perp$, we have to solve the spherical KBD equation (see the next section for more details). When $n = 3$, we can fill this degeneracy by using the solution to the spherical KBD equation given by \cref{prop:KBDsolCM}. Indeed, that proposition shows that the CT on $\Si^{n-1}$ induced by $d \odot d$ is a solution of the spherical KBD equation. Hence by \cref{ex:sphCoord}, this induced CT is diagonalized in spherical coordinates, and we see that $V$ separates in the following coordinates \cite{Rajaratnam2014a}.
	
	
	%
	
	%
	%
	
	\begin{equation}
	\psi(\rho a, p) = \rho (\sin(\phi) (\cos(\theta) a + \sin(\theta) e) + \cos(\phi) d)
	\end{equation}
	
	\item Cartesian \\
	In this case we obtain a product which decomposes $L$ as follows. First let $N_0 = \R d$ and $N_1 = d^\perp$, then for $(p_0, p) = ( x d ,p) \in N_{0} \times N_{1}$, we have:
	
	\begin{equation}
	\psi(p_0, p) = x d + p
	\end{equation}
	
	As in the above case, even when $n = 3$, $L$ does not have simple eigenfunctions. Hence we have to apply the BEKM separation algorithm with $V$ restricted to $N_1$ on $N_1$. When $n = 3$ one finds that the general solution to the KBD equation is $\tilde{r} \odot \tilde{r}$ where $\tilde{r}$ is the dilatational vector field in $N_1$. Thus if we take polar coordinates in $N_1$, we obtain separable coordinates for $V$. For $(p_0, p) = ( x d ,y \sigma(\theta)) \in N_{0} \times N_{1}$ with $y > 0$, we have:
	
	\begin{equation}
	\psi(p_0, y \sigma(\theta)) = x d + y (\cos(\theta) a + \sin(\theta) e)
	\end{equation}
\end{parts}

We conclude with some remarks. First the analysis given above is complete when $n \leq 3$. Although when $n > 3$ the warped product decompositions obtained may allow for partial separation of the Hamilton-Jacobi equation. When $n = 4$ it was shown in \cite{Rauch-Wojciechowski2005} that no additional solutions to the (spherical) KBD equation could be obtained. Hence our analysis above is complete when $n = 4$.

Furthermore the above analysis holds verbatim for the weighted Calogero-Moser system with unequal masses, which can be modeled using the natural Hamiltonian in $\E^n$ associated with the following potential (see e.g. \cite[Section~3.3]{Rauch-Wojciechowski2005}):

\begin{equation}
V = \sum_{1 \leq i < j \leq n} \frac{g_{ij}}{(m_i q_{i}- m_jq_{j})^{2}}
\end{equation}

The only difference is that in this case:

\begin{equation}
d = \frac{1}{\sqrt{M}} \sum_{i=1}^{n} \frac{e_{i}}{m_i},  \quad M = \sum_{i=1}^{n} \frac{1}{m_i^2}
\end{equation}

More examples can be found in \cite[section~7]{Waksjo2003}, where an algorithm equivalent to the BEKM separation algorithm is used to determine separability of some natural Hamiltonians defined in $\E^3$. See also \cite{Benenti1993} where some Kepler type potentials are tested for separability in elliptic coordinates in $\E^2$.

\subsubsection{Spherical KBD Equation} \label{par:sphKBDeq}

We first show how to derive the spherical KBD equation. Suppose $V \in \F(\eunn)$ is a potential in $\eunn$ which satisfies the KBD equation with $r \odot r$. Choose $a \in \eunn$ with $\kappa := a^2 = \pm 1$ and let $\rho := \bp{a, r}$. Then we can easily construct a warped product $\psi : \R^+ a \times_\rho \eunn(\kappa) \rightarrow \eunn$ which decomposes this CT. Let $\tau : \eunn(\kappa) \rightarrow \eunn$ be the standard embedding of this sphere. Hence to find separable coordinates for $V$, we have to apply the BEKM separation algorithm with $\tilde{V} := \tau^* V$ in $\eunn(\kappa)$.

If $\tilde{L}$ is the general CT in $\eunn(\kappa)$ and $\tilde{K} := \tr{\tilde{L}}R - \tilde{L}$ is the KBDT where $R$ is the metric in $\eunn(\kappa)$, then we have to solve the equation \cite[section~6.3]{Rajaratnam2014a}:

\begin{equation}
\d(\tilde{K} \d \tilde{V}) = 0
\end{equation}

Now let $K$ be the lift of $\tilde{K}$ (as a contravariant tensor) to $\eunn$ via the warped product $\psi$. Then proposition~5.2 in \cite{Rajaratnam2014a} shows that the above equation is locally satisfied iff

\begin{equation}
\d(K\d V) = 0
\end{equation}

Hence if we calculate this lift of $K$, we only need to solve the above equation in $\eunn$. We now proceed to calculate this lift. Note that it is sufficient to find a contravariant tensor, $K$, in $\eunn$ which is equal to $\tilde{K}$ for points in $\eunn(\kappa)$ and satisfies $\lied{K}{r} = 0$. We shall see that it will be sufficient to do this for the CT then calculate the KBDT using its defining equation. Also noting that $r$ is a \gls{cv}, we execute the following calculations in a more general context just using this fact. 

Let $r$ be a non-null CV, since $r \odot r$ is an OCT, it follows that any integral manifold of $r^\perp$ is a spherical submanifold. Hence \cref{prop:CtRestUmb} shows that any CT on $M$ induces one on any leaf of the foliation induced by $r^\perp$. The following proposition shows how to solve the problem described earlier in this more general context.

\begin{propMy} \label{prop:CTrestCVperp}
	Suppose $L$ is a CT on $M$ and $r$ is a non-null CV. Let $E := r^{\perp}$, and $L_{E} := L|_{E}$. Then $\tilde{L} := r^{2}L_{E}$ restricts to a CT on any integral manifold of $E$ and it satisfies $\lied{\tilde{L}}{r} = 0$ on $M$ where $\tilde{L}$ is in contravariant form.
\end{propMy}
\begin{proof}
	The proof of this fact is a straightforward calculation. We first note that since $r$ is a CV with conformal factor $\phi$, we have that
	
	\begin{equation}
	?\nabla_(i? ?r_j)? = \phi g_{ij}
	\end{equation}
	
	Suppose $u,v \in \Gamma(E)$, then
	
	\begin{align}
	(\lied{L_{ij}}{r}) u^i v^j & = (\nabla_{r}L_{ij}) u^i v^j+ L_{ij} (\nabla_{u} r^i) v^j + L_{ij} (\nabla_{v} r^i) u^j \\
	& = \alpha_{(i} r_{j)} u^i v^j + + 2 \phi L_{ij} u^i v^j \\
	& = 2 \phi L_{ij} u^i v^j
	\end{align}
	
	Thus
	
	\begin{align}
	(\lied{L^{ij}}{r}) u_i v_j & = \lied{(G^{ik}  L_{kj} G^{lj})}{r} u_i v_j \\
	& = -2 \phi L_{ij} u^i v^j + (\lied{L_{ij}}{r}) u^i v^j -2 \phi L_{ij} u^i v^j \\
	& = -2 \phi L_{ij} u^i v^j
	\end{align}
	
	Finally
	
	\begin{align}
	(\lied{(r^{2} L^{ij})}{r}) u_i v_j & = r^2 (\lied{L^{ij}}{r}) u_i v_j + (\nabla_{r} r^{2}) L^{ij} u_i v_j \\ 
	& = -2 r^2  \phi L^{ij} u_i v_j + 2 r^2  \phi L^{ij} u_i v_j \\
	& = 0
	\end{align}
	
	Thus since $r^{\flat}$ is closed, we conclude that $\lied{\tilde{L}}{r} = 0$. Also, as we noted earlier, \cref{prop:CtRestUmb} implies that $\tilde{L}$ induces a CT on any integral manifold of $E$.
\end{proof}
\begin{remark}
	The above ansatz for $\tilde{L}$ was deduced by studying results obtained by \citeauthor{Benenti2008} in \cite{Benenti2008}. Although one can also obtain $\tilde{L}$ by solving a certain differential equation.
\end{remark}

Returning to $\eunn$, let $r$ be the dilatational vector field and $L = r^2 L_E$ as in the above proposition. Note that $L_E$ is given in general by \cref{eq:CTGenEunnKap}. Let $G$ be the metric of $\eunn$, then $R = G_E$ is the induced metric on $\eunn(\frac{1}{r^2})$ and the above proposition shows that $\lied{(r^2 R)}{r} = 0$. Hence $r^2 R$ is the $r$-lift of the metric of $\eunn(\kappa)$ (up to sign). Hence if $\tr{L}$ is obtained by using the metric of $\eunn$, the lifted KBDT is given as follows:

\begin{equation} \label{eq:KBDTsphInEunn}
K_s = (\tr{L}\frac{1}{r^2}) (r^2 R) - L = \tr{L}R - L
\end{equation}

\noindent which is the KBDT in $\eunn(\kappa)$ embedded in $\eunn$. Also note that it follows from proposition~4.3 in \cite{Rajaratnam2014a} that $K_s$ is a KT in $\eunn$. Also, using \cref{eq:CTGenEunnKap}, one can calculate $K_s$ explicitly:

\begin{equation}
K_s = \tr{A} r^2 R - \bp{r,A r} G - r^{2} A + 2 A r \odot r
\end{equation}

Note that since the term $\tr{A} r^2 R$ is a multiple of the metric of $\eunn(\kappa)$, that term can be removed. We summarize our results in the following statement:

\begin{propMy}[Spherical KBD equation] \label{prop:KBDsph}
	Suppose $V \in \F(\eunn)$ is a potential in $\eunn$ which satisfies the KBD equation with $r \odot r$. Let $L$ be a CT in $\eunn(\kappa)$ with parameter matrix $A$. Then $V$ satisfies the KBD equation induced by $L$ in $\eunn(\kappa)$ iff it satisfies the spherical KBD equation (\cref{eq:KBDsph}) with $L$ in $\eunn$.
\end{propMy}

\subsection{In pseudo-Euclidean space}

We show how to execute the BEKM separation algorithm in pseudo-Euclidean space. Fix a non-trivial solution $L$ of the KBD equation in $\eunn$. First apply the classification given by \cref{thm:conTenCanForm} to $L$. We assume that $L$ is in one of the canonical forms listed in that theorem. If $L$ is a Cartesian CT then the analysis is straightforward, see \cref{ex:CalMosSys} for example. So we now assume $L$ is non-degenerate and each generalized eigenspace of $A_c$ has at most one proper generalized eigenvector\footnote{It was proven that we lose no generality with this assumption in Euclidean or Minkowski space.}.

First if $A_c$ has no multidimensional (real) eigenspaces, then it is not reducible by \cref{thm:classRedCtEunn}. Hence one obtains separable coordinates for the natural Hamiltonian on the subset where $L$ is an ICT.

Now suppose $A_c$ has multidimensional (real) eigenspaces $W_{1},\dotsc,W_{k}$. It is shown in \cref{par:CtRedConstrEn} that one can obtain data $(\bar{p};\bigobot\limits_{i=0}^{k} V_{i}; a_{1},...,a_{k})$ which determines a warped product decomposition $\psi : N_0 \times_{\rho_1} N_1 \cdots \times_{\rho_k} N_k \rightarrow \eunn$ in canonical form. Note that $\psi$ decomposes the KBDT, $K$, associated with $L$. We now work with $K$.

We consider a somewhat more general situation in order to incorporate the spherical case later. Suppose $K$ is an orthogonal KT in $\eunn$ which is decomposed by the warped product $\psi$ just constructed. Furthermore assume that each $N_i$ corresponds to a distinct eigenspace of $K$. Now we show how to apply the BEKM separation algorithm on the spheres $N_i$ by working only in a pseudo-Euclidean space.

\begin{parts}
	\item $N_i$ is a non-null sphere, i.e. $a_i^2 \neq 0$ \\
	Let $W_{i\perp} := W_i^\perp$ and $c_i := \bar{p} - \frac{a_i}{\kappa_i}$. Define $\phi : W_{i\perp} \times W_i \rightarrow \eunn$ to be the standard product decomposition. Embed $W_i$ in $\eunn$ as follows:
	
	\begin{equation}
	\tau_i : \begin{cases}
	W_i & \rightarrow \eunn \\
	p_i & \mapsto \phi(c_i,p_i) = c_i + p_i
	\end{cases}
	\end{equation}
	
	Note that $N_i = \tau_i(W_i(\kappa_i))$. Let $r_i$ be the dilatational vector field in $W_i$. By \cref{cor:WPsphMet} and Proposition~5.2 in \cite{Rajaratnam2014a}, it follows that $\tau_i^* V$ satisfies the KBD equation with $r_i \odot r_i$. Hence by \cref{prop:KBDsph} it is necessary and sufficient to solve the spherical KBD equation on $W_i$ with $\tau_i^* V$.
	\item $N_i$ is a null sphere, i.e. $a_i^2 = 0$ \\
	Embed $N_i$ in $\eunn$ as follows (see \cref{eq:genPsiEqn}):
	
	\begin{equation}
	\tau_i : \begin{cases}
	N_i & \rightarrow \eunn \\
	p_i & \mapsto \psi(\bar{p},\dotsc,\bar{p},p_i,\bar{p},\dotsc,\bar{p}) = p_i
	\end{cases}
	\end{equation}
	In this case $N_i$ is isometric to $V_i$ which is a pseudo-Euclidean space. Hence the BEKM separation algorithm can be applied on $V_i$.
\end{parts}

In the following section we show how to apply the BEKM separation algorithm on $\eunn(\kappa)$.

\subsubsection{In Spherical submanifolds of pseudo-Euclidean space}

We show how to execute the BEKM separation algorithm in $\eunn(\kappa)$. First we convert it to a problem in $\eunn$. Let $\tilde{V}$ be a potential in $\eunn(\kappa)$. Note that $\tilde{V}$ can be naturally lifted to a potential in $\eunn$ satisfying $\lied{\tilde{V}}{r} = 0$ using an appropriate coordinate system. Then, one can check that the potential 

\begin{equation}
V := \frac{\tilde{V}}{\kappa r^2}
\end{equation}

\noindent in $\eunn$ satisfies the KBD equation with $r \odot r$ in $\eunn$ and equals $\tilde{V}$ for points in $\eunn(\kappa)$. So we lose no generality in working with a potential $V \in \F(\eunn)$ which satisfies the KBD equation with $r \odot r$.

Note that by \cref{prop:KBDsph}, we only need to consider solutions of the spherical KBD equation in $\eunn$. So let $L$ be a non-trivial solution of the spherical KBD equation (\cref{eq:KBDsph}). As in the pseudo-Euclidean case, we assume each generalized eigenspace of $A$ has at most one proper generalized eigenvector. In order to execute the BEKM separation algorithm in $\eunn$, we will need the following lemma:

\begin{lemma}
	Let $L_c$ be the central CT associated with $L$ and $K_s = \tr{L} R - L$ be the KBDT associated with $L$. Suppose $L_c$ is reducible and let $\psi : N_0 \times_{\rho_1} N_1 \cdots \times_{\rho_k} N_k \rightarrow \eunn$ be a warped product which decomposes $L_c$. Then $\psi$ decomposes $K_s$.
\end{lemma}
\begin{proof}
	This follows from the proof of \cref{thm:eunnKapRestCT}. In that proof we obtained the following equation:
	
	\begin{align}
	R^{*}L_c & = \psi_{*} (\tilde{R}^{*}\tilde{L}_c + \sum_{i=1}^{k} \lambda_{i} G_{i})
	\end{align}
	
	Then we have:
	
	\begin{align}
	L & = r^2 R^{*}L_c  = \psi_{*} (\tilde{r}^2\tilde{R}^{*}\tilde{L}_c + \sum_{i=1}^{k} \lambda_{i} \tilde{r}^2 G_{i}) \\
	R & = \psi_* (\tilde{R} + \sum_{i=1}^{k} G_{i})
	\end{align}
	
	Hence the result follows.
\end{proof}

Now by \cref{thm:eunnKapRestCT} it follows that $L$ is reducible iff $L_c$ is reducible. Hence if $L_c$ is not reducible, one obtains separable coordinates for the natural Hamiltonian on the subset (of $\eunn(\kappa)$) where $L$ is an ICT.

If $L_c$ is reducible, then by the above lemma, one can follow the arguments given in the previous section using the warped product decomposition induced by $L_c$ which decomposes the KT $K_s$.

We now make some crucial remarks. Let $\psi : N_0 \times_{\rho_1} N_1 \cdots \times_{\rho_k} N_k \rightarrow \eunn$ be a warped product decomposition which decomposes $L_c$ and let $\phi : N_0(\kappa) \times_{\rho_1} N_1 \cdots \times_{\rho_k} N_k \rightarrow \eunn(\kappa)$ be an induced warped product decomposition of $\eunn(\kappa)$ as in \cref{thm:eunnKapRestWP}. First note that the separable coordinates are constructed using the warped product $\phi$. Also because the spherical factors $N_i$ (where $i > 0$) are simultaneously spherical factors of $\psi$ and $\phi$ (see \cref{thm:eunnKapRestWP}), we can work in the ambient space.

\section{Conclusion}

In this article we have given a classification of concircular tensors in spaces of constant curvature which permits us to apply them to the separation of variables problem as suggested in \cite{Rajaratnam2014a}. We have obtained canonical forms for these tensors modulo the action of the isometry group in \cref{sec:CtEunn,sec:CtEunnKap}, studied the webs described by irreducible concircular tensors in \cref{sec:cTIrred} and obtained warped product decompositions adapted to reducible orthogonal concircular tensors in \cref{sec:CtClassRed}. In \cref{sec:appNEx} we have shown how to apply these results to solve some of the motivating problems listed in the introduction.

In our solution, there is one important problem that has been unresolved. In Minkowski space, $M^n$, with $n \geq 3$, it is still computationally difficult to find the subset on which a given concircular tensor (CT) is a Benenti tensor. This implies that we still don't have a complete understanding of the separable coordinate systems for these spaces. However, when the space has Euclidean signature or $n = 2$, this is not a problem as is illustrated by \cref{ex:genElipCoord,ex:CCTCoordMink} respectively.

For future research, it would be interesting to see if concircular tensors can be applied to other types of separation such as non-orthogonal separation \cite{Benenti1992,Benenti1997a,Kalnins1979}, complex separation \cite{Degiovanni2007}, and conformal separation \cite{Benenti2005b}. Note that the first two types of separation are of no interest in Euclidean space but they are in Minkowski space. In \cite{Bolsinov2013}, a procedure is given to obtain the local canonical (normal) forms for CTs in pseudo-Riemannian manifolds. Hence the results developed therein may be of interest for the study of the first two types of separation.

\section*{Acknowledgments}

We would like to express our appreciation to Dong Eui Chang for his continued interest in this work.  The research was supported in part by National Science and Engineering Research Council of Canada Discovery Grants (D.E.C. and R.G.M.). The first author would like to thank Spiro Karigiannis for critically reading his thesis \cite{Rajaratnam2014}, which contains the contents of this article.

%% file: app1.tex
\section*{Appendices}
\addcontentsline{toc}{section}{Appendices}

\section{Lexicographic ordering of complex numbers}

%
%
%

Complex numbers can be given a natural lexicographic ordering (as in dictionaries) by using their Cartesian product structure:

\begin{definition} \label{defn:ordC}
	Suppose $\lambda = a + i b$ and $\omega = c + i d$ are complex numbers. We write $\lambda < \omega$ if: $b < d$ or ($b = d$ and $a < c$)
\end{definition}

In the following we use ``xor'' to mean exclusive or and ``or'' has its standard meaning. Suppose $\lambda, \omega, \nu \in \C$ and $a \in \R^+$, one can check that this ordering has the following properties:

\begin{description}
	\item[trichotomy: ] $\lambda = \omega$ xor $\lambda < \omega$ xor $\omega < \lambda$
	\item[transitivity: ] If $\lambda < \omega$ and $ \omega < \nu$ then $\lambda < \nu$
	\item[translation invariance: ] If $\lambda < \omega$ then $\lambda + \nu < \omega + \nu$
	\item[dilatation invariance: ] If $\lambda < \omega$ then $a \lambda  < a \omega $
	\item[skew symmetry: ] If $\lambda < \omega$ then $- \omega < - \lambda$
\end{description}

Furthermore we note that if $\lambda, \omega \in \R$ then this ordering reduces to the natural ordering of real numbers.